%% file: Arxiv_main.tex
\documentclass[10pt]{article}

\usepackage[a4paper]{geometry}
\usepackage{amsthm}

\RequirePackage[OT1]{fontenc}
\RequirePackage[numbers]{natbib}
\RequirePackage[colorlinks,citecolor=blue,urlcolor=blue]{hyperref}
\usepackage{amssymb}
\usepackage{enumerate}
\usepackage{relsize}
\usepackage{color} 
\usepackage{listings}
\usepackage{verbatim}
\usepackage{algorithm}
\usepackage[noend]{algpseudocode}
\usepackage{todonotes}      
\usepackage[flushleft]{threeparttable}
\usepackage{comment}

\usepackage{amsmath}
\usepackage{times}
\usepackage{bm}
\usepackage{dsfont,amsfonts}

\usepackage{enumerate}
\usepackage{mathscinet}
\usepackage[capitalize,noabbrev]{cleveref}
\usepackage{csquotes}

\usepackage{graphicx}
\usepackage[margin=1cm, textfont=small]{caption}

\usepackage{url}

\usepackage{booktabs,caption}
\usepackage[flushleft]{threeparttable}
\usepackage{setspace}

\allowdisplaybreaks

\DeclareMathOperator*{\argmax}{arg\,max}
\DeclareMathOperator*{\argmin}{arg\,min}

\DeclareMathOperator{\Tr}{Tr}
\newcommand{\I}{\mathcal I}
\newcommand{\dif}{\mathrm{d}}
\newcommand{\cmp}{\mathcal{C}}

\renewcommand{\vec}[1]{{\boldsymbol#1}} 
\newcommand{\mat}[1]{{\boldsymbol#1}} 

\newcommand{\csm}{\mathrm{CS}} 
\newcommand{\cpt}{\tau} 
\newcommand{\ncp}{\kappa} 
\newcommand{\N}{\mathbb{N}} 
\newcommand{\R}{\mathbb{R}} 
\newcommand{\eps}{\varepsilon}
\newcommand{\PP}{\mathcal{P}} 
\newcommand{\OS}{\mathrm{OS}} 
\newcommand{\aOS}{\mathrm{aOS}} 
\newcommand{\E}[1]{{\mathbb E}\left[ #1\right]} 
\newcommand{\floor}[1]{\lfloor #1 \rfloor}
\newcommand{\ceil}[1]{\lceil #1 \rceil}
\newcommand{\abs}[1]{\left| #1 \right|}
\newcommand{\norm}[1]{\left\| #1 \right\|}
\newcommand{\snorm}[1]{\|#1\|} 
\newcommand{\Prob}[1]{\mathbb{P}\left\{ #1 \right\}}
\newcommand{\var}[1]{\mathrm{Var}\left( #1 \right)}
\newcommand{\ind}{\mathds{1}} 
\newcommand{\iidsim}{\stackrel{\mathrm{i.i.d.}}{\sim}}
\newcommand{\thd}{\alpha} 

\newcommand{\gauss}{\mathcal N} 
\newcommand{\ber}{\mathrm{Ber}} 

\newcommand{\set}[1]{\left\{#1\right\}}

\newcommand{\inner}[2]{\left\langle#1,#2\right\rangle} 

\newtheorem{definition}{Definition}[section]
\newtheorem{lemma}[definition]{Lemma}
\newtheorem{theorem}[definition]{Theorem}

\newtheorem{prop}[definition]{Proposition}
\newtheorem{example}[definition]{Example}

\newtheorem{Model}{Model}

\theoremstyle{definition}

\newtheorem{remark}{Remark}

\DeclareRobustCommand*{\rev}{}

\graphicspath{{figures/}}

\begin{document}
\title{Optimistic search: Change point estimation for large-scale data via adaptive logarithmic queries}

\author{Solt Kov\'acs${}^{1}$\footnote{The first two authors contributed equally to this work.}\setcounter{footnote}{0}, Housen Li${}^{2}$\footnotemark,  Lorenz Haubner${}^{1}$, Axel Munk${}^{2, 3}$, Peter B\"uhlmann${}^{1}$\\
\vspace{0.1cm}\\
{\small${}^{1}$Seminar for Statistics, ETH Zurich, Switzerland}\\
{\small${}^{2}$Institute for Mathematical Stochastics, University of G\"ottingen, Germany}\\
{\small${}^{3}$Max Planck Institute for Biophysical Chemistry, G\"ottingen, Germany}}

\date{November 2022}

\maketitle

\begin{abstract}
\input{sections/0Abstract}
\end{abstract}


\input{sections/1Introduction.tex}
\input{sections/3Method.tex}

\input{sections/4Method_multiple.tex}

\input{sections/4bMultivariateHD.tex}

\input{sections/5Simulations.tex}
\input{sections/6Discussion.tex}

\section*{Acknowledgements}
Solt Kov\'acs and Peter B\"uhlmann have received funding from the European Research Council (ERC) under the European Union's Horizon 2020 research and innovation programme (Grant agreement No. 786461 CausalStats - ERC-2017-ADG). Axel Munk and Housen Li are funded by the Deutsche Forschungsgemeinschaft (DFG, German Research Foundation) under Germany’s Excellence Strategy - EXC 2067/1- 390729940. Axel Munk is also funded by DFG - FOR 5381. {\rev The authors thank Alexandre M\"{o}sching, editors and referees for helpful and careful comments.}


\clearpage
\appendix
\renewcommand\thefigure{\thesection\arabic{figure}}    
\setcounter{figure}{0}
\renewcommand\thetable{\thesection\arabic{table}}    
\setcounter{table}{0} 
\renewcommand\theequation{A\arabic{equation}}   
\setcounter{equation}{0}    
\input{sections/7AppendixandProofs.tex}

{\rev
\input{sections/proofs2}
}

\end{document}

%% file: sections/0Abstract.tex
Change point {\rev estimation} is often formulated as a search for the maximum of a gain function describing improved fits when segmenting the data. Searching through all candidates requires $O(n)$ evaluations of the gain function for an interval with $n$ observations. If each evaluation is computationally demanding (e.g.~in high-dimensional models), this can become infeasible. Instead, we propose \emph{optimistic search} methods with $O(\log n)$ evaluations exploiting specific structure of the gain function.

{Towards solid understanding of our strategy, we investigate in detail the $p$-dimensional Gaussian changing means setup,  including high-dimensional scenarios. For some of our proposals, we prove asymptotic minimax optimality for detecting change points and derive their asymptotic localization rate. These rates (up to a possible log factor) are optimal for the univariate and multivariate scenarios, and are by far the fastest in the literature under the weakest possible detection condition on the signal-to-noise ratio in the high-dimensional scenario. Computationally, our proposed methodology has the worst case complexity of $O(np)$, which can be improved to be sublinear in $n$ if some a-priori knowledge on the length of {the} shortest segment is available.} 

Our search strategies generalize far beyond the theoretically analyzed setup. We illustrate, as an example, massive computational speedup in change point detection for high-dimensional Gaussian graphical models. 

%% file: sections/1Introduction.tex
\section{Introduction}
\label{Introduction}
Change point (or break point)
estimation
tackles the problem of estimating the locations of abrupt structural changes for ordered {and noisy} data, {e.g.,}
by time or space.
One can distinguish between online (sequential) and offline (retrospective) detection problems. We will primarily focus on the latter setup where all ordered observations are available {at once} and only point to online detection in connection with our methods and results for the detection of a single change point. 
Applications include detecting changes in copy number variation \citep{CBS,ZhSi07}, ion channels \citep{Hotz_ionchannel}, financial time series \citep{BaPe98,Kim_finance,DaHo12}, climate data \citep{Reeves_climatology}, environmental monitoring systems \citep{Londschien}, among many others.
{We refer} also 
to the recent reviews in \cite{review_Niu,review_Truong}. 

We focus on computational 
{improvements} 
of change point {inference} {while maintaining statistical optimality}. Two common algorithmic approaches are optimal partitions via dynamic programming (cf.\ \cite{jackson2005algorithm,FKLW08}) and greedy procedures, e.g.~binary segmentation (BS, \cite{Vostrikova}) and its variants. The former approaches are mainly investigated for univariate data. 
Typical examples include $\ell_0$ penalization methods (cf.\ \cite{BKLMW09}; PELT \cite{Killick_etal} and FPOP \cite{Maidstone}) and multiscale methods (SMUCE \cite{Frick_etal} and FDRSeg \cite{LMS16}), which are known to be statistically minimax optimal {in a Gaussian setup}. However, finding the optimal partition requires in the worst case at least quadratic run time. In contrast, BS is typically faster and easier to adapt to more general scenarios, but worse in terms of estimation performance than methods finding the optimal partitioning. The wild binary segmentation (WBS, \cite{Fryzlewicz_WBS}) and its variants (e.g.\ narrowest over threshold, \cite{Baranowski}) improve on estimation performance of BS, but lose some of its computational efficiency. The recently proposed  seeded binary segmentation (SeedBS, \cite{SeedBS}) combines the best of both worlds, i.e.~improved estimation performance of BS, while keeping computational efficiency {(with log-linear run time in the worst case).}

Already for simple univariate cases, increasingly larger data sets with long time series 
led to the development of more efficient (univariate) approaches \cite{Maidstone, intelligent_sampling, Fryzlewicz_WBS2, SeedBS}. Computational issues are even much more pronounced for multivariate problems. This is in particular true for 
emerging high-dimensional {parametric} change point detection approaches (e.g.~\cite{LeonBuhl, RoyMichailidis, GibberdNelson, Gibberd_Roy, Bybee_Atchade_JMLR, Avanesov_theory, Wang_highdim_mean_change, Wang_Willet_regression, Londschien,YuCh21, Wang_Yu_Rinaldo_theory,WaYR21,DePY22}, 
or in {univariate \citep{vanegas2020multiscale} and}
multivariate \citep{Padilla_Yu_Wang_Rinaldo_multivariate_nonparametric} non-parametric change point detection methods. 
Many of these approaches rely on {computationally costly fits of algorithms such as the}
lasso \citep{lasso_original} or the graphical lasso \citep{glasso}.
Performing a full grid search in order to find a single split point requires as many fits as there are observations. Even with warm-starts, neighbouring fits with one additional observation are not straightforward to update (unlike e.g.~means in univariate cases) and the number of fits is the main driver of computational cost. Thus, for a few hundred or thousand observations, full grid search based methods (including BS and its variants and even more costly dynamic programming based approaches) can be very slow, beyond what is acceptable. This is a main motivation for our work here, namely to avoid too many fits with piecewise stationary model structure.

\subsection{Our contribution and related work}
The key idea is to replace the exhaustive search for a change point by an \emph{adaptive search} that dynamically determines the next search location given the previous ones {in the spirit of divide-and-conquer. We call this novel methodology for searching for single change points  \emph{optimistic search} (OS).}  We will show that {OS requires} only a logarithmic  number of evaluations, and is meanwhile statistically optimal in {various} scenarios. This is possible mainly because of a rather general observation. In numerous change point detection problems, population gain functions (describing the expected gain when splitting into two parts at a given split point) have a specific piecewise quasiconvex structure {where}
local maxima correspond to change points, and thus the presence of a single change point implies only one global maximum. 

{Based on OS,} in the single change point case, we propose for the first time, to our best knowledge, an algorithm (advanced optimistic search, aOS) with $O(\log n)$ function evaluations leading to {asymptotically minimax optimal detection and nearly minimax optimal localization of the change point} as opposed to full grid search 
methods requiring at least $O(n)$ evaluations for {univariate and multivaiate Gaussian changing mean problems. In high-dimensional case, the detection rate is also asymptotically {minimax} optimal, and the localization rates are (ignoring log-factors) sharper than the ones reported in the literature so far (unless a much stronger signal-to-noise ratio is assumed), but the minimax optimality of localization rate remains open (see \cref{ss:stgr}). These results are} surprising and fundamental, as essentially one loses nothing in terms of detectable signals and in terms of localization error (in an asymptotic sense) compared to full grid search. 
{In case of multiple change points, we develop a method (Optimistic Seeded Binary Segmentation, OSeedBS) on the combination of OS and seeded intervals \citep{SeedBS}, which achieves similar computational speedup and statistical optimality as in the single change point scenario.} In {$p$-dimensional} Gaussian changing mean problems, our methodology is significantly faster, {e.g., finding a change point in $O(p\log n)$ run time,}  if cumulative sums have been pre-computed, a reasonable assumption e.g.~in online change point detection. This also suggests to store data in cumulative sums format for offline analysis. {Even when the data are not properly stored, the worst case computation complexity {of our method (OSeedBS)} is $O(np)$, which is limited by the cost of calculating the cumulative sums. In particular, for the univariate setup with multiple change points, it improves the worst case computation complexity of $O(n\log n)$ that is 
the fastest in the literature, by a log factor. Further, the improvement becomes way more significant} in more complex models where the number of model fits involved in each evaluation of the gain function are the main driver of computational cost. A motivating example is given in~Section~\ref{motivating_example}.

Some problem specific solutions for gaining computational speedup in certain high-dimensional setups arose, e.g. for changing (Gaussian) graphical models \citep{Hallac_network, Bybee_Atchade_JMLR} or changing linear regression coefficients \cite{Kaul_2step, Kaul_multiple}, but these may not be easy to adapt to other scenarios. The following two proposals may be seen as most related to our approach. Our idea of avoiding a high number of model fits is vaguely related to the procedure of \cite{Kaul_2step} for change point detection in high-dimensional linear regression. For an initial split point they fit appropriate models which are then kept for evaluating (an approximation of) the gain function for candidate split points on the full grid, but without updating the costly high-dimensional fits. 
However, in scenarios where the evaluation of the gain on the full grid has a comparable cost as the model fit itself (e.g.~calculating and updating means in the Gaussian change in mean setup), their procedure would not lead to any speedup. In \cite{intelligent_sampling}, the authors proposed a procedure specific for the univariate case. They thin out the number of evaluations by searching on a small subsample to obtain preliminary estimates, and then use a dense sample in neighbourhoods for the final change point estimates. A rough subsample is not well suited for high-dimensional problems, unlike our approach of keeping all samples but avoiding evaluations on the full grid.

\subsection{A motivating example}
\label{motivating_example}

Rather than a problem specific solution, the proposed OS is a computationally attractive and general methodology for searching for a single change point. \cref{Fig:glasso_single_cpt} illustrates its computational efficiency in an example of single change point detection for a $200$-dimensional Gaussian graphical model (based on an estimator discussed in Section~\ref{highdim_simulations} and Appendix~\ref{Appendix_highdim}). The aim is to find the maximum of the gain function (black curve). {Evaluating the gain function at a single split point $t \in\{ 1, \ldots, n\}$ with $n=\text{2,000}$ requires two graphical lasso fits: one for the segment $(1, t]$ and one for $(t, n]$. For finding the maximum, the full grid search
took roughly $100$ times longer than
OS. The latter evaluated the gain only at two initial points (marked by two zeros) and subsequently at further $14$ split points (marked by colored numbers) which are determined dynamically. Then OS returns the maximum over all considered split points.  This leads to a massive reduction of computation time.}
When searching for multiple change points, OS can be combined flexibly with existing algorithms as discussed in 
\cref{Methodolgy_multiple}, resulting again in massive computational speedup with essentially no loss in statistical performance.

\begin{figure}
\centering
\includegraphics[width=0.82\textwidth]
{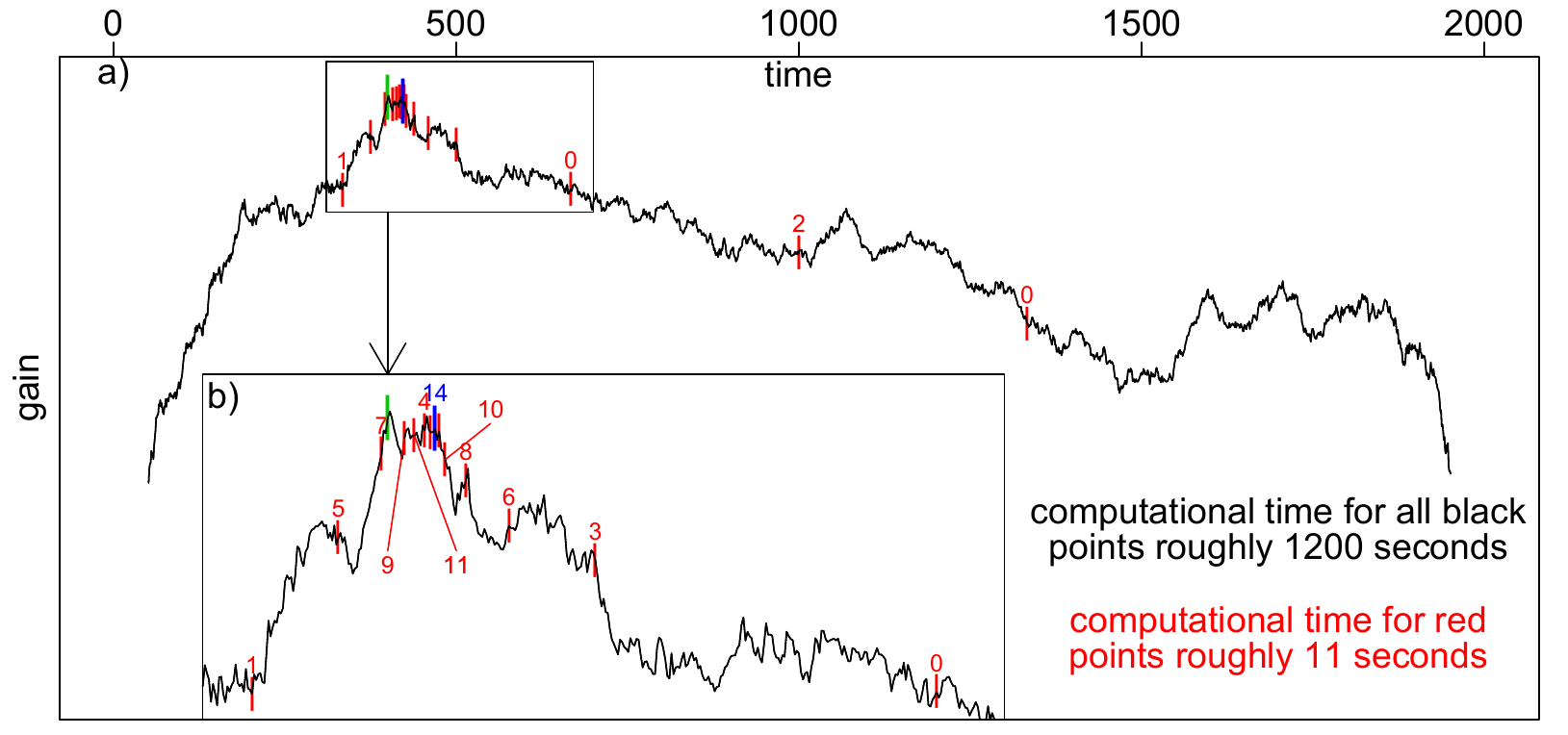} \\
\caption{Finding a single change point with full grid search (black) and OS (red) in a $200\times 200$-dimensional covariance change example with underlying graphical lasso fits. OS starts with two initial evaluations marked by the two zeros and then evaluates further $14$ split points adaptively, in the order marked by the respective colored numbers shown in the zoomed in part~b). The true underlying change point at observation~$400$ is marked in green and the final candidate returned by OS at observation $423$ in blue. The overall maximum of the black gain curve found by full grid search is at observation $402$.} 
\label{Fig:glasso_single_cpt}
\end{figure}

\subsection{Outline and announcing our results}

The crucial part of our methodology is the OS introduced in \cref{Methodology_single}. It is capable of finding a local maximum of the population gain function with $O(\log n)$ evaluations for $n$ observations. {For the sake of readability, we first introduce our methodology as well as the statistical guarantees in the} classical \emph{univariate} Gaussian change in mean setup, detailed in Section~\ref{s:gauss}. In the scenario when only one change point is present, 
the advanced version of OS, i.e.\ aOS, (Theorem~\ref{th:aoss}) is able to detect the change point under the weakest condition 
\begin{equation}\label{eq:opt}
\text{minimal jump size} \,\times\, \sqrt{\text{minimal segment length}}\, \gtrsim \,\sqrt{\log\log(n)/n} 
\end{equation}
thus being optimal in minimax sense. 
In Section~\ref{Methodolgy_multiple} we extend the methodology for multiple change points {and derive a}
minimax optimal performance result. Namely, under nearly the same condition as in \eqref{eq:opt}, the number of change points is identified correctly, and 
the location of each change point is estimated at the best rate {(up to possible log-factors)} that is available in literature, see Theorem~\ref{th:osbs}.

{We further examine general \emph{multivariate} and \emph{high-dimensional} Gaussian mean changes in \cref{s:mhdgauss}. We again obtain {via aOS} the detection and the localization rates of change points, which are (nearly) the best in the literature. Interestingly, unlike the univariate and multivarite cases, the localization rate may be much faster than the rate induced by the detection problem in several high-dimensional setups, including both sparse and dense scenarios, for instance, when the signal-to-noise ratio is much larger than the minimax detection rate. 
}

{\cref{Simulations,Discussion} contain empirical simulation results and conclusions.}
Additional material as well as proofs are given in the Appendix. {There we present several deviation inequalities for randomly weighted chi-squares which are of independent interest.}

\paragraph{Notation} For a real number $r$, we define downward rounding as $\lfloor r \rfloor = \max\{n\in \mathbb{Z} \,:\, n \le r\}$ and upward rounding as $\lceil r \rceil = \min\{n \in \mathbb{Z}\,:\,n \ge r\}$, {and also define $(r)_+ := \max\{r, 0\}$.} For two sequences of positive real numbers $\{a_n\}_{n=1}^\infty$ and $\{b_n\}_{n=1}^\infty$, we write $a_n \ll b_n$, or $a_n = o(b_n)$, if $\lim_{n \to \infty} a_n / b_n = 0$, write $a_n \lesssim b_n$, or $a_n = O(b_n)$, if $\limsup_{n \to \infty} a_n / b_n < \infty$, and write $a_n \asymp b_n$ if both $a_n \lesssim b_n$ and $b_n \lesssim a_n$\,. {We use bold symbols for vectors and matrices to differentiate them from scalar values. }

\section{Gaussian mean shifts with constant variance}\label{s:gauss}
We start {our development by considering the simple} model of univariate Gaussian changing means (Model~\ref{Gaussian_setup}) {below}. 
{The understanding for this setup is
generalized to multivariate and high-dimensional 
scenarios 
in  \cref{s:mhdgauss}.}

\begin{Model}[Univariate Gaussian changing means]
\label{Gaussian_setup}
Assume that observations $X_1, \ldots, X_n$ are independent and 
\begin{equation*}
\begin{split}
 X_{\cpt_{0}n+1}(=X_1), \ldots, X_{\cpt_1 n} & \sim \mathcal{N}(\mu_0, \sigma^2)\,, \\
& \vdots	\\
 X_{\cpt_{\ncp}n+1}, \ldots, X_{\cpt_{\ncp+1}n}(=X_{n}) 	& \sim \mathcal{N}(\mu_{\ncp}, \sigma^2)\,,
\end{split}
\end{equation*}
where $\{\cpt_i \,:\, i = 1, \ldots, \ncp\}$ gives the location of change points satisfying
$$
0 = \cpt_0 < \tau_1 < \cdots < \cpt_{\ncp+1} = 1\quad \text{and}\quad \cpt_i n \in \N\,,
$$
means $\mu_i \neq \mu_{i-1}$ for $i=1,\ldots,\ncp$ give the levels on segments, and the common standard deviation $\sigma > 0$ is known.
Assume w.l.o.g.~$\sigma = 1$.
Moreover, define the minimal segment length $\lambda$ as
$$
\lambda \equiv \lambda_n = \min_{i = 0,\ldots,\ncp} (\cpt_{i+1} - \cpt_i),
$$
and the minimal jump size $\delta$ as
$$
\delta \equiv \delta_n = \min_{i = 1,\ldots, \ncp} \delta_i \qquad \text{with} \quad \delta_i = \abs{\mu_i - \mu_{i-1}}.
$$
\end{Model}

The goal of change point {inference} is to estimate the number $\ncp$ and the locations $\cpt_i$'s of the true underlying change points from realizations {of} $X_1, \ldots, X_{n}\,$. A common criterion for determining the best split point is the CUSUM statistics \citep{Page54}, defined for an interval $(l,r]$ and a split point $t$ as 
\begin{equation}
\label{CUSUM}
    \csm_{(l,r]}(t) = \sqrt{\frac{r-t}{(r-l)(t-l)}}\sum_{i=l+1}^t X_i - 
        \sqrt{\frac{t-l}{(r-l)(r-t)}}\sum_{i=t+1}^r X_i\,, 
\end{equation}
with integers $0 \leq l < t < r \leq n$. The CUSUM statistic is the likelihood ratio test for a single change point at location $t$ in the interval $(l,r]$ against a constant signal. The population counterpart of $|\csm_{(l,r]}(\cdot)|$, i.e.\ replacing $X_i$ by $\E{X_i}$, 
has its maximum at one of the underlying change points. In noisy cases, the best split point candidate when dividing the segment $(l,r]$ into two parts is the location of the maximal absolute CUSUM statistics
$$    \hat t_{(l,r]} = \argmax_{t\in{\{l+1, \dots, r-1\}}} |\csm_{(l,r]}(t)|\,.
$$
We refer to the function $|\csm_{(l,r]}(\cdot)|$ as a gain function, denoted by $G_{(l,r]}(\cdot)$, because the square of it describes gains, namely the reductions in squared errors when fitting separate means on the left and right segments for split points in the segment $(\ell,r]$. 
The gain functions $G_{(l,r]}(\cdot)$ are initially defined on a discrete grid of split points, but {it is convinient to extend them continuously to $t\in (l, r]$}
(via e.g.~linear interpolation).

%% file: sections/3Method.tex
\section{Optimistic search for a single change point}
\label{Methodology_single}

In this section, we focus on \cref{Gaussian_setup} with a single change point ($\ncp = 1$), and introduce two versions of optimistic search, which are at the core of our methodology for multiple change points, as well.

\subsection{Naive optimistic search}

\begin{algorithm}[t]
\caption{Naive Optimistic Search (OS)}\label{Alg:OS}
\begin{algorithmic}[1]
\Require $r - l > 2, ~l,r \in \N$; and step size $\nu \in (0, 1)$ with $ 1/2$ by default \\
\textbf{initialize:} $\tilde l\leftarrow l$, $\tilde r\leftarrow r$ and $t \leftarrow  \floor{(l+\nu r)/(1+\nu)}$
\Function{\textnormal{OS}}{$\tilde l, t, \tilde r \mid \nu, l, r$}
\If {$\tilde r-\tilde l \leq 5$} \Comment{Stopping condition for recursion}
\State $\hat t_{(l,r]} \leftarrow \argmax\limits_{t\in\{\tilde l+1, \dots,\tilde r-1\}} G_{(l,r]}(t)$ 
	\Comment{Search over all points if less than $5$ remain}
\State \textbf{return} $\hat t_{(l,r]}$
\EndIf
\If {$\tilde r-t>t-\tilde l$} \Comment{Pick a new probe point in larger segment}
\State $w \leftarrow \lceil \tilde r - (\tilde r - t)\nu \rceil$
\If {$G_{(l,r]}(w) \geq G_{(l,r]}(t)$}
\State \Call{\textnormal{OS}}{$t, w, \tilde r \mid \nu, l, r$}
\Else
\State \Call{\textnormal{OS}}{$\tilde l, t, w \mid \nu, l, r$}
\EndIf
\Else
\State $w \leftarrow \lfloor \tilde l + (t-\tilde l)\nu \rfloor$
\If {$G_{(l,r]}(w) \geq G_{(l,r]}(t)$}
\State \Call{\textnormal{OS}}{$\tilde l, w, t \mid \nu, l, r$}
\Else
\State \Call{\textnormal{OS}}{$w, t, \tilde r \mid \nu, l, r$}
\EndIf
\EndIf
\EndFunction
\end{algorithmic}
\end{algorithm}
We introduce first the \emph{naive} version of optimistic search (OS) within a segment $(l,r]$ 
in \cref{Alg:OS}, {\rev which is a key building block for later introduced statistically optimal methods.} The procedure is similar to the golden section search \cite{Kiefer1953, Avriel1966, Avriel1968}, typically used to find the global extremum of unimodal functions.  OS splits an interval into three segments recursively and discards one of the outer segments in each iteration. For unimodal functions with one peak this search {\rev converges to} the global maximum. If there is only a single change point contained in $(l,r]$, then the (continuously embedded) population gain function is unimodal with the single peak at the true underlying change point. 
Optimism is required in noisy scenarios, and hence the naming of the method, as noisy counterparts are rather ``wiggly'' functions following the shape of the population gain function only approximately.

When initializing by calling $\OS(l, \floor{(l+\nu r)/(1+\nu)}, r \mid \nu, l, r)$ for an interval $(l,r]$, the search first probes the points $t = l + (r-l) \nu/(1+\nu)$ and $w = r - (r-l) \nu/(1+\nu)$ (up to rounding), i.e.~the two first probe points are equally distant from $l$ and $r$ respectively. Depending on the gains at the probe points $t$ and $w$ either $(l,t]$ or $(w,r]$ is then discarded. The possible decisions for the general case when the search is already narrowed down to the sub-interval $(\tilde{l},\tilde{r}]$ is depicted in Figure~\ref{Fig:optim-search}. 
Note that in general, the lengths of the two candidate intervals for discarding are not necessarily equal. 
In case (a) we have $G_{(l,r]}(t) < G_{(l,r]}(w)$ and hence the less promising blue area is discarded. In case (b) we have $G_{(l,r]}(t) > G_{(l,r]}(w)$ and the red area is discarded. Also note that one of the previous probe points will be one of the new boundary points while the other probe point is going to be one of the new probe points in the middle with gain that is thus at least as high as for the new boundary. This leads to a ``triangular structure'' that one probe point in the middle has a higher gain than both boundary points, throughout the search.

\begin{figure}
\centering
\includegraphics[width=0.85\textwidth]
{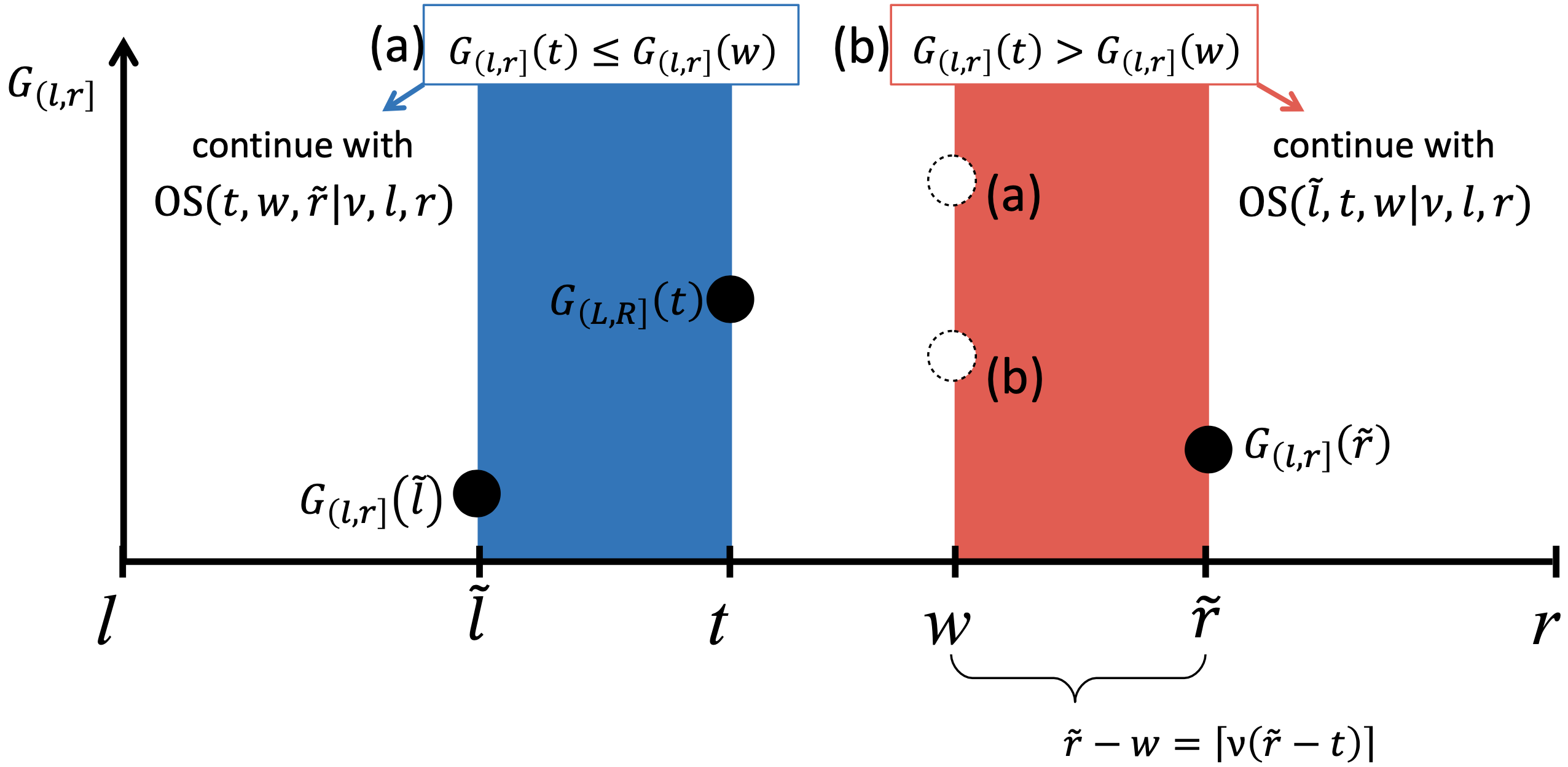} \\
\caption{Naive optimistic search step $\OS(\tilde l, t, \tilde r \mid \nu, l, r)$ within the current leftover segment $(\tilde l,\tilde r]\subseteq(l,r]$ given the previous evaluation at $t$ and step size $\nu$. As $\tilde r-t>t-\tilde l$, the new probe point $w$ is taken within $(t,\tilde r]$ as $w = \lceil \tilde r - (\tilde r - t)\nu \rceil$. Depending on the gain $G_{(l,r]}(w)$ vs.\ $G_{(l,r]}(t)$, one either continues with $\OS(t, w, \tilde r \mid \nu, l, r)$ (discarding the blue part) or $\OS(\tilde l, t, w \mid \nu, l, r)$ (discarding the red part).} 
\label{Fig:optim-search}
\end{figure}

We set $\nu=1/2$ by default, but in general $\nu$ can be interpreted as (relative) step size, expected to reflect some kind of trade off between computational performance and estimation accuracy. The choice of evaluating the last $5$ points remaining (in line $3$) is somewhat arbitrary and can be also set to e.g.~$10$ or~$20$.
In rare cases, when $r-t = t-l$ or $G_{(l,r]}(w) = G_{(l,r]}(t)$ one can also take a (pseudo) random choice or incorporate additional information (e.g.~variance in the segments for Model~\ref{Gaussian_setup}) to decide on the new probe point and the segment to discard.

\begin{theorem}[Naive optimistic search]\label{th:noss}
Under Model~\ref{Gaussian_setup} with a single change point, i.e.~$\ncp =1$, we assume that the minimal segment length $\lambda$ and the minimal jump size $\delta$ satisfy
\begin{equation}\label{eq:area}
\delta \lambda \sqrt{n} \ge C_0 \sqrt{{\rev \log\log n}} 
\end{equation}
for some large enough constant $C_0$. Let $\hat\cpt = \hat t_{(0,n]}/n$ be the estimated change point by OS (\cref{Alg:OS}) on $(0,n]$. Then:
$$
     \lim_{n \to \infty}\Prob{\abs{\hat\cpt - \cpt} \le C_1 \frac{{\rev \log\log n}}{\delta^2 n}} = 1 \qquad \text{with some constant } C_1\,.    
$$
\end{theorem}

\cref{th:noss} states that OS detects the only change point with a localization error of order $\log\log n / (\delta^2 n)$, which is minimax optimal up to a possible log-factor {\rev (see e.g.~Lemma~1 in \cite{VFLR20}).
In comparison with the weakest condition to ensure consistency of change point estimation, which is $\delta \sqrt{\lambda n} \gtrsim \sqrt{\log\log n}$ (see \cite{Liu_Gao_Samworth_dyadic}), OS is sub-optimal. 
However, in the particular case that the minimal segment length does not vanish, i.e.\ $\lambda \asymp 1$, the condition~\eqref{eq:area} becomes $\delta \sqrt{n}\gtrsim \sqrt{\log\log n}$ {and}
OS is then optimal (up to multiplying constants).} This is the situation where the length of the left segment is comparable to that of the right segment, which we thus refer to as a ``balanced'' scenario. In contrast, ``unbalanced'' scenarios are ones where the lengths of shorter and longer segments are very different. 
It is possible to show that the suboptimal condition \eqref{eq:area} cannot be improved
{and}
is intrinsic to OS, rather than an artifact in our theoretical analysis (see \cref{b:nos} in \cref{Appendix_Simulations}). 

\subsection{Advanced optimistic search}
In \cref{Alg:aOS} we propose \emph{advanced optimistic search} (aOS) that improves on the naive version to tackle also unbalanced cases with the change point {being} very close to the boundary. The main idea is to check a preliminary set of dyadic locations (up to rounding) to localize the change point approximately and then apply OS in a suitable (balanced) neighborhood around the preliminary estimate in order to achieve a better localization. The preliminary estimate and the two locations marking its neighborhood are chosen from the dyadic locations, namely, as the location of the biggest gain as well as the closest dyadic neighbors thereof (to the left and to the right). 
Intuitively, as the dyadic points are denser on the boundaries, the advanced search is suitable even in very unbalanced scenarios where the naive version fails. From a theoretical perspective, this modification 
leads to minimax optimality.

\begin{algorithm}[H]
\caption{Advanced Optimistic Search (aOS)}\label{Alg:aOS}
\begin{algorithmic}[1]
\Require $r-l > 2; ~l,r \in \mathbb{N}$ and step size $\nu \in (0,1)$ with $1/2$ by default
\Function{\textnormal{aOS}}{$\nu, l, r$}
\State $k \leftarrow \lfloor\log_2((r-l)/2)\rfloor$
\State $\mathcal{D} \leftarrow \bigl\{ \floor{l+2^{-i}(r-l)},\,\ceil{r-2^{-i}(r-l)}\,:\, i =1, \ldots, k\bigr\}$ \Comment{Dyadic locations}
\State $t_* \leftarrow \argmax\limits_{t \in \mathcal{D}} G_{(l,r]}(t)$ \Comment{Find the best split point on the ``dyadic'' grid}
\If{$t_* \le (r+l)/2$}
\State $\tilde l \gets \floor{t_* -(t_*-l)/2}$  and  $\tilde r \leftarrow \ceil{t_*+(t_*- l)}$ \Comment{``dyadic neighbours'' of $t_*$}
\Else
\State $\tilde l \leftarrow \floor{t_* - (r-t_*)}$  and $\tilde r \leftarrow \ceil{t_*+(r-t_*)/2}$ \Comment{``dyadic neighbours'' of $t_*$}
\EndIf
\State $\hat t_{(l,r]} \leftarrow$ \Call{\textnormal{OS}}{$\tilde l, t_*, \tilde r \mid \nu, l, r$} \Comment{Naive optimistic search on $(\tilde l,\tilde r]$ containing $t_*$}
\State \textbf{return} $\hat t_{(l,r]}$
\EndFunction
\end{algorithmic}
\end{algorithm}

\begin{theorem}[Advanced optimistic search]\label{th:aoss}
Under Model~\ref{Gaussian_setup} with a single change point, i.e.~$\ncp =1$, we assume that the minimal segment length $\lambda$ and the minimal jump size $\delta$ satisfy
\begin{equation}\label{eq:best}
\delta  \sqrt{\lambda n} \ge C_0 \sqrt{{\rev\log\log n}} 
\end{equation}
for some large enough constant $C_0$. Let $\hat\cpt = \hat t_{(0,n]}/n$ be the estimated change point by aOS (\cref{Alg:aOS}) on $(0,n]$. Then:
$$
\lim_{n \to \infty}\Prob{\abs{\hat\cpt - \cpt} \le C_1 \frac{{\rev \log\log n}}{\delta^2 n}} = 1 \qquad \text{with some constant } C_1\,.
$$
\end{theorem}

Similar to OS (Theorem~\ref{th:noss}), it is shown in Theorem~\ref{th:aoss} that aOS is able to localize the only change point at the best possible rate up to a log-factor, but now under a much weaker condition~\eqref{eq:best} instead. In comparison with the weakest condition $\delta \sqrt{\lambda n} \gtrsim \sqrt{\log \log n}$ in~\cite{Liu_Gao_Samworth_dyadic}, {\rev we lose nothing except for a possibly larger multiplying constant.} Therefore, aOS possesses the (nearly) statistical minimax optimality like the full grid search, which checks every possible split point in $\{1, \ldots, n\}$. Note that aOS (and OS) only requires $O(\log n)$ evaluations of the gain function (\cref{Lem:cOS_speed} later), in sharp contrast to $O(n)$ required by the full grid search. It is a surprising fact that computational speed-ups come at almost no cost of statistical performance at all. {In this sense.}
``free lunch'' is possible!

{\rev The idea of preliminary check of dyadic locations dated back to \cite{Rufibach_Walther} (or even earlier to wavelets) and was recently explored in \cite{Liu_Gao_Samworth_dyadic,SeedBS}. }

In practice, variants of (a)OS might be equally viable e.g.\ the combination of OS and aOS, referred to as \emph{combined OS,} see \cref{Appendix_cOS} for details. 

\begin{lemma} \label{Lem:cOS_speed}
OS and aOS (\cref{Alg:OS,Alg:aOS}) terminate in $O(\log(r-l))$ and thus at most $O(\log n)$ steps (i.e.\ number of gain function evaluations).
\end{lemma}

For the univariate Gaussian setting, the overall computational cost is only $O(\log n)$ if cumulative sums have been pre-computed and are freely available, as in that case each evaluation is possible in $O(1)$ time. Otherwise the $O(n)$ cost of calculating the cumulative sums becomes dominant. We remark that availability of cumulative sums (or similar ``sufficient statistics'' for the evaluation of the gain function) is 
a practical recommendation
to store 
data for off- and online change point problems.

%% file: sections/4Method_multiple.tex
\section{Methodology and theory for multiple change points}
\label{Methodolgy_multiple}

We consider now the setup of multiple change points, and investigate how our methodology can be extended to such a more ambitious setup in order to still have a sublinear number of evaluations of the gain function and yet with theoretical {optimality} guarantees for the estimation performance. 

{Obviously, the extension of optimistic searches to multiple change points is not straight-forward.} Here we adopt the idea of Seeded Binary Segmentation (SeedBS, \cite{SeedBS}), which 
searches for a single change point in various intervals with the hope that some of these intervals contain only a single change point, where the detection is ``easy''. While the best split point in each interval is a candidate, the decision which candidates to declare finally as change points depends on a subsequent selection step. The intervals are called seeded intervals and they are constructed deterministically (see \cref{def:seeded_intervals} below). SeedBS is thus very similar to wild binary segmentation (WBS, \cite{Fryzlewicz_WBS}) and the narrowest over threshold method \cite{Baranowski}. The latter two procedures use random intervals instead of the deterministic ones
and 
in general leads to total length and number of considered intervals to be larger and thus computationally more expensive.

\begin{definition}[Seeded intervals;  \cite{SeedBS}]
\label{def:seeded_intervals}
Let $a \in [1/2,1)$ denote a given decay parameter. Let $I_1=(0,n]$. For $k = 2,\ldots, \lceil\log_{1/a}(n)\rceil$ (i.e.~logarithm with base $1/a$) define the $k$-th layer $\I_k$ as the collection of $n_k$ intervals of initial length $l_k$ that are evenly shifted by $s_k$:
\begin{equation*}
    \I_k = \bigcup\limits_{i=1}^{n_k} \{ (\lfloor (i-1)s_k \rfloor,
        \lceil(i-1)s_k + l_k \rceil ] \},  
\end{equation*}
where
$n_k = 2\lceil (1/a)^{k-1} \rceil - 1$,  
$l_k  = n a^{k-1}$ and
$s_k  = (n-l_k)/(n_k-1)$.
The overall collection of seeded intervals is
\begin{equation*}
    \I = \bigcup\limits_{k=1}^{\lceil\log_{1/a}(n)\rceil} \I_k.
\end{equation*}
\end{definition}

Note that $\I$ covers the whole range of scales and locations in an efficient way such that there are $O(n)$ intervals, which is constructed to guarantee appropriate background for different types of change points. When there is only one change point, all intervals that do not have a starting point at $1$ or do not have an end point at $n$ can be discarded, reducing the number of intervals to $O(\log n)$. In the case of multiple change points, assuming a certain minimal spacing between change points also allows to discard intervals that are too short. 

\begin{algorithm}[H]
\caption{Optimistic Seeded Binary Segmentation {(OSeedBS)}}\label{Alg:OSeedBS}
\begin{algorithmic}[1]
\Require a decay parameter $a\in[1/2,1)$, a minimal segment length $m \ge 2$, and tuning parameters for the selected optimistic search (naive or advanced)
\Function{OSeedBS}{}
 \State $\I \gets$  seeded intervals with decay $a$ and at least $m$ observations.
\For{$(l,r] \in \I$}
\State \qquad $\hat t_{(l,r]}\gets$ the split point returned by the optimistic search on $(l,r]$.
\EndFor
\State Apply some selection method to $\bigl(\hat t_{(l,r]},G_{(l,r]}(\hat t_{(l,r]})\bigr), (l,r] \in \I$, to output the change point estimates $\hat\cpt_1,\ldots, \hat \cpt_{\hat\ncp}$.
\State {Post-process: refine the estimated change points by applying the optimistic search to intervals $\bigl((\hat\tau_{i-1}+\hat\tau_i)/2,\;(\hat\tau_i + \hat\tau_{i+1})/2\bigr]$, $i = 1,\ldots, \ncp$.}
\EndFunction
\end{algorithmic}
\end{algorithm}

The difference between SeedBS \citep{SeedBS} and its optimistic counterpart OSeedBS is essentially in line~$4$ of \cref{Alg:OSeedBS}, where we perform either OS or aOS rather than full grid search. 
The selection method in line~$5$ can be for example greedy or narrowest-over-threshold (NOT) selection, see \cite{SeedBS}. The computational times of OSeedBS depend critically on the minimal segment length $m$. If $m=O(n)$, only a handful of intervals are considered, with $O(\log n)$ evaluations each, and thus $O(\log n)$ evaluations overall. For the other extreme, when $m$ is very small, many intervals need to be generated and thus the main driver of the number of evaluations (and hence, the computational cost) is the number of considered intervals. For $m=2$, the number of intervals and also the total number of evaluations is $O(n)$. The estimation performance of course also depends on the choice of $m$. If chosen too large, estimation performance will be bad as change points within short segments may not be detected. Thus, 
$m$ offers some kind of trade-off between estimation performance and computational efforts. Such trade-offs are inherent also in other methods, see e.g.~the number of random intervals chosen in WBS.

\begin{theorem}\label{th:osbs}
Under \cref{Gaussian_setup}, we assume that the minimal segment length $\lambda$ and the minimal jump size $\delta$ satisfy \begin{equation}\label{eq:bestm}
\delta  \sqrt{\lambda n} \ge C_0 \sqrt{\log n} 
\end{equation}
for some large enough constant $C_0$. Assume further that there is an a-priori known lower bound $\lambda_*$ of all segment lengths, i.e. 
\begin{equation}\label{eq:mlen}
\lambda \ge \lambda_* \asymp n^{-\omega} \qquad \text{for some constant } \omega \in [0,\,1]\,.
\end{equation}
By $\hat\ncp$ and $\hat\cpt_1 < \cdots < \hat\cpt_{\hat \ncp}$ denote respectively the number and the locations of estimated change points by OSeedBS (Algorithm~\ref{Alg:OSeedBS}), with the NOT selection method, and the seeded intervals of lengths larger than $m = \lfloor \lambda_* n/3\rfloor$. Then:
\begin{enumerate}[i.]
\item\label{i:osbs1}
There exist constants $C_1$, $C_2$, independent of $n$, $\omega$, $a$ and $m$, such that, given the threshold for the selection method 
$\gamma = C_1 \sqrt{\log n}$,
$$
\lim_{n \to \infty}\Prob{\hat \ncp = \ncp,\,\, \max_{i = 1, \ldots, \ncp} \delta_i^2\abs{\hat \cpt_i - \cpt_i} \le C_2 \frac{\log n}{n}} = 1\,.
$$
\item\label{i:osbs2}
The number of evaluations is $O(\min\{n^\omega \log n,\, n\})$.
\end{enumerate}
\end{theorem}

We emphasize that the assumption~\eqref{eq:mlen} is only needed for computational {efficiency ensuring}
a sublinear number of evaluations as specified in part~\ref{i:osbs2}~of Theorem~\ref{th:osbs}. If the data are stored in the format of cumulative sums, then the overall computational cost itself is also $O(\min\{n^\omega \log n,\, n\})$, i.e.\ it equals the number of evaluations. However, if cumulative sums are not available, then the $O(n)$ cost of calculating cumulative sums becomes dominant and the overall computational cost is $O(n)$, see \cref{App:Computation}.
For the statistical guarantee in part~\ref{i:osbs1}, an assumption on the minimal spacing, i.e.~\eqref{eq:mlen}, becomes obvious if we choose $\omega = 1$ and $m =2$, since it is pointless to work on a higher resolution than the sampling rate $1/n$ without further model assumption, and thus in this sense it imposes no restriction at all. In case of multiple change points, the signal strength condition \eqref{eq:bestm} is the weakest one that still allows for detection. It
coincides with the best known results (e.g.~\cite{Frick_etal,Baranowski,Chan_Chen2017,VFLR20})
with {the} only difference in multiplying constants.
Following the proof in \cref{ss:mcpt}, we can easily replace it by 
$$
\min_{i=1,\ldots,\ncp} \Bigl(\min\{\cpt_{i+1}-\cpt_i,\,\cpt_i - \cpt_{i-1}\}\delta_i^2 n\Bigr)\,\, \gtrsim \,\,\log n\qquad \text{as} \quad n \to \infty.
$$
This is slightly more general, as it allows for frequent large jumps and {rare}
small jumps over long segments (cf.~\cite{ChKi19}). However, we prefer the current version as in \cref{th:osbs}, for notational simplicity. Note, moreover, that the localization rate reported in part~\ref{i:osbs1} of \cref{th:osbs} is minimax optimal up to a possible log factor. {In the particular case of $\kappa \asymp n^\theta$ with some constant $\theta > 0$, the derived localization rate is indeed optimal (namely, the log factor being necessary, see \cite{VFLR20}).
}

WBS \citep{Fryzlewicz_WBS}, and the similar narrowest over threshold method \citep{Baranowski}, 
can also be sped up using OS (or aOS). However, in the worst case with very short segments, e.g.~in frequent change point scenarios with up to $O(n)$ change points, these two methods need to draw up to $O(n^2)$ random intervals, which prohibits sublinear number of evaluations overall. Nonetheless, we expect substantial computational gains using OS (or aOS) in connection with many other multiple change point detection techniques compared to the respective full grid search based counterparts.

%% file: sections/4bMultivariateHD.tex
\section{{Extension to multivariate and high-dimensional scenarios}}\label{s:mhdgauss}
{
In the previous sections, the theoretical findings on the univariate Gaussian {changing means}
\cref{Gaussian_setup}) reveal the \emph{statistical insight} that the computational efficiency can be improved by almost one order (more precisely, from $O(n)$ to $O(\log n)$ evaluations) with nearly no loss of statistical efficiency, using optimistic search strategy.  {We show that this is also true for Gaussian changing means problems}
of general {and potentially high} dimension.

\subsection{{The multivariate model}
and {some} technical simplification}\label{ss:mats}
\begin{Model}[Gaussian changing means]\label{m:mhgauss}
Assume that vectors $\vec{X}_1, \ldots, \vec{X}_n \in \R^p$ are independent and 
\begin{equation*}
\begin{split}
 \vec{X}_{\cpt_{0}n+1}(=\vec{X}_1), \ldots, \vec{X}_{\cpt_1 n} & \sim \mathcal{N}(\vec{\mu}_0,  \vec{I}_p)\,, \\
& \vdots	\\
 \vec{X}_{\cpt_{\ncp}n+1}, \ldots, \vec{X}_{\cpt_{\ncp+1}n}(=\vec{X}_{n}) 	& \sim \mathcal{N}(\vec{\mu}_{\ncp}, \vec{I}_p)\,,
\end{split}
\end{equation*}
where $\{\cpt_i \,:\, i = 1, \ldots, \ncp\}$ gives the points satisfying
$$ 
0 = \cpt_0 < \tau_1 < \cdots < \cpt_{\ncp+1} = 1\quad \text{and}\quad \cpt_i n \in \N\,,
$$
mean vectors $\vec{\mu}_i \neq \vec{\mu}_{i-1} \in \R^p$ for $i=1,\ldots,\ncp$ give the levels on segments, and the common covariance matrix is $\vec{I}_p\in \R^{p \times p}$ the identity matrix. 
Define the minimal segment length $\lambda$ as
$$
\lambda \equiv \lambda_n = \min_{i = 0,\ldots,\ncp} (\cpt_{i+1} - \cpt_i)\,,
$$
and the minimal jump size $\delta$ as
$$
\delta \equiv \delta_n = \min_{i = 1,\ldots, \ncp} \delta_i \qquad \text{with} \quad \delta_i = \norm{\vec{\mu}_i - \vec{\mu}_{i-1}}\,,
$$
with $\norm{\cdot}$ the Euclidean norm. In addition, assume that there is a \emph{known} integer $s \in \{1, \ldots, p\}$ such that, for $i = 0, \ldots, \kappa$, 
$$
\norm{\vec{\mu}_i - \vec{\mu}_{i-1}}_0 := \#\{j = 1, \ldots, p\; \mid\;  \mu_{i,j} \neq \mu_{i-1,j}\} \le s,
$$
with $\mu_{i,j}$ the $j$-th entry of $\vec{\mu}_i$. 
\end{Model}

In \cref{m:mhgauss} the locations of change points are shared over $p$ coordinates, and thus it potentially allows aggregation of detection power among different coordinates. In case that the change of means happens in only a sparse {fraction}
of coordinates (i.e.\ $s \ll p$), one {should focus only}
on the coordinates where the mean changes.
The selection of changing coordinates can be achieved by a simple thresholding {rule}, see \cite{Liu_Gao_Samworth_dyadic}. This motivates us to consider the following gain function:
\begin{equation}\label{e:mhgain}
    G^o_{(l,r]}(t) := \sum_{j =1}^p\left(\csm_{(l,r],j}(t; \vec{X})^2 - \thd^2\right)_+,
\end{equation}
where $0 \leq l < t < r \leq n$ are integers, $\alpha \ge 0$ is a user-specified threshold, and  $\csm_{(l,r],j}(t; \vec{X})$ is the CUSUM statistics in the $j$-th coordinate of $\vec{X} = (\vec{X}_1, \ldots, \vec{X}_n)$, namely,
$$
\csm_{(l,r],j}(t; \vec{X}) = \sqrt{\frac{r-t}{(r-l)(t-l)}}\sum_{i=l+1}^t X_{i,j} -\sqrt{\frac{t-l}{(r-l)(r-t)}}\sum_{i=t+1}^r X_{i,j}, 
$$
with $X_{i,j}$ the $j$-th entry of $\vec{X}_i$.

In the gain function \eqref{e:mhgain} the CUSUM statistics is utilized for change point estimation as well as for coordinate selection. This entanglement of change point estimation and coordinate selection complicates the theoretical analysis. We employ \emph{two technical modifications} to ease the theoretical analysis.

The first is a sample splitting trick that removes the aforementioned entanglement. We split the samples from \cref{m:mhgauss} into two independent groups, with one group at odd times, and the other at even times. One group of samples is then used for the estimation of change points, and the other for the selection of coordinates. For simplicity, we assume that there are two independent copies of samples, denoted as $\{\vec{X}_1, \ldots, \vec{X}_n\}$ and $\{\vec{Y}_1, \ldots, \vec{Y}_n\}$, from \cref{m:mhgauss}. Then the modified gain function is defined as 
\begin{equation}\label{e:mhgm}
    G_{(l,r]}(t) := \sum_{j=1}^p \left(\csm_{(l,r],j}(t; \vec{X})^2-1\right)\ind\left\{\abs{\csm_{(l,r],j}(t; \vec{Y})} \ge \thd\right\},
\end{equation}
with threshold $\alpha \ge 0$ and integers $0\le l < t< r \le n$.

Recall that the basic operation in optimistic searches is the comparison between a pair of locations to determine which one is more likely to be a change point. Such a comparison is done via the absolute scores determined by the gain function, but it is also feasible whenever relative scores are available. Thus, as the second modification, we introduce a relative score between two locations $t$ and $w$ in the form of {a} \emph{comparison function}
\begin{multline}\label{e:mhcmp}
    \cmp_{(l,r]}(t,w) := 
    \sum_{j=1}^p \biggl(\Bigl(\csm_{(l,r],j}(t; \vec{X})^2-\csm_{(l,r],j}(w; \vec{X})^2\Bigr) \times \\ \ind\Bigl\{\max\bigl(\abs{\csm_{(l,r],j}(t; \vec{Y})}, \abs{\csm_{(l,r],j}(w; \vec{Y})} \bigr)\ge \thd\Bigr\}\biggr).
\end{multline}
The location $t$ is preferred as a change point candidate rather than $w$, if and only if $\cmp_{(l,r]}(t,w) \ge 0$. This second modification means that in OS and aOS (\cref{Alg:OS,Alg:aOS}) we use the comparison function in \eqref{e:mhcmp} instead of the gain function in \eqref{e:mhgm} to decide which location is preferred. Besides, in the dyadic search in aOS (\cref{Alg:aOS}, line 4),
the maximum of the gain function should be replaced by the location at which it is preferable in terms of comparison function over all other dyadic locations. However, for a found change point candidate by OS or aOS, the gain function in \eqref{e:mhgm} is {still} used to decide whether it should {be} selected as an estimated change point, in OSeedBS (\cref{Alg:OSeedBS}, line 5).

These two modifications are needed only when $\alpha > 0$.
In case of $\alpha = 0$, the modified gain function \eqref{e:mhgm} is the same as the original gain \eqref{e:mhgain}, and the comparison function \eqref{e:mhcmp} is simply the difference of {the} original gain function \eqref{e:mhgain} at two locations. 

\subsection{Statistical guarantees}\label{ss:stgr}
Given the above two technical modifications, we can establish the following statistical properties of the optimistic searches.  

\begin{theorem}[Single change point]\label{th:mh1cp} 
Under \cref{m:mhgauss} with $\kappa = 1$ (i.e.\ a single change point), assume that  the minimal segment length $\lambda$ and the minimal jump size $\delta$ satisfy
\begin{subequations}\label{e:db1cp}
\begin{equation}\label{e:db1cp:a}
  n \lambda \delta^2\ge C_0 \rho_{\circ}(n, p, s),
\end{equation}
where $C_0 > 0$ is a sufficiently large constant, and 
\begin{equation}\label{e:db1cp:b}
\rho_{\circ}(n, p, s) := \begin{cases}
\sqrt{p \log\log n} & \text{if } s \ge \sqrt{p \log\log n},\\
\max\left(s\log\frac{e \sqrt{p \log\log n}}{s},\log\log n\right)& \text{if } s \le \sqrt{p \log\log n}.
\end{cases}
\end{equation}
\end{subequations}
Let $\hat\cpt = \hat t_{(0,n]}/n$ be the estimated change point by aOS (\cref{Alg:aOS}) on $(0,n]$, with the gain \eqref{e:mhgm}, the comparison function \eqref{e:mhcmp} and 
$$
\thd \equiv \thd(n,p,s) = 
\begin{cases}
0 & \text{ if } s \ge \sqrt{p\log\log n}\text{ or } s = p,\\
\sqrt{2\log \frac{e^2{p\log\log n}}{s^2}}& \text{ otherwise}.
\end{cases}
$$
Then, for some constant $C_1> 0$, it holds that
$$
\lim_{n \to \infty}\Prob{\abs{\hat\cpt - \cpt} \le C_1 \max\Bigl\{\frac{ \log\log n}{\delta^2 n},\, \frac{\min\{s^2, p \log\log n\}}{n^2\lambda \delta^4}\Bigr\}} = 1.
$$
\end{theorem}

\cref{th:mh1cp} includes \cref{th:aoss} as a special case when $p = 1$, {with}
the threshold $\alpha = 0$. In general, the condition \eqref{e:db1cp} is the weakest possible for the detection of a single change point (cf.~\cite{Liu_Gao_Samworth_dyadic}), except that the 
constant $C_0$ might be suboptimal. For the localization of the change point, since $\rho_{\circ}(n, p, s) \ge \sqrt{\min\{s^2, p\log\log n\}}$,  \cref{th:mh1cp} implies
\begin{equation}\label{e:lcd}
\Prob{\abs{\hat\cpt - \cpt} \le C_1\frac{\rho_{\circ}(n,p,s)}{n\delta^2}} \to 1,\quad \text{as }n \to \infty.
\end{equation} 
This is the induced localization rate from the detection of change points, as reported in \cite{PiCV20}.  Intuitively, the induced rate can be obtained since one may treat the localization of a single change point up to an accuracy $\varepsilon_n$ \enquote{equivalently} as the detection of a single change point at $t = n\varepsilon_n$ with the same jump size. This connection between detection and localization of change points leads to almost sharp localization rates in univariate and multivariate cases, but may yield suboptimal rates in case of high dimensions. In fact, when $p \gg \log\log n$, the localization rate in \cref{th:mh1cp} is strictly faster than the rate in \eqref{e:lcd}, in the dense scenario ($s \ge \sqrt{p\log\log n}$) if $n\lambda\delta^2 \gg \sqrt{p \log\log n}$, and in the sparse scenario ($s \le \sqrt{p\log\log n}$) if $s \gtrsim \log\log n$. 

Moreover, \cref{th:mh1cp}  together with the condition
\begin{equation}\label{e:brcp1}
n \lambda \delta^2 \gtrsim \min\Bigl\{\frac{s^2}{\log\log n}, p\Bigr\},
\end{equation}
leads to 
$$
\Prob{\abs{\hat\cpt - \cpt} \lesssim \frac{\log\log n}{n\delta^2}} \to 1,\quad \text{as }n \to \infty,
$$
which is not improvable except for the factor of $\log\log n$ (cf.\ \cite[Proposition~3]{Wang_highdim_mean_change}). In the literature, stricter conditions than \eqref{e:brcp1} are requested for the same localization rate (ignoring the log factor), see e.g.\  \cite{BhBM17} and \cite{KFJS21}. 
We stress that in the low dimensional case of $p \lesssim \log\log n$, or in the highly sparse case of $s \lesssim \log\log n$, the  condition~\eqref{e:brcp1} is simply a consequence of \eqref{e:db1cp}, and thus the localization rate in \cref{th:mh1cp} is minimax optimal (up to a $\log\log n$ factor). However, in general, it remains unclear, whether the localization rate in \cref{th:mh1cp} is optimal or not, under the weakest detection condition~\eqref{e:db1cp}, see \cref{Discussion}.

The optimistic searches can be applied to the inference of multiple change points, if one incorporates the idea of 
SeedBS, which results in 
OSeedBS (\cref{Alg:OSeedBS}), 
as in the univariate setup of \cref{Gaussian_setup}. In particular, because of multiscale nature of seeded intervals, OSeedBS extends the statistical optimality of optimistic searches for a single change point to the general case of multiple change points.   
\begin{theorem}[Multiple change points]\label{th:mhmcp}
Under \cref{m:mhgauss}, we assume that the minimal segment length $\lambda$ and the minimal jump size $\delta$ satisfy 
\begin{subequations}\label{e:dbmcp}
\begin{equation}\label{e:dbmcp:a}
  n \lambda \delta^2\ge C_0 \rho(n, p, s),
\end{equation}
where $C_0 > 0$ is a sufficiently large constant, and 
\begin{equation}\label{e:dbmcp:b}
\rho(n, p, s) := \begin{cases}
\sqrt{p \log n} & \text{ if } s \ge \sqrt{p \log n},\\
\max\left(s\log\frac{e \sqrt{p \log n}}{s},\log n\right)& \text{ if } s \le \sqrt{p \log n}.
\end{cases}
\end{equation}
\end{subequations}
By $\hat\ncp$ and $\hat\cpt_1 < \cdots < \hat\cpt_{\hat \ncp}$ denote respectively the number and the locations of estimated change points by OSeedBS (\cref{Alg:OSeedBS}) with the NOT selection. Set the threshold $\thd$ in the gain \eqref{e:mhgm} and the comparison function \eqref{e:mhcmp} as
$$
\thd \equiv \thd(n,p,s) = 
\begin{cases}
0 & \text{ if } s \ge \sqrt{p\log n}\text{ or } s= p,\\
\sqrt{2\log \frac{e^2{p\log n}}{s^2}}& \text{ otherwise},
\end{cases}
$$
and the selection threshold $\gamma$ in NOT as $\gamma = C_1 \rho(n,p,s)$ for some constant $C_1 > 0$.
Then, there exists a constant $C_2$, such that, as $n\to\infty$,
$$
\Prob{\hat \ncp = \ncp,\,\, \abs{\hat \cpt_i - \cpt_i} \le C_2\max\Bigl\{\frac{ \log n}{\delta_i^2 n}, \frac{\min\{s^2, p \log n\}}{n^2\lambda \delta_i^4}\Bigr\}, \, i =1, \ldots, \ncp} \to 1.
$$
\end{theorem}

It is clear from the proof (\cref{ss:mcpt}) that if all segment lengths are lower bounded by $\lambda_*$, then \cref{th:mhmcp} remains valid even if the minimal length $m$ of seeded intervals is chosen as $m = \lfloor \lambda_* n/3\rfloor$. As a consequence, it covers part \ref{i:osbs1} of \cref{th:osbs} as a special case of $p = 1$. Such a-priori knowledge of $\lambda_*$ will lead to computational speed-ups, see \cref{ss:cmp} later.  
 
Since $\rho(n,p,s)\ge \sqrt{\min\{s^2, p\log n\}}$, the localization rate of \cref{th:mhmcp} implies the induced rate from the detection of change points, namely,  
\begin{equation}\label{e:lcdm}
\Prob{\hat\ncp = \ncp,\, \abs{\hat\cpt_i - \cpt_i} \le C_2\frac{\rho(n,p,s)}{n\delta_i^2},\, i=1,\ldots, \ncp} \to 1,\quad \text{as }n \to \infty,
\end{equation} 
which was reported in \cite{PiCV20}. Similar to the case of a single change point, the localization rate in \cref{th:mhmcp} can be strictly faster than the induced rate in \eqref{e:lcdm}. For instance, in the high-dimensional setup of $p \gg \log n$, this occurs in the dense scenario ($s \ge \sqrt{p\log n}$) when $n\lambda_i\delta_i^2 \gg \sqrt{p\log n}$, and in the sparse scenario ($s \le \sqrt{p\log n}$) when $s \gtrsim \log n$. 

The condition~\eqref{e:dbmcp} is minimax optimal in detection of two or more change points, while the minimax optimality of localization rates remains unclear. An exception is the case of $n\lambda_i\delta_i^2 \gtrsim \min\{s^2/\log n,\, p\}$, the localization rate of \cref{th:mhmcp} is of order $(\log n) / (\delta_i^2n)$ and thus not improvable except a possible log factor. In high dimensions, the optimal localization rates remain yet unknown, similar to the case of a single change point, see \cref{Discussion}.

Inspecting the proof of \cref{th:mhmcp} in \cref{ss:mcpt}, we can show that with no post-processing (line~6, \cref{Alg:OSeedBS}), the localization error rate is 
$$
\eps_i \asymp \max\set{\frac{\log n}{n\delta_i^2},\;\frac{\min\set{s^2,\, p\log n}}{n\delta_i^2 \gamma}},
$$
provided that the selection threshold $\gamma$ satisfies
$$
\rho(n,p,s)\lesssim \gamma \lesssim n\lambda \delta^2. 
$$
Then, if $\gamma \asymp n\lambda \delta_i^2$, the same localization rate as in \cref{th:mhmcp} can be achieved, but this is not practical or feasible, as $n\lambda \delta^2_i$ is often unknown, and may be larger than $n\lambda \delta^2$. Thus, if 
$$
\frac{\log n}{n\delta_i^2} \ll \frac{\min\set{s^2,\,p\log n}}{n\lambda\delta_i^4},
$$
the post-processing is necessary for the faster rate in \cref{th:mhmcp}.  Otherwise, e.g.\ in the univariate case, we can drop the post-processing step for the sake of saving computation. 

One restriction in the sparse scenario is that the level of sparsity $s$ is required (which appears in the threshold $\alpha$). One may adjust the gain function by considering a  proper selection of guesses on $s$, e.g.\
$$
\{1, 2, \ldots, 2^{\lceil \log_2\sqrt{p\log n}\rceil}\} \cup \{p\}
$$
and aggregating over such a selection, in a similar way as in \cite{Liu_Gao_Samworth_dyadic}. The careful investigation is left as future research. 

\subsection{Computational complexity}\label{ss:cmp}
The number of gain function evaluations remains the same as in the univariate case (\cref{Lem:cOS_speed} and \cref{th:osbs}\ref{i:osbs2}). If the data is stored in cumulative sums, each evaluation of gain function involves $O(p)$ computations, which leads to an additional factor of $p$ in run time; Otherwise, each evaluation of gain function will cost $O(np)$ computations, in which case the calculation of cumulative sums is recommended as a preprocessing of the data, which requires $O(np)$ computations only once.  

\begin{prop}\label{p:cmp}
Assume that the data from \cref{m:mhgauss} is stored in the format of cumulative sums, that is, $S_{t,j} := \sum_{i =1}^t X_{i,j}$ with $t = 1, \ldots, n$ and $j = 1,\ldots, p$. Then:
\begin{enumerate}[i.]
\item\label{i:cmp1}
In case of a single change point, the optimistic searches (\cref{Alg:OS,Alg:aOS}) require $O(p\log n)$ computations in the worst case. 
\item\label{i:cmp2}
In case of multiple change points, assume further that 
$$
\lambda \ge \lambda_* \asymp n^{-\omega} \qquad \text{for some constant } \omega \in [0,\,1]\,.
$$
OSeedBS (\cref{Alg:OSeedBS}) with  $m = \lfloor \lambda_* n/3\rfloor$ requires $O(p\min\{n^\omega\log n,\, n\})$ computations in the worst case.
\end{enumerate}
\end{prop}

Clearly, \cref{p:cmp} includes \cref{Lem:cOS_speed} and  part \ref{i:osbs2} of \cref{th:osbs} as a special case of $p =1$. In comparison, the full grid search has $O(pn)$ run time for a single change point, and SeedBS has $O(pn\log n)$ run time for multiple change points. Thus, the computational speedup based on optimistic searches can be polynomial in sample size. 

Moreover, we stress that OSeedBS can be sped up by {a} slight modification of the procedure as follows. On every seeded interval, we first select $s$ coordinates such that their corresponding squared CUSUM statistics evaluated at the middle point of this interval are the largest among all coordinates. Afterwards, we only consider the gain function restricted to the selected $s$ coordinates.  With this modification, OSeedBS will have an improved  worst case run time $O(s\min\{n^\omega\log n,\, n\} + p n^\omega)$, and will still enjoy the same statistical guarantee of OSeedBS established in \cref{th:mhmcp} (see \cref{ss:mcpt}).
}

%% file: sections/5Simulations.tex
\section{Simulations}
\label{Simulations}

{
We provide a simulation study of our optimistic search methods (including combined OS in \cref{Appendix_cOS}) on univariate changing means as well as higher dimensional changing covariance problems.}

\subsection{Single change point in univariate Gaussian means}
\begin{example} \label{setup:single_mean_shift}
Let $X_{1}, \ldots, X_{100}\sim\mathcal{N}(0,\sigma^2)$ and $X_{101}, \ldots, X_{100+n}\sim\mathcal{N}(0.5,\sigma^2)$ be independent observations with a single change in the mean value at observation $100$.  
\end{example}

{Simulation results are reported in \cref{tab:simulation_results} and \cref{Fig:single_cpt} in \cref{Appendix_Simulations}.}
While OS clearly struggles when the lengths of the two segments are very unbalanced ($n$ large, in particular with 
high noise level),
aOS has a much better
performance.
However, for the more balanced scenarios (up to $n=400$ for $\sigma=1.5$ for example), OS performs well.
The combined OS (\cref{Appendix_cOS}) 
has a slightly improved performance compared to aOS (in particular for the rather balanced scenarios). 
The full grid search has the best performance for the rather challenging scenarios that are very unbalanced and/or have a high noise level, but aOS and combined OS come remarkably close. We note that increasing absolute errors in change point location for higher values of $n$ despite having more available observations is actually reasonable, as there are meanwhile more potential candidates for change points on the grid. This is also compatible with the theoretical bound of order $\log \log n /\delta^2$.  


\subsection{Multiple change points in univariate Gaussian means}
\cref{setup:blocks_signal} {below} describes the blocks signal \cite{DoJo94} with the noise level as used by \cite{Fryzlewicz_WBS}.

\begin{example} \label{setup:blocks_signal}
Consider a total of $2048$ observations with $11$ change points at locations $205$, $267$, $308$, $472$, $512$, $820$, $902$, $1332$, $1557$, $1598$ and $1659$ as well as mean values $0$, $14.64$, $-3.66$, $7.32$, $-7.32$, $10.98$, $-4.39$, $3.29$, $19.03$, $7.68$, $15.37$ and $0$ between the change points to which independent Gaussian noise with a standard deviation of $\sigma = 10$ is added.
\end{example}

{The results are displayed in \cref{Fig:blocks_signal}}.
The average performances of both OS and combined OS are very close to the one from the full search. 
Overall, it turns out that the optimistic variants of OSeedBS have a competitive average performance compared to the full grid search based SeedBS. Moreover, as long as the minimal segment length constraints are short enough to guarantee coverage of each single change point, both SeedBS and variants of OSeedBS perform well.

\begin{figure}[!htb]
\centering
\includegraphics[width=0.8\textwidth]
{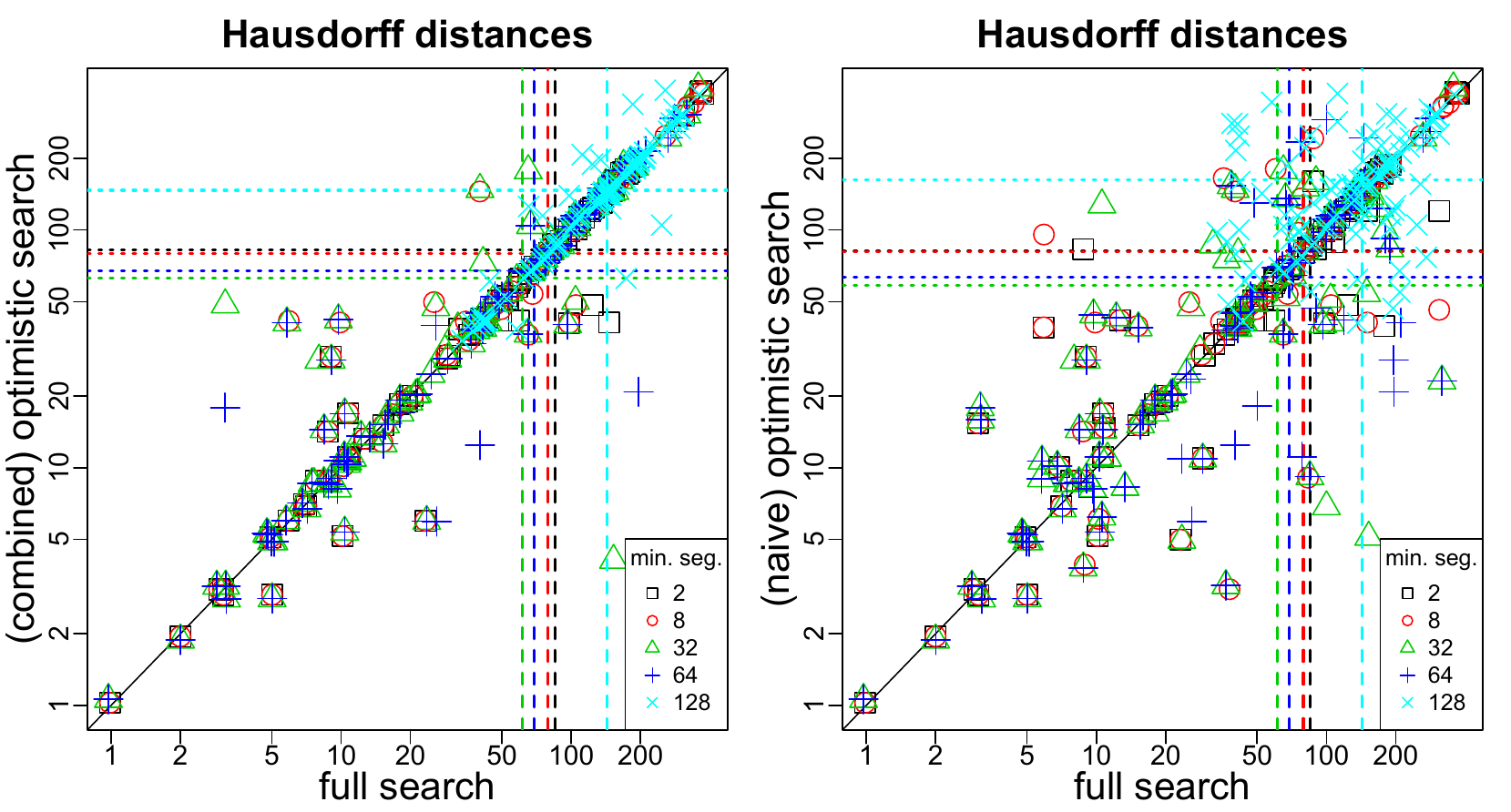} \\
\caption{{Results on Example 6.2.} Pairwise plots of Hausdorff distances of the locations of the best $11$ change point candidates (with greedy selection) compared to the true ones in $100$ simulations for SeedBS (decay $a=1/\sqrt{2}$) with various minimal segment length constraints and full grid search in each seeded interval (horizontal-axis) versus combined OS (vertical-axis, left) and OS (vertical-axis, right). 
The vertical and horizontal dashed lines indicate the average Hausdorff distances
for the respective minimal segment length constraint and search method within the seeded intervals. Note the logarithmic scales on both axes.} 
\label{Fig:blocks_signal}
\end{figure}


\subsection{{High-dimensional Gaussian covariance changes}}
\label{highdim_simulations}
{As an exploration on the potential of optimistic searches,} we introduce a changing covariance setup in \cref{GGM_setup} with specific 
instances for 
simulations 
in \cref{setup:high-dim}. 

\begin{Model}
\label{GGM_setup}
Assume that observations $\vec{X}_1, \ldots, \vec{X}_n$ are independent and 
\begin{equation*}
\begin{split}
 \vec{X}_{\cpt_{0}n+1}(=\vec{X}_1), \ldots, \vec{X}_{\cpt_1 n} & \sim \mathcal{N}(\vec{0}, \vec{\Sigma}_0)\,, \\
& \vdots	\\
 \vec{X}_{\cpt_{\ncp}n+1}, \ldots, \vec{X}_{\cpt_{\ncp+1}n}(=\vec{X}_{n}) 	& \sim \mathcal{N}(\vec{0}, \vec{\Sigma}_\ncp)\,,
\end{split}
\end{equation*}
where $\{\cpt_i \,:\, i = 1, \ldots, \ncp\}$ gives the locations of change points satisfying
$$
0 = \cpt_0 < \tau_1 < \cdots < \cpt_{\ncp+1} = 1\quad \text{and}\quad \cpt_i n \in \N\,.
$$
The means are $\vec{0}\in\mathbb{R}^p$ while $\vec{\Sigma}_i \neq \vec{\Sigma}_{i-1}\in\mathbb{R}^{p\times p}$ for $i=1,\ldots,\ncp$ give the covariances on the segments.
\end{Model}

A simulation setup of \cref{GGM_setup} 
was considered in \cite{SeedBS} as an example to demonstrate the computational efficiency of seeded intervals over random intervals utilized in WBS. We will show that further speed-ups in such computationally challenging setup for many available algorithms can be easily obtained utilizing our optimistic search strategies. 

\begin{example}
\label{setup:high-dim}
Let ${\Sigma}_{ij} = \exp{(-\frac{1}{2} \abs{t_i-t_j})}$ with $t_i - t_{i-1} = 0.75, i = 2,\ldots,20$ be a chain network model (see e.g.~Example~4.1 in \cite{fan2009_chain_network}) with $p=20$ variables. A modified version $\tilde{\vec{\Sigma}}$ is obtained by replacing the top left $5\times5$ block of $\vec{\Sigma}$ by a $5$-dimensional identity matrix. 
\begin{enumerate}
    \item[\emph{(a)}] In the setup from \cite{SeedBS} we set in \cref{GGM_setup} $\vec{\Sigma}_0=\vec{\Sigma}$, $\vec{\Sigma}_1=\tilde{\vec{\Sigma}}$, $\vec{\Sigma}_2=\vec{\Sigma}$, $\vec{\Sigma}_3=\tilde{\vec{\Sigma}}$, etc.~and draw $100$ observations for each segment until obtaining a total of $n=2,000$ observations. Hence, there are $20$ segments of length $100$ each, with a total of $\ncp=19$ change points.
    \item[\emph{(b)}] In \cref{GGM_setup} we set $\vec{\Sigma}_0=\vec{\Sigma}_2=\vec{\Sigma}_4=\vec{\Sigma}$, $\vec{\Sigma}_1=\vec{\Sigma}_3=\vec{\Sigma}_5=\tilde{\vec{\Sigma}}$, and draw $550$, $300$, $700$, $250$, $100$ and $100$ observations for the respective segments, obtaining again a total of $n=2,000$ observations, but this time with $6$ segments, i.e., a total of $\ncp=5$ change points.
\end{enumerate}
\end{example}

We consider a gain function (detailed definition in \cref{Appendix_highdim})
based on the multivariate Gaussian log-likelihood where the underlying precision matrices are obtained by
the graphical lasso
\cite{glasso}. The graphical lasso is rather costly especially when repeatedly fitting at each possible split point $s$ on a grid. The essential problem is that the estimator $\hat{\vec{\Omega}}_{(u,s]}^{\mathrm{glasso}}$ of precision matrix for a segment $(u, s]$ cannot be efficiently updated (not even using warm starts) to obtain $\hat{\vec{\Omega}}_{(u,s+1]}^{\mathrm{glasso}}$ for the segment $(u,s+1]$. Hence, the overall number of graphical lasso fits is the main driver of computational time.

This chosen gain function is motivated by the fact that its population version attains local maxima only at change points. More precisely, the population gain has the form of
\begin{equation}\label{e:gainHD}
    G^*_{(l,r]}(t)
    =
    \frac{r-l}{n}\log(|\vec{\Sigma}_{(l,r]}|) 
    -\frac{t-l}{n}\log(|\vec{\Sigma}_{(l,t]}|)
    -\frac{r-t}{n}\log(|\vec{\Sigma}_{(t,r]}|)\,,
\end{equation}
where $\vec{\Sigma}_{(l,r]}$ is the average covariance matrix on the segment $(l,r]\subseteq(0,1]$,
$$
\vec{\Sigma}_{(l,r]} = \sum_{i = 0}^\ncp \frac{\abs{(\cpt_i,\,\cpt_{i+1}]\cap (l,\,r]}}{\abs{(l,\,r]}}\, \vec{\Sigma}_i\qquad\text{with} \;\;\abs{(a,\,b]} = b-a,
$$
and $\abs{\vec{A}}$ is the determinant of a matrix $\vec{A}$.

\begin{lemma} \label{Lem:GGM_convexity}
The function $G^{*}_{(l,r]}(\cdot)$ defined in \eqref{e:gainHD} is piecewise (i.e.~in between change points) convex, and up to some special cases as detailed in the proof (see Appendix~\ref{App:GGM_convexity_proof}) even strictly convex.
\end{lemma}
{
\cref{Lem:GGM_convexity} implies that in the presence of a single change point in $(l,r]$, $G^*_{(l,r]}$ is unimodal and in case $(l,r]$ contains multiple change points, then each strict local maximum corresponds to a change point.}

We compare the estimation performance and computational times of {various methods.} In order to eliminate the effect of model selection, for all algorithms we selected greedily as many change points as the true underlying number (with some exceptions for WBS with small~$M$).

\begin{figure}[!ht]
\centering
\includegraphics[width=0.9\textwidth]
{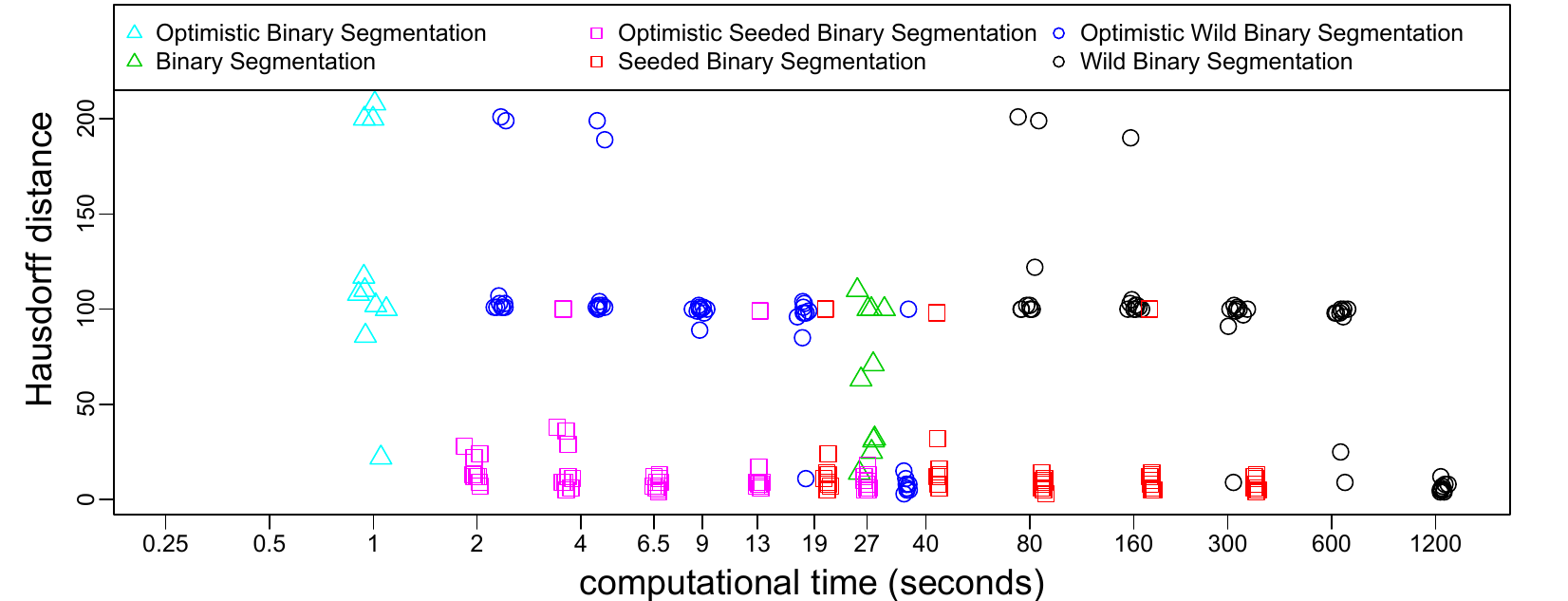} \\
\includegraphics[width=0.9\textwidth]
{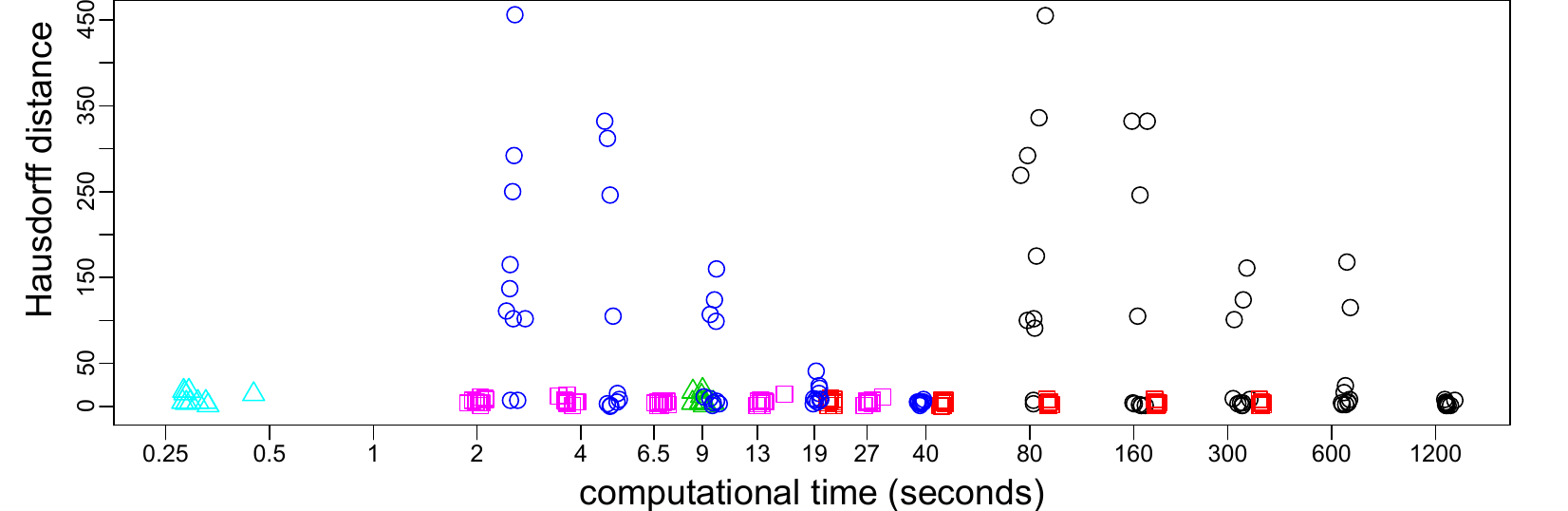} \\
\vspace{-0.2cm}
\caption{Estimation performances (in terms of Hausdorff distance) and computational times on the changing Gaussian graphical model (\cref{GGM_setup}) from \cref{setup:high-dim} (based on 10 simulations) with setups~(a) on the top and~(b) on the bottom. 
The symbols differentiate between the basic algorithms and the colors indicate whether full grid search or optimistic search was used. The five point clouds for SeedBS and OSeedBS correspond to decay parameters $a=2^{-1}, 2^{-1/2}, 2^{-1/4}, 2^{-1/8}, 2^{-1/16}$ for the seeded intervals, while the five point clouds for WBS and OWBS correspond to $M = 100, 200, 400, 800, 1600$ random intervals. {For all algorithms, the true number of change points (or the maximally many if not achieved) has been used .}}
\label{Fig:glasso_multiple_changes}
\end{figure}

We display the results in \cref{Fig:glasso_multiple_changes} which show roughly speedups of factor~$30$ for OBS compared to BS, factor~$35$ for OWBS versus WBS and factor~10--14 for OSeedBS versus SeedBS in both of the considered setups. 
We provide a further discussion on the speedups of different approaches and potential benefits of combining them in Appendix~\ref{Appendix_computational_gains_highdim}.

%% file: sections/6Discussion.tex
\section{Discussion}
\label{Discussion}
We introduced optimistic search strategies that avoid the full grid search and thus lead to computationally fast change point detection methods in great generality. 
For {univariate, multivariate and high-dimensional} Gaussian changing means setups we proved that aOS is asymptotically minimax optimal for detecting a single change point with only a logarithmic number of evaluations of the gain function. For multiple change point problems we combined optimistic searches (OS and aOS) with seeded binary segmentation, leading to asymptotically minimax optimal detection while having superior runtime compared to existing approaches. {In addition, the localization rate of change points is by far the sharpest, given the weakest possible condition on the signal-to-noise ratios. It is unclear though whether our localization rate is optimal or not in certain high-dimensional scenarios. In the literature, a faster localization rate is shown to be possible in certain regimes with much {larger} signal-to-noise ratio (see \cite[Theorems~1 and~2]{Wang_highdim_mean_change}). In particular, it indicates that our localization rate is \emph{not adaptively} minimax optimal over all possible ranges of signal-to-noise ratios. The complete understanding of localization rates is, to the best of our knowledge, still open for high-dimensional Gaussian mean changes, which offers an interesting avenue for future research in this direction. Overall, our theoretical results reveal a surprising fact that the computational acceleration up to one order in sample size can be achieved (by optimistic searches) with nearly no loss of statistical efficiency.} 

Our methodology is also most relevant for complex change point detection problems with computationally expensive model fits, as demonstrated by the massive computational gains in examples involving high-dimensional graphical models.

%% file: sections/7AppendixandProofs.tex
{
\section{Combined optimistic search}
\label{Appendix_cOS}

Combining the results of naive and advanced optimistic search (i.e.\ OS and aOS), thus referred to as the \emph{combined optimistic search}, leads to slightly better empirical performance than the individual searches, but at a slightly higher computational cost, see \cref{Alg:combinedOS}. Also from a theory point of view, the combined optimistic search enjoys the same statistical minimax optimality as the advanced version, see \cref{r:cos} later in \cref{App:proofs}.

\begin{algorithm}[H]
\caption{Combined Optimistic Search}\label{Alg:combinedOS}
\begin{algorithmic}[1]
\Require $r-l > 2; ~l,r \in \mathbb{N}$ and step size $\nu \in (0,1)$ with $1/2$ by default
\Function{\textnormal{cOS}}{$\nu, l, r$}
\State $\hat t_0\leftarrow \aOS(\nu, l, r)$ \Comment{Advanced optimistic search}
\State $\hat t_1 \leftarrow \OS(l, \floor{(l + \nu r)/(1+\nu)}, r\mid\nu, l, r)$ \Comment{Naive optimistic search}
\If{$G_{(l,r]} (\hat t_0) \ge G_{(l, r]}(\hat t_1)$}
\State $\hat t_{(l,r]} \leftarrow \hat t_0$
\Else
\State $\hat t_{(l,r]} \leftarrow \hat t_1$
\EndIf
\State \textbf{return} $\hat t_{(l,r]}$
\EndFunction
\end{algorithmic}
\end{algorithm}
}

\section{Additional material on the univariate Gaussian simulations}
\label{Appendix_Simulations}

{
\begin{example}\label{b:nos}
Consider a specific example of Model~\ref{Gaussian_setup} with $\ncp = 1$ and $\cpt_1 = \lambda \le 1/3$. For  simplicity, let  $\nu = 1/2$  and $n/3 \in \N$. In the first step of OS (naive optimistic search), we check the gain function at $n/3$ and $2n/3$. In order to avoid wrongly discarding $(0, n/3]$, we have to ensure $$\abs{\csm_{(0,n]}\Bigl(\frac{n}{3}\Bigr)} \ge \abs{\csm_{(0,n]}\Bigl(\frac{2n}{3}\Bigr)}\,.$$
In fact, we have
\begin{multline*}
    \Prob{\abs{\csm_{(0,n]}\Bigl(\frac{n}{3}\Bigr)} <\csm_{(0,n]}\Bigl(\frac{2n}{3}\Bigr)}
    \le\Prob{\abs{\csm_{(0,n]}\Bigl(\frac{n}{3}\Bigr)} < \abs{\csm_{(0,n]}\Bigl(\frac{2n}{3}\Bigr)}} \\
    \le 2\,\Prob{\abs{\csm_{(0,n]}\Bigl(\frac{n}{3}\Bigr)} <\csm_{(0,n]}\Bigl(\frac{2n}{3}\Bigr)}\,.
\end{multline*}
Elementary calculation using 
properties of the Gaussian distribution reveals 
$$
\Prob{\abs{\csm_{(0,n]}\Bigl(\frac{n}{3}\Bigr)} <\csm_{(0,n]}\Bigl(\frac{2n}{3}\Bigr)} = \Phi\Bigl(-\delta \lambda\sqrt{\frac{n}{2}}\Bigr)\Phi\Bigl(\delta \lambda\sqrt{\frac{3n}{2}}\Bigr)\,,
$$
with $\Phi$ the distribution function of a standard Gaussian random variable. Thus, if and only if $\delta \lambda \sqrt{n} \to \infty$ it holds that 
$$
\Prob{\abs{\csm_{(0,n]}\Bigl(\frac{n}{3}\Bigr)} \ge \abs{\csm_{(0,n]}\Bigl(\frac{2n}{3}\Bigr)}} \to 1 \qquad \text{as}\quad n \to \infty\,.
$$
Note that $\delta \lambda \sqrt{n} \to \infty$ is, up to a log factor, equivalent to the condition~\eqref{eq:area}, which guarantees that the probability of making a mistake in the first step of OS vanishes eventually. 
\end{example}}

\cref{tab:simulation_results} displays various results for the \cref{setup:single_mean_shift} in the main text. The top part of \cref{tab:simulation_results} shows the localization error of the change point estimates found by the naive, advanced and combined optimistic search, as well as the full grid search for various choices of $n$ (from 100 to 5,000) for three different noise levels ($\sigma = 0.5, 1, 1.5)$. The bottom part of \cref{tab:simulation_results} shows the number of evaluations as a measure of computation times.

\begin{table}
\caption{\label{tab:simulation_results} Simulation results for \cref{setup:single_mean_shift} for various choices of noise level $\sigma$ and number of observations $n$ from the second segment. Reported are average absolute differences between the true change point at location $100$ and the location of the best single split point found by the respective search method (top three blocks) as well as the average number of evaluations (bottom block). Values are averaged over $10000$ simulations (rounded to two digits) and in parentheses the corresponding standard deviations (rounded to integers). The average number of evaluations for noise levels $\sigma=0.5$ and $\sigma=1.5$ are not reported as they are very similar to the case $\sigma=1$.}
\centering
\fbox{\begin{tabular}{llcccc} 
&   &   &   &   & \\
&  & \multicolumn{4}{c}{average absolute estimation error for search methods}\\
\hline
\em noise level& $n$& \em naive        & \em advanced     & \em combined       & \em full search\\  
\hline
&               100  &   3.38 (7)   & 2.77 (4)     & 2.88 (5)       &  3.24 (5) \\
&               200  &   2.72 (4)   & 4.22 (7)     & 2.95 (5)       &  3.17 (5) \\
&               300  &   3.43 (7)   & 4.45 (8)     & 3.21 (5)       &  3.16 (5) \\
$\sigma=0.5$&   400  &   4.68 (10)  & 3.95 (6)     & 3.37 (5)       &  3.16 (5) \\
&               500  &   6.55 (27)  & 4.24 (8)     & 3.09 (5)       &  3.08 (5) \\
&               1000 &  13.75 (74)  & 3.84 (6)     & 3.35 (5)       &  3.08 (5) \\
&               2000 & 171.74 (387) & 3.92 (7)     & 3.26 (6)       &  3.01 (4) \\
&               5000 &1021.12 (1338)& 3.92 (7)     & 3.52 (6)       &  3.05 (5) \\
\hline
&               100  &  15.86   (20) &  15.26  (23) &   15.07  (21)  &     16.79  (22) \\
&               200  &  12.37   (18) &  28.93  (43) &   15.78  (26)  &     17.44  (28) \\
&               300  &  19.50   (34) &  26.91  (45) &   19.30  (35)  &     17.73  (33) \\
$\sigma=1$&     400  &  30.58   (56) &  26.02  (54) &   20.14  (42)  &     17.85  (37) \\
&               500  &  50.09   (87) &  26.97  (59) &   21.06  (49)  &     18.80  (44) \\
&               1000 & 136.75  (240) &  29.70  (94) &   24.59  (81)  &     21.24  (72) \\
&               2000 & 544.70  (547) &  35.73 (160) &   34.16 (156)  &     24.21 (116) \\
&               5000 &1948.79 (1328) &  48.08 (341) &   51.94 (354)  &     38.34 (298) \\
\hline
&               100  &  25.24   (25)  &   33.95   (35) &  31.70    (32)   &   34.19    (33) \\
&               200  &  23.77   (29)  &   60.82   (62) &  39.03    (50)   &   42.05    (52) \\
&               300  &  41.23   (54)  &   65.17   (82) &  50.79    (72)   &   48.55    (72) \\
$\sigma=1.5$&   400  &  62.98   (85)  &   70.69  (107) &  58.85    (95)   &   56.11    (93) \\
&               500  &  96.54  (114)  &   82.27  (134) &  70.03   (121)   &   62.41   (115) \\
&               1000 & 253.11  (291)  &  121.14  (256) &  114.73  (243)   &   98.52   (226) \\
&               2000 & 739.92  (534)  &  202.01  (504) &  203.74  (493)   &   156.51  (434) \\
&               5000 &2171.28 (1211)  &  436.96 (1269) &  455.99 (1260)   &   355.35 (1123) \\
\hline
\hline
&   &   &   &   & \\
&  & \multicolumn{4}{c}{average number of evaluations for search methods}\\
\hline
\em noise level& $n$& \em naive        & \em advanced     & \em combined       & \em full search\\  
\hline
&            100 & 16.18 (1)  &  25.10 (1)  &  41.28 (2)    &     199 (0) \\
&            200 & 17.31 (1)  &  25.92 (2)  &  43.24 (2)    &     299 (0) \\
$\sigma=1$&  500 & 19.08 (1)  &  29.34 (2)  &  48.43 (2)    &     599 (0) \\
&           1000 & 19.36 (1)  &  30.95 (1)  &  50.31 (2)    &    1099 (0) \\
&           2000 & 21.37 (1)  &  33.00 (1)  &  54.36 (2)    &    2099 (0) \\
&           5000 & 23.69 (1)  &  35.02 (1)  &  58.71 (2)    &    5099 (0) \\
\end{tabular}}
\end{table}

\cref{Fig:single_cpt} shows found change points using various search methods in each 1000 simulations for a balanced ($n=200$) and an unbalanced ($n={5000}$) scenario. The failure of the naive optimistic search in most cases for the unbalanced scenario is again clearly visible, while for the advanced, and in particular the combined optimistic search, the found change points very often lie exactly on the diagonal when compared to the full grid search and hence exactly the candidate proposed by the full grid search were found. 

The simulation results in \cref{tab:simulation_results} and \cref{Fig:single_cpt} confirm our theoretical results that the naive optimistic search is not consistent for very unbalanced signals, while the advanced and combined versions are. In terms of computation, the number of evaluations for optimistic search variants can be orders of magnitude smaller compared to full grid search, in particular if $n$ is large.

\begin{figure}[H]
\centering
\includegraphics[width=\textwidth]
{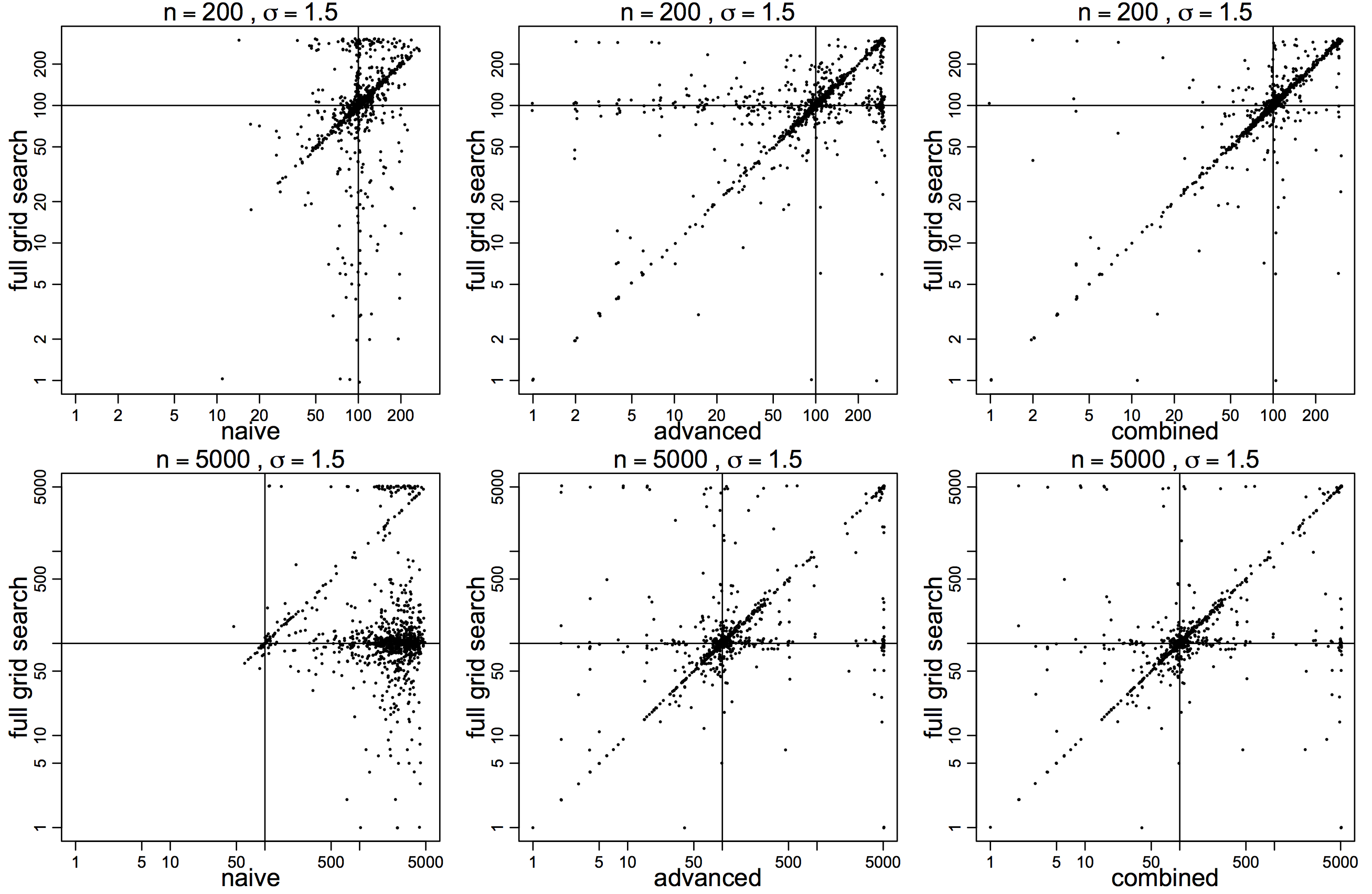} \\
\caption{Pairwise plots of found change points using different optimistic search methods (horizontal-axis) versus the ones returned by the full grid search (vertical-axis) for a noise level $\sigma = 1.5$ and $n=200$ (top) as well as $n={5000}$ (bottom) in $1000$ simulations from \cref{setup:single_mean_shift}. The vertical and horizontal lines indicate the location of the true change point at observation~$100$.} 
\label{Fig:single_cpt}
\end{figure}

\section{Change point detection for high-dimensional Gaussian graphical models}
\label{Appendix_highdim}

In the following, we briefly describe an estimator introduced by \cite{Londschien} for change point detection in high-dimensional Gaussian graphical models, as this is the basis of all change point detection algorithms (BS, SeedBS, WBS and their optimistic variants) that we investigate in \cref{highdim_simulations}. For a segment $(u,v]$ with $0 \leq u < v \leq n$ let $\vec{S}_{(u,v]}$ denote the empirical covariance matrix within that segment. Let $0 < \epsilon < 1/2$ be some required minimal relative segment length and $\gamma>0$ be a regularization parameter. According to a proposal of \cite{Londschien}, for a segment $(u,v]$ with $v-u>2\epsilon n$ we define the split point candidate as
\begin{multline}
    \label{glasso_change_est}
    \hat\eta_{(u,v]} 
= \argmax\limits_{t \in \{ u+\epsilon n,\ldots, v-\epsilon n\} } 
L_n(\hat{\vec{\Omega}}_{(u, v]}^{\mathrm{glasso}}; \vec{S}_{(u, v]}) \\
- \left(L_n(\hat{\vec{\Omega}}_{(u, t]}^{\mathrm{glasso}}; \vec{S}_{(u,t]}) + 
L_n(\hat{\vec{\Omega}}_{(t, v]}^{\mathrm{glasso}}; \vec{S}_{(t, v]}) \right),
\end{multline}
where 
\begin{equation*}
L_n(\vec{\Omega}; \vec{S}_{(u,v]}) =
\frac{v-u}{n}\left(\Tr(\vec{\Omega} \vec{S}_{(u,v]}) - \log(|\vec{\Omega}|) \right)
\end{equation*}
is a multivariate Gaussian log-likelihood based loss in the considered segment $(u,v]$ (scaled according to its length), and $\hat{\vec{\Omega}}_{(u,v]}^{\mathrm{glasso}}$ is the graphical lasso precision matrix estimator \citep{glasso} with a scaled regularization parameter $\sqrt{n/(v-u)}\gamma$, i.e.,
\begin{equation*}
\label{eq:glasso_estimator}
\hat{\vec{\Omega}}_{(u,v]}^{\mathrm{glasso}} = \argmin\limits_{\mathbb{R}^{p \times p} \ni \vec{\Omega} \succ 0} \Tr(\vec{\Omega} \vec{S}_{(u,v]}) - \log(|\vec{\Omega}|) + \sqrt{(n / (v- u))}\gamma \|\vec{\Omega}\|_1.
\end{equation*}

Each split point $t$ in \eqref{glasso_change_est} requires fitting two graphical lasso estimators. While various algorithms for computing exact or approximate solutions for the graphical lasso estimator exist, scalings of $O(p^3)$ (or worse) are common (assuming that the input covariance matrices have been pre-computed and that no special structures such as the block diagonal screening proposed by \cite{Witten2011_connected_components} as well as \cite{Mazumder2012_connected_components} can be exploited). Hence, the graphical lasso is rather costly especially when repeatedly fitting at each possible split point $t$ on the grid $u+\epsilon n,\ldots, v-\epsilon n$. Note that the essential problem is that it is not easy to re-use the estimator $\hat{\vec{\Omega}}_{(u,t]}^{\mathrm{glasso}}$ for the segment $(u, t]$ to obtain $\hat{\vec{\Omega}}_{(u,t+1]}^{\mathrm{glasso}}$ for the segment $(u,t+1]$. One could use $\hat{\vec{\Omega}}_{(u,v]}^{\mathrm{glasso}}$ as a warm start, but not all algorithms that have been developed to compute graphical lasso fits are guaranteed to converge with warm starts (see \citealp{Mazumder2012_dpglasso}) and even the ones that do converge would not save orders of magnitude in terms of runtime. Note that this lack of efficient updates is common for more complex (e.g.~high-dimensional) scenarios and it is in sharp contrast with {e.g.~change point detection in means for  $p$-dimensional Gaussian variables.} There, one needs to calculate means, but the mean for the segment $(u,t+1]$ can be updated in {$O(p)$} cost if the mean for the segment $(u,t]$ is already available, and hence the computational cost is typically proportional to the total length of considered segments. In contrast, in the estimator in~\eqref{glasso_change_est}, the number of graphical lasso fits, as given by the number of considered split points, is the main driver of computational cost. Our optimistic search techniques rely on evaluating far fewer split points $t$ than the full grid search and thus provide an option for massive computational speedups. Of course, the price to pay is having no guarantee to obtain exactly the optimal split point, but the ``optimistic'' approximation to $\hat\eta_{(u,v]}$, i.e., the one obtained via the optimistic search, is still fairly good, see the simulations in \cref{highdim_simulations}. 

In the simulations, we used the \textbf{glasso} \textsf{R} package, available on CRAN, for the graphical lasso fits. For all six methods, we set $\epsilon=0.01$, i.e.\ skipping $20$ observations on the boundaries of each considered interval and overall, no change points were searched in intervals containing less than $60$ observations. We set $\gamma = 0.007$. Regularization in these examples is not essential in the sense that we do not have truly high-dimensional scenarios, but for split points close to the boundaries of the search interval and in short intervals, where the number of observations is close to $p$, regularization can be still helpful. We could have increased $p$ in order to cover truly high-dimensional setups in our simulations, but given the scaling $O(p^3)$ of the graphical lasso, this very quickly goes beyond reasonable computational times for the full grid search based approaches that we want to include as references in terms of achievable estimation error. 

\subsection{Proof of \cref{Lem:GGM_convexity}}
\label{App:GGM_convexity_proof}

First note that in the following we only consider the case $l<\cpt_i < t < \cpt_{i+1}< r$, but similar arguments can be used also in the presence of a single change point in $(l,r]$ or when considering a split point $t$ in the segment from $l$ to the first change point within $(l,r]$. 
Recall that $\vec{\Sigma}_{(l,r]}$ denotes the convex combination of the covariance matrices within the segment $(l,r]\subseteq(0,1]$ with the weights given by the relative segment lengths within $(l,r]$. In particular, for $l<\cpt_i < t < \cpt_{i+1}< r$,
\begin{equation*}
\vec{\Sigma}_{(l,t]} 
=\frac{1}{t-l}
\left(
(\cpt_i - l)\vec{\Sigma}_{(l,\cpt_i]} 
+
(t-\cpt_i)\vec{\Sigma}_{i+1}
\right)
\end{equation*}
and 
\begin{equation*}
\vec{\Sigma}_{(t,r]} 
=\frac{1}{r-t}
\left(
(\cpt_{i+1}-t)\vec{\Sigma}_{i+1} 
+ 
(r - \cpt_{i+1})\vec{\Sigma}_{(\cpt_{i+1}, r]} 
\right),
\end{equation*}
where $\vec{\Sigma}_{i+1} = \vec{\Sigma}_{(\cpt_{i}, \cpt_{i+1}]}$ is the covariance matrix in the $i+1$-st segment $(\cpt_{i}, \cpt_{i+1}]$. We seek to find the first and second derivatives of $G^*_{(l,\,r]}(t)$. First note that
\begin{equation*}
\frac{\partial}{\partial t} \vec{\Sigma}_{(l,t]} 
=
\frac{\cpt_{i} - l}{(t-l)^2}
\left(\vec{\Sigma}_{i+1} - \vec{\Sigma}_{(l,\cpt_i]} \right)
\end{equation*}
and
\begin{equation*}
\frac{\partial^2}{\partial t^2} \vec{\Sigma}_{(l,t]} 
=
-\frac{2(\cpt_{i} - l)}{(t-l)^3}
\left(\vec{\Sigma}_{i+1} - \vec{\Sigma}_{(l,\cpt_i]} \right)
=
-\frac{2}{t-l}\frac{\partial}{\partial t} \vec{\Sigma}_{(l,t]}.
\end{equation*}
We need the following expressions (see e.g.~\citealp{matrix_cookbook}) for derivatives of an invertible matrix $A(s)$ depending on $s$,
\begin{equation*}
\frac{\partial}{\partial t} \log(|\vec{A}(t)|) = \Tr\left(\vec{A}(t)^{-1} \cdot \frac{\partial}{\partial t} \vec{A}(t)\right);
\end{equation*}
\begin{equation*}
\frac{\partial}{\partial t} \Tr(\vec{A}(t)) = \Tr\left( \frac{\partial}{\partial t} \vec{A}(t)\right);
\end{equation*}
\begin{equation*}
\frac{\partial}{\partial t} \vec{A}(t)^{-1} = -\vec{A}(t)^{-1} \cdot \frac{\partial}{\partial t} \vec{A}(t)  \cdot \vec{A}(t)^{-1}.
\end{equation*}
We next compute the first and second derivatives of $G^{*}_{(l,r]}(t)$. Recall from equation~\eqref{e:gainHD} that
\begin{equation*}
    G^*_{(l,r]}(t)
    =
    \frac{r-l}{n}\log(|\vec{\Sigma}_{(l,r]}|) 
    -\frac{t-l}{n}\log(|\vec{\Sigma}_{(l,t]}|)
    -\frac{r-t}{n}\log(|\vec{\Sigma}_{(t,r]}|)\,.
\end{equation*}
Consider first the middle part of $G^{*}_{(l,r]}(t)$, i.e.\ 
$$
L^*(t) := -\frac{t-l}{n}\log(\abs{\vec{\Sigma}_{(l,t]}})\,.
$$
Then for the first derivative
$$
\frac{\dif}{\dif t}L^{*}(t)
=
-\frac{1}{n}\log(|\vec{\Sigma}_{(l,t]}|)
-\frac{t-l}{n} 
\Tr\left( \vec{\Sigma}_{(l,t]}^{-1} \cdot 
\frac{\partial}{\partial t} \vec{\Sigma}_{(l,t]} 
\right),
$$
and for the second derivative
\begin{align*}
\frac{\dif^2}{\dif t^2}L^{*}(t) = &
-\frac{1}{n}\Tr\left( \vec{\Sigma}_{(l,t]}^{-1} \cdot 
\frac{\partial}{\partial t} \vec{\Sigma}_{(l,t]} 
\right)
-\frac{1}{n} 
\Tr\left( \vec{\Sigma}_{(l,t]}^{-1} \cdot 
\frac{\partial}{\partial t} \vec{\Sigma}_{(l,t]} 
\right) \\
& \mbox{}\qquad{}\qquad{}\qquad{}
-\frac{t-l}{n} 
\frac{\partial}{\partial t}
\left(
\Tr\left( \vec{\Sigma}_{(l,t]}^{-1} \cdot 
\frac{\partial}{\partial t} \vec{\Sigma}_{(l,t]} 
\right)
\right)\\
= &
-\frac{2}{n}\Tr\left( \vec{\Sigma}_{(l,t]}^{-1} \cdot 
\frac{\partial}{\partial t} \vec{\Sigma}_{(l,t]} 
\right)
-\frac{t-l}{n} 
\Tr\left(
-\vec{\Sigma}_{(l,t]}^{-1} \cdot \frac{\partial}{\partial t} \vec{\Sigma}_{(l,t]}  \cdot \vec{\Sigma}_{(l,t]}^{-1}
\cdot 
\frac{\partial}{\partial t} \vec{\Sigma}_{(l,t]}
\right)\\
& \mbox{}\qquad{}\qquad{}\qquad{}
-\frac{t-l}{n} 
\Tr\left( \vec{\Sigma}_{(l,t]}^{-1} \cdot 
\frac{\partial^2}{\partial t^2} \vec{\Sigma}_{(l,t]}
\right)\\
= &
\frac{t-l}{n} 
\Tr\left(
\vec{\Sigma}_{(l,t]}^{-1} \cdot \frac{\partial}{\partial t} \vec{\Sigma}_{(l,t]}  \cdot \vec{\Sigma}_{(l,t]}^{-1}
\cdot 
\frac{\partial}{\partial t} \vec{\Sigma}_{(l,t]}
\right)
\\
= & \frac{t-l}{n} 
\left\|{\vec{\Sigma}_{(l,t]}^{-1/2} \cdot \frac{\partial}{\partial t} \vec{\Sigma}_{(l,t]}  \cdot \vec{\Sigma}_{(l,t]}^{-1/2}}\right\|_{F}^2 \ge 0.
\end{align*}
By symmetry, we can obtain similarly for the right part of $G^{*}_{(l,r]}(t)$,
$$
\frac{\partial^2}{\partial t^2}\left(-\frac{r-t}{n}\log(|\vec{\Sigma}_{(t,r]}|)\right) = \frac{r-t}{n} 
\left\|{\vec{\Sigma}_{(t,r]}^{-1/2} \cdot \frac{\partial}{\partial t} \vec{\Sigma}_{(t,r]}  \cdot \vec{\Sigma}_{(t,r]}^{-1/2}}\right\|_{F}^2 \ge 0.
$$
As the left part of $G^{*}_{(l,r]}(t)$ is constant, we have 
$G^{*''}_{(l,r]}(t)\geq 0$ for $\cpt_i < t < \cpt_{i+1}$ and hence, $G^{*}_{(l,r]}$ is convex in between change points $\cpt_i$ and $\cpt_{i+1}$. Moreover, we see that with the exception of the special cases when $\vec{\Sigma}_{(l,\cpt_i]} = \vec{\Sigma}_{i+1} = \vec{\Sigma}_{(\cpt_{i+1},t]}$, $G^{*}_{(l,r]}$ is even strictly convex in the interval $(\cpt_{i}, \cpt_{i+1}]$. Note that for such special cases $G^{*}_{(l,r]}(t)=0$ for arbitrary $t \in (\cpt_i, \cpt_{i+1}]$, i.e.~the population gain function is flat in between change points $\cpt_i$ and $\cpt_{i+1}$. 
Note that special cases can only occur in the presence of two or more change points within the considered segment. In particular, in case $\cpt_i$ is the single change point contained in $(l,r]$, $G^{*}_{(l,r]}$ is strictly convex in $(l,\cpt_i]$ and strictly convex in $(\cpt_i,r]$. \hfill \qedsymbol

\section{Comments on computational gains for high-dimensional simulations}
\label{Appendix_computational_gains_highdim}

The achievable speedups using optimistic search in general are dependent on the cost of the model fit in each segment (how they depend on the number of observations $n$ and the dimensionality $p$), whether there are possibilities to update neighboring fits efficiently, but also on the length of the series, the number of change points, which basic algorithm (BS, SeedBS, WBS or yet another one) is used with which specific tuning parameters, etc. Nonetheless, we would like to further comment on some of the observed computational gains in the high-dimensional simulations presented in \cref{highdim_simulations} in the main text.

The biggest computational gains for optimistic search occur when the underlying search intervals are long. Random intervals have expected length $O(n)$ and thus many of them are comparably long. For these long intervals we gain a lot by optimistic searches. However, the lengths of the intervals in lower layers of seeded intervals are quite short (decaying exponentially) and what becomes dominant in that case is the number of very short intervals. For example, while there is only a single interval containing $2,000$ observations (first layer), there were more than sixty intervals on the lowest layer we considered with the minimally required segment length of $m=60$ observations. This explains why the speedup for OSeedBS versus SeedBS is a factor 10--14, while for BS and WBS we could achieve factor $30$ or more. Skipping the last few layers of seeded intervals would have saved considerable computational time for OSeedBS, which is important to keep in mind when interpreting the results from \cref{Fig:glasso_multiple_changes} in the main text. From a practical perspective, when utilizing OSeedBS, one should thus limit the number of covered layers in order to consider fewer of the very short intervals that are a driver of computational cost. However, the minimal segment length in seeded intervals cannot be too large either, as in that case one is risking not covering each single change points sufficiently (similar to what happened in the shown examples for WBS and OWBS with a small number of random intervals $M$). The choice for the minimal segment length for seeded intervals might come fairly naturally in some applications, where segments below a certain size are uninteresting or when considering high-dimensional problems requiring a minimal number of observations for fitting reasonable models. 

A pragmatic approach could be to combine the best of both worlds from OBS and OSeedBS. For example, find a first set of change points with fewer number of seeded intervals and then, to protect against the possibility that there could be even further change points that were not discovered due to having chosen a too large minimal segment length, in between the first found change points from the seeded intervals, one could perform a further OBS-like search that adapts better to the number of change points within these shorter search intervals. This way adaptively one could invest more computational effort if there is evidence for further change points beyond the ones found by the rough first set of seeded intervals, but without the need to go over each and every very short interval as would be the case with further layers of seeded intervals containing very short intervals. Thus, one could keep computational advantages from OBS and at the same time exploit the better expected statistical performance of OSeedBS.

%% file: sections/proofs2.tex
\allowdisplaybreaks

\section{Proofs of statistical guarantees}\label{App:proofs}

Here we provide proofs of statistical guarantees for optimistic searches in terms of consistency and localization rates. For ease of reading, we rewrite \cref{m:mhgauss} (i.e.\ a $p$-dimensional version of \cref{Gaussian_setup}) as
\begin{equation}\label{e:addform}
\vec{X}_t = \vec{f}_t + \vec{\xi}_t \qquad \text{for}\quad t = 1, \ldots, n,
\end{equation}
where $\vec{f}_t \equiv \vec{f}(t/n)$ with $\vec{f}:(0,1] \to \R^p$ defined as $\vec{f}(x) := \sum_{i = 0}^\ncp \vec{\mu}_i \ind_{(\cpt_i, \cpt_{i+1}]}(x)$, and $\vec{\xi}_t \iidsim \mathcal{N}(0,\vec{I}_p)$, i.e.~independently, standard $p$-dimensional Gaussian distributed. Let $\vec{X} := (\vec{X}_1,\ldots, \vec{X}_n)^{\top} \in \R^{n \times p}$ be the data matrix, $\vec{F} := (\vec{f}_1, \ldots, \vec{f}_n)^{\top} \in \R^{n\times p}$ the signal matrix and $\vec{\Xi}:= (\vec{\xi}_1, \ldots, \vec{\xi}_n)^{\top}\in \R^{n\times p}$ the noise matrix. Then, in an equivalent matrix form of \eqref{e:addform}, it holds that $\vec{X}  =  \vec{F} + \vec{\Xi}$. Similarly, we denote another independent sample as $\vec{Y}  =  \vec{F} + \tilde{\vec{\Xi}}$ (which is needed for sample splitting, see \cref{ss:mats}). 

Towards a matrix-vector formulation of CUSUM statistics, we introduce
$$
\vec{e}_t := \Biggl(\underbrace{\sqrt{\frac{n-t}{nt}},\cdots, \sqrt{\frac{n-t}{nt}}}_{t}, \underbrace{-\sqrt{\frac{t}{n(n-t)}},\cdots, -\sqrt{\frac{t}{n(n-t)}}}_{n-t}\Biggr)^\top \in \R^n,
$$
for $t \in\set{1,\ldots,n-1}$, and $\vec{e}_n :=(1/\sqrt{n}, \ldots, 1/\sqrt{n})^\top \in \R^n$.  Then
$$
\inner{\vec{e}_t}{\vec{e}_n} = 0,\qquad \norm{\vec{e}_t} = \norm{\vec{e}_n} = 1,\qquad t = 1, \ldots, n-1,
$$
where $\inner{\cdot}{\cdot}$ and $\norm{\cdot}$ are the inner product and the norm, respectively, in Euclidean spaces. Further notation is as follows. Let  $\norm{\cdot}_\infty$ denote the supremum norm of vectors.  For real numbers $a, b$, let $a \vee b := \max\set{a,b}$ and $a \wedge b :=\min\set{a,b}$. Let also $\Phi(\cdot)$ be the distribution function of a standard Gaussian random variable. 

In all the proofs, we try to give constants as explicitly as possible, but those constants may not be the best ones. The limiting behaviour is considered as the sample size $n\to \infty$, and the involved quantities, including the sparsity level $s$, the dimension $p$, the underlying signal $\vec{f}$, and thus the minimal segment length $\lambda$ and the minimal jump size $\delta$, are allowed to depend on $n$. 

\subsection{Technical tools}
We need several deviation bounds on chi-squares related quantities. 

\begin{lemma}[Tail of chi-squares]\label{lm:tail.chi2}
Let $Y := \sum_{i = 1}^k w_i (X_i + \mu_i)^2$, with $X_i \iidsim \gauss(0, 1)$, and constants $w_i \ge 0$ and $\mu_i \in \R$. Then
$$
\E{Y} = \sum_{i=1}^k w_i(1+ \mu_i^2)\quad \text{and}\quad \var{Y} = 2 \sum_{i = 1}^k w_i^2(1 + 2 \mu_i^2).
$$
Further, for every $x \ge 0$, it holds
\begin{align*}
\Prob{Y \le \E{Y} - x} &\le \exp\left(-\frac{x^2}{2\var{Y}}\right),\\
\text{and }\qquad \Prob{Y \ge \E{Y} + x} &\le \exp\left(-\frac{\var{Y}}{4\norm{\vec{w}}_\infty^2} \psi\Bigl(\frac{2\norm{x\vec{w}}_\infty}{\var{Y}}\Bigr)\right) \\
&\le \exp\left(-\frac{x^2}{2\var{Y} + 4x\norm{\vec{w}}_\infty}\right),
\end{align*}
where $\psi(x) = 1 + x -\sqrt{1+2x}$.
\end{lemma}

\begin{proof}
The expectation and variance are easy to compute. Note that 
$$
\psi(x) = \frac{x^2}{1+x+\sqrt{1+2x}} \ge \frac{x^2}{2+2x}. 
$$
Then the second assertion is a reformulation of \cite[Lemma 2]{LLM12} or the Hanson--Wright inequality.
\end{proof}

\begin{lemma}[Tail of Bernoulli weighted chi-squares]\label{lm:bwcs}
Let $X_i \iidsim \gauss(0,1)$, $B_i \stackrel{\text{ind.}}{\sim} \ber(\beta_i)$, with $0 \le \beta_i \le 1$, $i = 1, \ldots, k$, and $(X_i)_{i = 1}^k$ and $(B_i)_{i =1}^k$ be independent. Let also
$$
Y:=\sum_{i = 1}^k\left( w_i (X_i + \mu_i)^2 - w_i(1 + \mu_i^2) \right)B_i
$$
with constants $w_i \ge 0$, $\mu_i \in \R$. Then
$$
\E{Y} = 0 \quad \text{and}\quad \var{Y} = 2\sum_{i=1}^k\beta_i w_i^2(1+2\mu_i^2),
$$
and, for every $x \ge 0$, it holds
\begin{subequations}\label{e:gbBwC}
\begin{align}
\Prob{Y \ge x}  & \le \label{e:gbBwCa}\\
\exp&\left(-\min\left\{\frac{x}{8 \max_i \left(w_i({1+2\mu_i^2})^{1/2}\right)},\;\frac{ x^2}{6\var{Y}}\right\}\right), \nonumber\\
\Prob{Y \le -x}  & \le \label{e:gbBwCb}\\
\exp&\left(-\min\left\{\frac{x}{2 \max_i \left(w_i({1+2\mu_i^2})^{1/2}\right)},\;\frac{ x^2}{4\var{Y}}\right\}\right).\nonumber
\end{align}
\end{subequations}
Further, if $\min_{1\le i\le k}\beta_i \ge 1/2$, then for every $x \ge 0$,
\begin{subequations}\label{e:bBwC}
\begin{align}
\Prob{Y \ge x}  & \le \exp\left(-\min\left\{ \frac{x}{8\norm{\vec{w}}_{\infty}},\,\frac{x^2}{8\var{Y}}\right\}\right),\\
\Prob{Y \le -x}  & \le  \exp\left(-\frac{x^2}{4\var{Y}}\right).
\end{align}
\end{subequations}

\end{lemma}

\begin{proof}
We introduce the shorthand notation 
$$Z_i :=  w_i (X_i + \mu_i)^2 - w_i(1 + \mu_i^2).
$$
Then $Y = \sum_{i = 1}^k Z_i B_i$, and $\E{Y} = \sum_{i=1}^k \E{Z_i}\Prob{B_i = 1} = 0$ and
$$
\var{Y} = \sum_{i=1}^k \E{Z_i^2} \Prob{B_i = 1} =  2\sum_{i=1}^k\beta_iw_i^2(1+2\mu_i^2).
$$

Note that, for $-\infty < a<1/(2w_i)$,
$$
\E{\exp(a Z_i)} = \exp\left(\frac{2a^2w_i^2\mu_i^2}{1-2aw_i} - aw_i - \frac{1}{2}\log(1-2aw_i)\right),
$$
and also that 
\begin{equation}\label{e:logb}
-x - \frac{1}{2}\log(1-2x) \le 
\begin{cases}
\frac{x^2}{1-2x} & \text{ if } 0 \le x < 1/2,\\
x^2 & \text{ if } x \le 0.
\end{cases}
\end{equation}
We consider two separate cases:
\begin{itemize}
\item The case of general $\beta_i\in [0,1]$. For $0 \le x \le (1+2\mu^2)^{-1/2}/4$, it holds 
\begin{align*} 
\exp\left(\frac{2x^2\mu^2}{1-2x} - x - \frac{1}{2}\log(1-2x)\right) -1 & \le \exp\left(\frac{(1+2\mu^2)x^2}{1-2x} \right) -1\\ 
& \le \exp\left(2{(1+2\mu^2)x^2}\right) -1\\
& \le 3(1+2\mu^2)x^2.
\end{align*}
Thus, for $a$ such that  $0 \le  {4 a \max_i\left( w_i({1+2\mu_i^2})^{1/2}\right)} \le 1$, we obtain 
\begin{align*}
\log(\E{\exp(aY)})& = \sum_{i = 1}^k \log\bigl(\beta_i \E{\exp(a Z_i)} + 1-\beta_i\bigr) \\
& \le \sum_{i = 1}^k \log\bigl(1 + 3 \beta_i (1+2\mu_i^2)a^2 w_i^2\bigr) \\
& \le \sum_{i = 1}^k 3\beta_i (1+2\mu_i^2)a^2 w_i^2.
\end{align*}
By the Chernoff bound, we obtain
\begin{align*}
& \log \left(\Prob{Y \ge x }\right) \\
\le\; & \inf_{0 \le  {4 a \max_i\left( w_i({1+2\mu_i^2})^{1/2}\right)} \le 1} \left(-ax + \sum_{i=1}^k3\beta_i (1+2\mu_i^2)a^2 w_i^2\right)\\
\le\; & -\min\left\{\frac{x}{8 \max_i \left(w_i({1+2\mu_i^2})^{1/2}\right)},\;\frac{x^2}{12\sum_{i=1}^k\beta_i(1+2\mu_i^2)w_i^2}\right\}.
\end{align*}
Similarly, we have, for $-1\le a \max_i \left( w_i (1+2\mu_i^2)^{1/2} \right)\le 0$,
$$
\log(\E{\exp(aY)}) \le \sum_{i = 1}^k \log\bigl(1 + 2 \beta_i (1+2\mu_i^2)a^2 w_i^2\bigr)  \le \sum_{i = 1}^k 2\beta_i (1+2\mu_i^2)a^2 w_i^2,
$$
and
\begin{align*}
&\log\left(\Prob{Y \le -x }\right) \\
 \le\;& \inf_{-1\le a \max_i \left( w_i (1+2\mu_i^2)^{1/2} \right)\le 0} \left(ax + \sum_{i=1}^k2\beta_i (1+2\mu_i^2)a^2 w_i^2\right)\\
  \le\;& -\min\left\{\frac{x}{2 \max_i \left(w_i({1+2\mu_i^2})^{1/2}\right)},\;\frac{x^2}{8\sum_{i=1}^k\beta_i(1+2\mu_i^2)w_i^2}\right\}.
\end{align*}

\item The case of $\beta_i \ge 1/2$ for all $i\in \{1,\ldots,k\}$. For $0 \le x < 1/2 \le \beta \le 1$,
$$
1+ \beta\exp\left(\frac{2x^2\mu^2}{1-2x} - x - \frac{1}{2}\log(1-2x)\right) -\beta \le \exp\left(\frac{2\beta x^2(1+2\mu^2)}{1-2x}\right).
$$
This implies, for $0 \le 2a\norm{\vec{w}}_{\infty} <  1$,
$$
\log(\E{\exp(aY)}) = \sum_{i = 1}^k \log\bigl(\beta_i \E{\exp(a Z_i)} + 1-\beta_i\bigr) \le \sum_{i = 1}^k \frac{2\beta_i a^2w_i^2(1+2\mu_i^2)}{1-2aw_i}.
$$
Then, again by the Chernoff bound, we have 
\begin{align*}
\log\left(\Prob{Y \ge x }\right) & \le \inf_{0 \le   2a\norm{\vec{w}}_\infty <1} \left(-ax + \sum_{i=1}^k\frac{2\beta_i a^2w_i^2(1+2\mu_i^2)}{1-2aw_i}\right)\\
& \le \inf_{0 \le   4a\norm{\vec{w}}_\infty <1} \exp\left(-ax + \sum_{i=1}^k{4\beta_i a^2w_i^2(1+2\mu_i^2)}\right)\\
& \le -\min\left\{ \frac{x}{8\norm{\vec{w}}_{\infty}},\,\frac{x^2}{16\sum_{i=1}^k\beta_i(1+2\mu_i^2)w_i^2}\right\}.
\end{align*}
Similarly, for $a \le 0$, we have
\begin{align*}
\log(\E{\exp(aY)}) &= \sum_{i = 1}^k \log\bigl(\beta_i \E{\exp(a Z_i)} + 1-\beta_i\bigr) \\
& \le  \sum_{i = 1}^k \log\Bigl(\beta_i \exp\left(2\beta_i a^2 w_i^2(1+2\mu_i^2)\right) + 1-\beta_i\Bigr) \\
& \le \sum_{i = 1}^k 2\beta_i a^2 w_i^2(1+2\mu_i^2),
\end{align*}
and 
\begin{align*}
\log\left(\Prob{Y \le -x }\right) & \le \inf_{a \le 0} \left(ax + \sum_{i=1}^k2\beta_i (1+2\mu_i^2)a^2 w_i^2\right)\\
& = -\frac{x^2}{8\sum_{i=1}^k\beta_i (1+2\mu_i^2) w_i^2}.
\end{align*}
\end{itemize}
Therefore, the assertions follow, as $\var{Y} = {2\sum_{i=1}^k\beta_i (1+2\mu_i^2) w_i^2}$.
\end{proof}

\begin{remark}
In comparison with \cref{lm:tail.chi2}, there is an additional term $\max_i \left(w_i({1+2\mu_i^2})^{1/2}\right)$ in the bound of lower tail probability, see \eqref{e:gbBwC}, when there are Bernoulli weights. We stress that such a term is necessary, especially when $\beta_i \searrow 0$ for every $i\in\{1,\ldots,k\}$. However, in case of $\min_i \beta_i \ge 1/2$, $Y$ behaves the same as if there are no Bernoulli weights (i.e.\ $\beta_1 = \cdots = \beta_p = 1$), see \eqref{e:bBwC}. In particular, up to difference in constants, \cref{lm:bwcs} includes \cref{lm:tail.chi2} as a special case. 
\end{remark}

\begin{lemma}[Lower tail of Bernoulli weighted non-central chi-squares]\label{lm:bwcs2}
Let $X_i \iidsim \gauss(0,1)$, $B_i \stackrel{\text{ind.}}{\sim} \ber(\beta_i)$, with $0 \le \beta_i \le 1$, $i = 1, \ldots, k$, and $(X_i)_{i = 1}^k$ and $(B_i)_{i =1}^k$ be independent. Let also
$$
Y:=\sum_{i = 1}^k B_iw_i (X_i + \mu_i)^2 
$$
with constants $w_i \ge 0$, $\mu_i \in \R$. Assume that, for some constant $C > 0$, 
\begin{equation}\label{e:beta1ass}
1-\beta_i\le \exp\bigl(-Cw_i(1+\mu_i^2)\bigr),\quad i = 1,\ldots, k.
\end{equation}
Then, for every $x \ge 0$, it holds that
\begin{align*}
&\Prob{Y -\sum_{i=1}^kw_i(1+\mu_i^2) \le -x} \\
\le\;&\exp\left(-\min\left\{\frac{Cx}{4},\,\frac{x^2}{4\sum_{i=1}^k w_i^2(1+2\mu_i^2)}\right\}\right)\cdot \prod_{i=1}^k\left( 1+\exp\Bigl(-\frac{Cw_i (1+\mu_i^2)}{2}\Bigr)\right)\\
\le\;& 2^k\exp\left(-\min\left\{\frac{Cx}{4},\,\frac{x^2}{4\sum_{i=1}^k w_i^2(1+2\mu_i^2)}\right\}\right).
\end{align*}
\end{lemma}

\begin{proof}
Note first that, for $a \le 0$,
\begin{align*}
&\E{\exp\left(aY -a\sum_{i=1}^k w_i (1+\mu_i^2)\right)} \\
=\,& \prod_{i=1}^k \Bigl(\beta_i \E{\exp\bigl(a w_i(X_i+\mu_i)^2 - a w_i (1+\mu_i^2)\bigr)}+(1-\beta_i) \exp\bigl(- a w_i (1+\mu_i^2)\bigr)\Bigr)\\
=\, & \prod_{i=1}^k \left(\beta_i \exp\left(\frac{2a^2w_i^2\mu_i^2}{1-2aw_i} - aw_i - \frac{1}{2}\log(1-2aw_i)\right)+(1-\beta_i) \exp\bigl(- a w_i (1+\mu_i^2)\bigr)\right) \\
\le \, & \prod_{i=1}^k \Bigl( \exp\left(a^2w_i^2(1+2\mu_i^2)\right)+\exp\bigl(-(a +C)w_i (1+\mu_i^2)\bigr)\Bigr),
\end{align*}
where the last inequality is due to \eqref{e:logb} and the assumption in \eqref{e:beta1ass}. Then, we apply the Chernoff bound and obtain
\begin{align*}
&\Prob{Y -\sum_{i=1}^kw_i(1+\mu_i^2) \le -x} \\
\le\, & \inf_{a \le 0}\E{\exp\left(aY -a\sum_{i=1}^k w_i (1+\mu_i^2) +ax \right)}\\
\le\,& \inf_{a \le 0} e^{ax}\prod_{i=1}^k \Bigl( \exp\left(a^2w_i^2(1+2\mu_i^2)\right)+\exp\bigl(-(a +C)w_i (1+\mu_i^2)\bigr)\Bigr)\\
\le \, & \inf_{-C/2 \le a \le 0} e^{ax}\prod_{i=1}^k \left( \exp\left(a^2w_i^2(1+2\mu_i^2)\right)+\exp\Bigl(-\frac{Cw_i (1+\mu_i^2)}{2}\Bigr)\right)\\
\le \, & \inf_{-C/2 \le a \le 0} e^{ax}\prod_{i=1}^k\left( 1+\exp\Bigl(-\frac{Cw_i (1+\mu_i^2)}{2}\Bigr)\right) \exp\left(a^2w_i^2(1+2\mu_i^2)\right)\\
\le\, & \inf_{-C/2 \le a \le 0} \exp\left(ax +a^2\sum_{i=1}^kw_i^2(1+2\mu_i^2)\right)\cdot\prod_{i=1}^k\left( 1+\exp\Bigl(-\frac{Cw_i (1+\mu_i^2)}{2}\Bigr)\right) \\
\le \, &  \exp\left(-\min\left\{\frac{Cx}{4},\frac{x^2}{4\sum_{i=1}^kw_i^2(1+2\mu_i^2)}\right\}\right)\cdot \prod_{i=1}^k\left( 1+\exp\Bigl(-\frac{Cw_i (1+\mu_i^2)}{2}\Bigr)\right),
\end{align*}
which concludes the proof, since $\exp\bigl(-{Cw_i (1+\mu_i^2)}/{2}\bigr)\le 1$.
\end{proof}

\begin{remark}[Upper tail]\label{r:utb}
We stress that the bound of the upper tail probability of $Y$ follows readily from \cref{lm:tail.chi2}, since
$$
\Prob{Y -\sum_{i=1}^kw_i(1+\mu_i^2)\ge x}  \le \Prob{\sum_{i = 1}^k w_i (X_i + \mu_i)^2- \sum_{i=1}^k w_i(1+\mu_i^2)\ge x}.
$$
It is a bit unusual that the concentration inequalities here are centered at $\sum_{i=1}^k w_i(1+\mu_i^2)$ rather than $\E{Y} =\sum_{i=1}^k\beta_i w_i(1+\mu_i^2)$, but this makes little difference as $\beta_i$'s are fairly close to one, which is assumed in \eqref{e:beta1ass}. The current version is chosen in order to ease our later proofs. 
\end{remark}

We consider next some concentration inequalities on the difference of (Bernoulli weighted) chi-squares. The crucial part is to decouple the possible correlation between the involved chi-squares. 

\begin{lemma}[Tail of difference of chi-squares]\label{lm:diff}
Let $\vec{X}_j \sim \gauss(\vec{\mu}_j, \vec{I}_n)$ be independent, with $\vec{\mu}_j\in \R^n$, $j =1,\ldots, k$, and $t,r\in \set{1, \ldots, n}$ be arbitrary. Define the \emph{relative difference} of $t,r$ within the background $(0,n]$ as
\begin{equation}\label{e:reld}
d(t,r) := \frac{\abs{t- r}}{\min\set{r\vee t,\; n - (r \wedge t)}}, 
\end{equation}
which is always in $[0,\;1]$. Then, for every $x \ge 0$, it holds that 
\begin{multline*}
\Prob{\sum_{j=1}^k\inner{\vec{e}_t}{\vec{X}_j}^2-\sum_{j=1}^k\inner{\vec{e}_r}{\vec{X}_j}^2 \le \sum_{j=1}^k\inner{\vec{e}_t}{\vec{\mu}_j}^2-\sum_{j=1}^k\inner{\vec{e}_r}{\vec{\mu}_j}^2  - x}\\ 
\le 2\exp\left(-\min\set{\frac{x^2}{32 \bigl(2d(t,r)\wedge 1\bigr)  \left(k + 2 \sum_{j=1}^k\snorm{\vec{\mu}_j}^2\right)},\;\frac{x}{16\sqrt{2d(t,r)\wedge 1}}}\right).
\end{multline*}
\end{lemma}

\begin{proof}
We first compute the eigendecomposition
$$
\vec{e}_t\vec{e}_t^\top - \vec{e}_r\vec{e}_r^\top = 
\begin{pmatrix}
\vec{u}_+ & \vec{u}_-
\end{pmatrix}
\begin{pmatrix}
\theta & 0 \\
0 & - \theta
\end{pmatrix}
\begin{pmatrix}
\vec{u}_+ & \vec{u}_-
\end{pmatrix}^\top  = \theta \vec{u}_+\vec{u}_+^\top - \theta \vec{u}_-\vec{u}_-^\top,
$$
with $\theta = \sqrt{1- \inner{\vec{e}_t}{\vec{e}_r}^2} \ge 0$, $\norm{\vec{u}_+} = \norm{\vec{u}_-} = 1$ and $\inner{\vec{u}_+}{\vec{u}_-} = 0$. More precisely, 
\begin{align*}
&\vec{u}_+ = \frac{1}{2\sqrt{1 + \inner{\vec{e}_t}{\vec{e}_r}}}(\vec{e}_t + \vec{e}_r) + \frac{1}{2\sqrt{1 - \inner{\vec{e}_t}{\vec{e}_r}}}(\vec{e}_t - \vec{e}_r), \\
\text{and}\quad &\vec{u}_- = \frac{1}{2\sqrt{1 + \inner{\vec{e}_t}{\vec{e}_r}}}(\vec{e}_t + \vec{e}_r) - \frac{1}{2\sqrt{1 - \inner{\vec{e}_t}{\vec{e}_r}}}(\vec{e}_t - \vec{e}_r).
\end{align*}
Elementary calculation shows that 
$$
d(t,r) \le \theta^2 = \frac{n \abs{t - r}}{(r\vee t) \bigl(n - (r \wedge t)\bigr)} \le \min\set{2 d(t,r),\;1}.
$$
Define $\xi_+ := \sum_{j =1}^k \inner{\vec{u}_+}{\vec{X}_j}^2$ and $\xi_- := \sum_{j =1}^k \inner{\vec{u}_-}{\vec{X}_j}^2$. Then $\xi_+$ and $\xi_-$ are independent, because $\vec{u}_+^\top \bigl(\vec{X}_1,\ldots,\vec{X}_k\bigr)$ and $\vec{u}_-^\top \bigl(\vec{X}_1,\ldots,\vec{X}_k\bigr)$ are jointly Gaussian and uncorrelated, thus independent.

Note that 
\begin{align*}
& \left(\sum_{j=1}^k\inner{\vec{e}_t}{\vec{X}_j}^2-\sum_{j=1}^k\inner{\vec{e}_r}{\vec{X}_j}^2 \right)-\left( \sum_{j=1}^k\inner{\vec{e}_t}{\vec{\mu}_j}^2-\sum_{j=1}^k\inner{\vec{e}_r}{\vec{\mu}_j}^2 \right)\\
= \; & \sum_{j = 1}^k \vec{X}_j^\top (\vec{e}_t\vec{e}_t^\top - \vec{e}_r\vec{e}_r^\top) \vec{X}_j - \sum_{j = 1}^k \vec{\mu}_j^\top (\vec{e}_t\vec{e}_t^\top - \vec{e}_r\vec{e}_r^\top) \vec{\mu}_j \\
= \; & \theta (\xi_+ -  \xi_-) - \theta (\E{\xi_+} - \E{\xi_-}).
\end{align*}
Thus, we have
\begin{align}
& \Prob{ \theta (\xi_+ -  \xi_-) \le \theta (\E{\xi_+} - \E{\xi_-}) -x}\nonumber\\
\le\; & \Prob{ \xi_+  \le \E{\xi_+} - \frac{x}{2\theta}} + \Prob{ \xi_-  \ge \E{\xi_-} + \frac{x}{2\theta}}.\label{e:bnd}
\end{align}
By \cref{lm:tail.chi2}, we obtain, for the first term in \eqref{e:bnd}, 
\begin{align*}
\Prob{ \xi_+  \le \E{\xi_+} - \frac{x}{2\theta}} \; & \le \exp\left(-\frac{x^2}{8\theta^2\var{\xi_+}}\right)  \\
& = \exp\left(- \frac{x^2}{16 \theta^2\left(k + 2 \sum_{j=1}^k\inner{\vec{u}_+}{\vec{\mu}_j}^2\right)}\right) \\
& \le \exp\left(- \frac{x^2}{16 \theta^2\left(k + 2 \sum_{j=1}^k\snorm{\vec{\mu}_j}^2\right)}\right)\\
& \le  \exp\left(- \frac{x^2}{16\bigl(2d(t,r)\wedge 1\bigr) \left(k + 2 \sum_{j=1}^k\snorm{\vec{\mu}_j}^2\right)}\right),
\end{align*}
and, similarly, for the second term in \eqref{e:bnd}, 
\begin{align*}
& \Prob{ \xi_-  \ge \E{\xi_-} + \frac{x}{2\theta}} \\
\le\; &  \exp\left(-\frac{x^2}{8 \theta^2 \var{\xi_-} + 8 \theta x}\right)\\
\le\; &  \exp\left(-\frac{x^2}{16 \theta^2  \left(k + 2 \sum_{j=1}^k\snorm{\vec{\mu}_j}^2\right) + 8 \theta x}\right)\\
\le\; & \exp\left(-\min\set{\frac{x^2}{32 \bigl(2d(t,r)\wedge 1\bigr)  \left(k + 2 \sum_{j=1}^k\snorm{\vec{\mu}_j}^2\right)},\;\frac{x}{16\sqrt{2d(t,r)\wedge 1}}}\right).
\end{align*}

These two upper bounds, together with \eqref{e:bnd}, conclude the proof.
\end{proof}

\begin{remark}[Centering]\label{r:mean0}
Note that $\vec{u}_+, \vec{u}_- \in \mathrm{span}\set{\vec{e}_t,\vec{e}_r} \subseteq \bigl(\mathrm{span}\set{\vec{e}_n}\bigr)^\perp$. Then, it holds
$$
\inner{\vec{u}_+}{\vec{\mu}_j}^2 = \inner{\vec{u}_+}{(\vec{I}_n - \vec{e}_n\vec{e}_n^\top)\vec{\mu}_j}^2 \le \norm{(\vec{I}_n - \vec{e}_n\vec{e}_n^\top)\vec{\mu}_j}^2
$$
and a similar bound holds for $\inner{\vec{u}_-}{\vec{\mu}_j}^2$. Thus, \cref{lm:diff} can be slightly improved if every $\snorm{\vec{\mu}_j}^2$ is replaced by $\snorm{(\vec{I}_n - \vec{e}_n\vec{e}_n^\top)\vec{\mu}_j}^2$. 
\end{remark}
\begin{remark}[Simple bound]\label{r:ebnd}
A simple bound on the tail of difference of chi-squares can be derived directly from \cref{lm:tail.chi2} as follows: 
\begin{align}
&\mathbb{P}\Biggl\{\sum_{j=1}^k\inner{\vec{e}_t}{\vec{X}_j}^2-\sum_{j=1}^k\inner{\vec{e}_r}{\vec{X}_j}^2 \le \nonumber\\
&\mbox{}\qquad\qquad\qquad\qquad\qquad\sum_{j=1}^k\inner{\vec{e}_t}{\vec{\mu}_j}^2-\sum_{j=1}^k\inner{\vec{e}_r}{\vec{\mu}_j}^2  - x\Biggr\} \nonumber\\ 
\le\;& \Prob{\sum_{j=1}^k\inner{\vec{e}_t}{\vec{X}_j}^2-\sum_{j=1}^k\left(\inner{\vec{e}_t}{\vec{\mu}_j}^2+1\right) \le  - \frac{x}{2}} \nonumber\\
&\mbox{}\qquad\qquad\quad + \Prob{\sum_{j=1}^k\inner{\vec{e}_r}{\vec{X}_j}^2 -\sum_{j=1}^k\left(\inner{\vec{e}_r}{\vec{\mu}_j}^2+1\right)  \ge\frac{x}{2}} \nonumber\\
\le \; & 2 \exp\left(-\min\set{\frac{x^2}{32 \left(k +2\sum_{j=1}^k\bigl(\inner{\vec{e}_r}{\vec{\mu}_j}^2\vee\inner{\vec{e}_t}{\vec{\mu}_j}^2\bigr)\right)},\;\frac{x}{16}}\right) \label{e:smpbnd} \\
\le\;&  2 \exp\left(-\min\set{\frac{x^2}{32 \left(k +2\sum_{j=1}^k \norm{\vec{\mu}_j}^2\right)},\;\frac{x}{16}}\right).\label{e:smpbnd2}
\end{align}
This simple bound above is no better than that in \cref{lm:diff}, and the one in \cref{lm:diff} is strictly sharper if $d(t,r) < 1/2$, and particularly, if $d(t,r) = o(1)$.
\end{remark}

In exactly the same way as \cref{lm:diff}, we can derive the tail bound on the differences of Bernoulli weighted chi-squares. 

\begin{lemma}[Tail of difference of Bernoulli weighted chi-squares]\label{lm:difb}
Assume the same setup as in \cref{lm:diff}, and let $B_i \stackrel{\text{ind.}}{\sim} \ber(\beta_i)$, with $0 \le \beta_i \le 1$, $i = 1, \ldots, k$, be independent from $X_1, \ldots, X_k$. Then, for every $x \ge 0$:
\begin{enumerate}[i.]
\item\label{i:difb1}
It holds that 
\begin{multline*}
\Prob{\sum_{j=1}^kB_j\left(\inner{\vec{e}_t}{\vec{X}_j}^2-\inner{\vec{e}_t}{\vec{\mu}_j}^2\right)-\sum_{j=1}^kB_j\left(\inner{\vec{e}_r}{\vec{X}_j}^2-\inner{\vec{e}_r}{\vec{\mu}_j}^2\right) \le - x}\\ 
\le2\exp\Biggl(-\min\biggl\{\frac{x}{16\sqrt{(2d(t,r)\wedge 1)\max_i \bigl({1+2\norm{\vec{\mu}_i}^2}\bigr)} },\\\frac{ x^2}{48\bigl(2d(t,r)\wedge 1\bigr) \sum_{j=1}^k\beta_j\left(1 +2 \norm{\vec{\mu}_j}^2\right)}\biggr\}\Biggr).
\end{multline*}
\item\label{i:difb2}
If $\beta_i\ge 1/2$ for every $i = 1,\ldots, k$, then
\begin{multline*}
\Prob{\sum_{j=1}^kB_j\left(\inner{\vec{e}_t}{\vec{X}_j}^2-\inner{\vec{e}_t}{\vec{\mu}_j}^2\right)-\sum_{j=1}^kB_j\left(\inner{\vec{e}_r}{\vec{X}_j}^2-\inner{\vec{e}_r}{\vec{\mu}_j}^2\right) \le - x}\\ 
\le2\exp\Biggl(-\min\biggl\{\frac{x}{16\sqrt{2d(t,r)\wedge 1}},\frac{ x^2}{64\bigl(2d(t,r)\wedge 1\bigr)\left(k +2\sum_{j=1}^k \norm{\vec{\mu}_j}^2\right)}\biggr\}\Biggr).
\end{multline*}
\item\label{i:difb3}
If the condition \eqref{e:beta1ass} holds, then
\begin{multline*}
\Prob{\sum_{j=1}^kB_j\inner{\vec{e}_t}{\vec{X}_j}^2-\sum_{j=1}^kB_j\inner{\vec{e}_r}{\vec{X}_j}^2 \le \sum_{j=1}^k\inner{\vec{e}_t}{\vec{\mu}_j}^2-\sum_{j=1}^k\inner{\vec{e}_r}{\vec{\mu}_j}^2  - x}\\ 
\le 2\exp\left(-\min\set{\frac{x^2}{32 \bigl(2d(t,r)\wedge 1\bigr)  \left(k + 2 \sum_{j=1}^k\snorm{\vec{\mu}_j}^2\right)},\;\frac{(C\wedge 1/2) x }{8\sqrt{2d(t,r)\wedge 1}}}\right)\\
\times  \prod_{i=1}^k\left( 1+\exp\Bigl(-\frac{Cw_i (1+\mu_i^2)}{2}\Bigr)\right).
\end{multline*}
\end{enumerate}
\end{lemma}

The proof of \cref{lm:difb} is omitted for brevity, since the only difference to that of \cref{lm:diff} is to employ \cref{lm:bwcs,lm:bwcs2} instead of \cref{lm:tail.chi2}. \cref{r:mean0,r:ebnd} apply here as well. 

\subsection{Single change point}\label{ss:scp}
We consider here the particular case of $\ncp = 1$, i.e.\ a single change point, in Models \ref{Gaussian_setup} and \ref{m:mhgauss}, and provide proofs for \cref{th:noss,th:aoss,th:mh1cp}. Since the involved CUSUM statistic is invariant to constant shifts, we assume, without loss of generality, in \eqref{e:addform}, 
$$
\vec{f}(x) = \vec{\delta}\ind_{(0, \lambda]}(x), \qquad x \in (0,1],
$$
where $\vec{\delta} = (\delta_1,\ldots, \delta_p)^\top \in \R^p$, the change point $\cpt = \lambda \in (0, 1/2]$, and the jump size $\delta = \norm{\vec{\delta}}$, cf.~\cref{m:mhgauss}.

The following proofs rely on the observation that the localization error is no larger than the minimal length of search intervals that contain the only true change point. Moreover, assume that the search interval in a step still contains the change point, i.e.~no mistake has been made yet. Then excluding the segment containing the true change point, and thus making a mistake, can only happen when both probe points lie to the left of the change point or when both probe points lie to the right of the change point. In such cases, in order to avoid wrongly excluding the segment containing the true change point, we have to ensure that the difference of population gain function at two investigated probe points is larger than the random oscillation caused by the noise with high probability.

For ease of notation, we assume the default step size $\nu = 1/2$ in optimistic searches, which implies that the three parts within each search window have relative lengths $1:1:1$, $1:1:2$ or $2:1:1$. We further assume that there is no rounding in determining the dyadic search locations and the probe points in all search intervals.

\subsubsection{Naive optimistic search}

\begin{proof}[Proof of \cref{th:noss}]
We will prove that the assertion of theorem holds with constants
\begin{equation*}
C_0 = 20\quad\text{and}\quad C_1 = 3^2 2^{17}.
\end{equation*}
The proof consists of two steps. 

Firstly, we show that, with probability tending to one, the naive optimistic searches makes no mistake whenever a search window is of length no shorter than  $O(n\lambda)$. To this end, we introduce $I_1,I_2,\ldots$, such that $I_1 \supset I_2 \supset \cdots$, as the search windows of the naive optimistic search in the population case, i.e.~$I_1 = (0,\,n]$, $I_2 = (0,\,2n/3]$ and so on. 
It is easy to see that the change point $n\lambda$ lies in every $I_k$, thus no mistakes. Let 
$$
k_* := \min\bigl\{k \,:\, (t \wedge w) \le n\lambda  \text{ with $t,w$ the probe points of $I_k$}\bigr\},
$$
i.e., when one probe point drops in $[0, n\lambda]$ for the first time. Note that the left end point of $I_k$, $k \le k_*$, is always 0, and $\abs{I_k} \ge \abs{I_{k_*}} \asymp n\lambda$. Then it is sufficient to show that the first $k_*$ steps of the naive optimistic search coincide with the population case, with probability tending to one. 

We thus define 
$$
\PP_1 := \set{(t,w)\;:\; \text{$t,w$ are probe points of $I_k$ with $t\le w$ for $k < k_*$}}.
$$
Note that $\PP_1$ is deterministic, and only depends on the signal $f$. Recall that $\vec{F} = \bigl(f(1/n), \ldots, f(n/n)\bigr)^\top \in \R^n$. Fix an arbitrary pair of probe points $(t,w) \in \PP_1$. Then it holds that 
$$
\inner{\vec{e}_{w}}{\vec{F}}^2 \le \inner{\vec{e}_{t}}{\vec{F}}^2 = n\lambda^2\delta^2 \frac{n-t}{t}\le n^2\lambda^2\delta^2 \frac{1}{t},
$$
and
$$
\inner{\vec{e}_{t}}{\vec{F}}^2 - \inner{\vec{e}_{w}}{\vec{F}}^2 = n^2\lambda^2\delta^2\frac{\abs{t-w}}{tw} \ge n^2\lambda^2\delta^2\frac{1}{3t}.
$$
We apply inequality \eqref{e:smpbnd2} in \cref{r:ebnd} and obtain
\begin{align*}
& \Prob{\abs{\csm_{(0,n]}(t)} \le \abs{\csm_{(0,n]}(w)}} = \Prob{\inner{\vec{e}_t}{\vec{X}}^2 \le \inner{\vec{e}_w}{\vec{X}}^2} \\
\le \; & \Prob{\inner{\vec{e}_t}{\vec{X}}^2 - \inner{\vec{e}_w}{\vec{X}}^2 \le \inner{\vec{e}_t}{\vec{F}}^2 - \inner{\vec{e}_w}{\vec{F}}^2 - n^2\lambda^2\delta^2\frac{1}{3t}}\\
\le \; & 2\exp\left(-\min\set{\frac{n^4 \lambda^4 \delta^4}{288(t^2 + 2t n^2\lambda^2\delta^2)},\;\frac{n^2 \lambda^2\delta^2}{48 t}}\right)\\
\le \; & 2\exp\left(-\min\set{\frac{n^2 \lambda^4 \delta^4}{32(1 + 6 n\lambda^2\delta^2)},\;\frac{n \lambda^2\delta^2}{16}}\right),
\end{align*}
where the last inequality is due to $t \le n/3$. Since $\abs{\PP_1} \le \log(2\lambda)/\log(3/4) \le 4 \log n$, the bound above in combination with the union bound implies
\begin{align*}
& \Prob{\abs{\csm_{(0,n]}(t)} \le \abs{\csm_{(0,n]}(w)} \text{ for some }(t,w)\in \PP_1} \\
\le \; & \sum_{(t,w)\in \PP_1}\Prob{\abs{\csm_{(0,n]}(t)} \le \abs{\csm_{(0,n]}(w)}}\\
\le \; & 8 \exp\left(-\min\set{\frac{n^2 \lambda^4 \delta^4}{32(1 + 6 n\lambda^2\delta^2)},\;\frac{n \lambda^2\delta^2}{16}}+\log\log n\right).
\end{align*}
Thus, if 
$$
n\lambda^2\delta^2\ge 400\,\log\log n = C_0^2 \log\log n,
$$
the probability that the first $k_*$ steps of naive optimistic search differ from the population version satisfies
$$
\Prob{\abs{\csm_{(0,n]}(t)} \le \abs{\csm_{(0,n]}(w)} \text{ for some }(t,w)\in \PP_1} \le \frac{8}{\log n} \to 0. 
$$

Secondly, we consider the search windows in the naive optimistic search that are shorter than $O(n\lambda)$, namely, later steps $k \ge k_*$.   Recall that the true change point $n \lambda$ can be wrongly excluded from consecutive search intervals only when both pairs of probe points $t,w$ lie on the same side of $n\lambda$.  It is thus sufficient to consider 
\begin{multline*}
\PP_2:=\Bigl\{(t,w) \;:\; \text{$t,w$ are pairs of probe points from steps $k\ge k_*$ of}\\
 \text{the naive optimistic search, such that $t,w$ lie on the same side of}\\
 \text{the change point, and $t$ is closer to the change point}\Bigr\}.
\end{multline*}
We fix arbitrarily a pair of probe points $(t,w) \in \PP_2$,  and assume that the first $k_*$ steps of the naive optimistic search coincide with the population case, which happens with probability towards one as shown earlier.  It follows that $t,w \le 4n\lambda$ and further that 
$$
\inner{\vec{e}_{t}}{\vec{F}}^2 - \inner{\vec{e}_{w}}{\vec{F}}^2 =
\begin{cases}
 n^2\lambda^2\delta^2\frac{\abs{t-w}}{tw} \ge \frac{1}{16}\abs{t-w}\delta^2&\text{ if }t,w \ge n\lambda,\\
 n^2\delta^2\frac{\abs{t-w}(1-\lambda)^2}{(n-t)(n-w)} \ge  \frac{1}{4}\abs{t-w}\delta^2&\text{ if }t,w \le n\lambda.
\end{cases}
$$
The relative distance $d(\cdot, \cdot)$ in \eqref{e:reld} satisfies
$$
d(t,w) = \frac{\abs{t-w}}{\min\set{t\wedge w, (n-t)\wedge (n-w)}}\le \frac{4\abs{t-w}}{n\lambda}.
$$
Then, by \cref{lm:diff} we have, for $(t,w)\in \PP_2$,
\begin{align*}
&\Prob{\abs{\csm_{(0,n]}(t)} \le \abs{\csm_{(0,n]}(w)}}= \Prob{\inner{\vec{e}_t}{\vec{X}}^2 \le \inner{\vec{e}_w}{\vec{X}}^2} \\
\le \; & \Prob{\inner{\vec{e}_t}{\vec{X}}^2 - \inner{\vec{e}_w}{\vec{X}}^2 \le \inner{\vec{e}_t}{\vec{F}}^2 - \inner{\vec{e}_w}{\vec{F}}^2 - \frac{1}{16} \abs{t-w}\delta^2}\\
\le\;& 2\exp\left(-\min\set{\frac{x^2}{64 d(t,r)  \left(1 + 2 n\lambda\delta^2\right)},\;\frac{x}{16\sqrt{2d(t,r)}}}\right)\\
\le\;& 2\exp\left(-\min\set{\frac{n\lambda\abs{t-w}\delta^4}{ 2^{16} \left(1 + 2 n\lambda\delta^2\right)},\;\frac{\sqrt{n\lambda\abs{t-w}}\delta^2}{2^{19/2}}}\right).
\end{align*}
Note that $\PP_2$ is contained in a mother set $\PP_2^*$ of size $\le (4 \log n)^2$, and that $\PP_2^*$ is determined only by the signal $f$, see later Part~2 in the proof of \cref{th:mh1cp} for a formal proof. Let 
$$
\eps_* := \frac{C_1}{2} \frac{\log\log n}{\delta^2} = 3^2 2^{16}\frac{\log\log n}{\delta^2}.
$$
Then, by the union bound again, we obtain
\begin{align*}
& \Prob{\abs{\csm_{(0,n]}(t)} \le \abs{\csm_{(0,n]}(w)} \text{ for some }(t,w)\in \PP_2 \text{ with }\abs{t-w}\ge\eps_*} \\
\le\;& \Prob{\abs{\csm_{(0,n]}(t)} \le \abs{\csm_{(0,n]}(w)} \text{ for some }(t,w)\in \PP_2^* \text{ with }\abs{t-w}\ge\eps_*} \\
\le \; & \sum_{(t,w)\in \PP_2^*,\,\abs{t-w}\ge\eps_*}\Prob{\abs{\csm_{(0,n]}(t)} \le \abs{\csm_{(0,n]}(w)}}\\
\le\; & 32\exp\left(-\min\set{\frac{n\lambda\eps_*\delta^4}{ 2^{16} \left(1 + 2 n\lambda\delta^2\right)},\;\frac{\sqrt{n\lambda\eps}\delta^2}{2^{19/2}}}+2\log\log n\right) \\
\le\;&\frac{32}{\log n} \to 0.
\end{align*}
This implies an upper bound of $2\eps_*$ on the localization error of the change point, which concludes the proof.  
\end{proof}

\subsubsection{Advanced and combined optimistic searches}

Since \cref{th:aoss} is a special of \cref{th:mh1cp} when $p =1$, we only need to prove \cref{th:mh1cp}.

\begin{proof}[Proof of \cref{th:mh1cp}]
We divide the $p$ coordinates of observations in \cref{m:mhgauss} into three groups:
\begin{enumerate}[i.]
\item The set of coordinates with large jump sizes
$$
L := \set{j\; :\;  \abs{\delta_j}^2 \ge \frac{1}{32}\frac{\norm{\vec{\delta}}^2}{s}};
$$
\item The set of coordinates with small jump sizes
$$
S := \set{j\; :\; 0< \abs{\delta_j}^2 <  \frac{1}{32}\frac{\norm{\vec{\delta}}^2}{s}};
$$
\item The set of coordinates with no jumps 
$$
N := \set{j\;:\; \delta_j = 0}.
$$
\end{enumerate}
The constant $1/32$ above can be replaced by any constant that is sufficiently small. 
Clearly, $L,S, N$ are disjoint, $\abs{L} + \abs{S} \le s$ and $L \cup S \cup N = \set{1,\ldots, p}$.  It further holds that $L \neq \emptyset$, since 
$$
\sum_{j \in S}\abs{\delta_j}^2 \le \frac{1}{32} \norm{\vec{\delta}}^2 <  \norm{\vec{\delta}}^2 =\sum_{j \in S\cup L}\abs{\delta_j}^2.
$$
The remaining proof is split into two parts. 

\emph{Part 1. Global search over dyadic locations.} In this part, we will show
\begin{equation}\label{e:dyadic}
\Prob{n\lambda \in (t_*/2,\;2t_*] } \to 1,
\end{equation}
with $t_*$ the output of the dyadic search (i.e.\ line 4 in \cref{Alg:aOS}), provided that the constant $C_0$ in \eqref{e:db1cp:a} is sufficiently large. One possible choice is 
\begin{equation}\label{e:chC0}
C_0 = 2^{18}\cdot3^2.
\end{equation}

Since $0 < \lambda \le 1/2$, there exists an integer $k_0 \ge 1$ such that $\lambda \in ( 2^{-k_0  -1}, 2^{-k_0}]$. Then \eqref{e:dyadic} is equivalent to 
$$
\Prob{t_* = 2^{-k_0-1}n \;\text{ or }\; 2^{-k_0}n} \to 1.
$$
Thus, we only need to show that 
\begin{subequations}
\begin{equation}
\Prob{\cmp(t_{k_0+1}, t_{k}) \le 0 \quad \text{for some } \; k = k_0+2, \ldots, \log_2 n} \to 0, \label{e:ldgs}
\end{equation}
and 
\begin{equation}
\Prob{\cmp(t_{k_0}, t_{k}) \le 0 \quad \text{for some } \; k = 1, \ldots, k_0-1} \to 0,\label{e:rdgs}
\end{equation}
where  $\cmp(\cdot,\cdot)\equiv \cmp_{(0,n]}(\cdot,\cdot)$ is the comparison function defined in \eqref{e:mhcmp}, and $t_k := 2^{-k} n$ for $k \in \set{1, \ldots, \log_2 n}$ the dyadic locations.
\end{subequations}

We first prove \eqref{e:ldgs}. Note that for $k \in \set{k_0 + 2, \ldots, \log_2 n}$, i.e.\ on the left side of the change point, and for $j \in \{1,\ldots, p\}$,
\begin{align*}
\csm_{j}(t_{k_0+1}; \vec{F})^2 -  \csm_{j}(t_{k}; \vec{F})^2 & \ge \csm_{j}(t_{k_0+1}; \vec{F})^2 -  \csm_{j}(t_{k_0+2}; \vec{F})^2 \\
&  = n ^2(1-\lambda)^2{\delta}_j^2 \frac{t_{k_0+1} - t_{k_0+2}}{(n-t_{k_0+1})(n-t_{k_0+2})}\\
& \ge  (1-\lambda)^2{\delta}_j^2 \frac{n\lambda}{4(1-\lambda/2)(1-\lambda/4)}\\
& \ge \frac{2}{21}n \lambda {\delta}_j^2,
\end{align*}
where $\csm_{j}(\cdot; \vec{F}) \equiv \csm_{(0,n],j}(\cdot; \vec{F}) $ is the CUSUM statistics in the $j$-th coordinate applied to the signal matrix $\vec{F}$. Thus, it holds that
\begin{align*}
& \sum_{j\in L} \Bigl(\csm_{j}(t_{k_0+1}; \vec{F})^2 -  \csm_{j}(t_{k}; \vec{F})^2\Bigr)\\
 =\;& \sum_{j\in L \cup S} \Bigl(\csm_{j}(t_{k_0+1}; \vec{F})^2 -  \csm_{j}(t_{k}; \vec{F})^2\Bigr) \\
 &\qquad\qquad\qquad- \sum_{j\in S} \Bigl(\csm_{j}(t_{k_0+1}; \vec{F})^2 -  \csm_{j}(t_{k}; \vec{F})^2\Bigr) \\
\ge\; &  \frac{2}{21}n\lambda\sum_{j \in L \cup S} \delta_j^2 -  \sum_{j\in S} \csm_{j}(t_{k_0+1}; \vec{F})^2 \\
\ge\; &  \frac{2}{21}n\lambda \norm{\vec{\delta}}^2 - \abs{S} \frac{1}{32}\frac{n\lambda\norm{\vec{\delta}^2}}{s} \;\ge\;  \frac{1}{16}n\lambda \norm{\vec{\delta}}^2.
\end{align*}
Fix an arbitrary $k \in \set{k_0 + 2, \ldots, \log_2 n}$, and let 
$$
B_j := \ind\Bigl\{\max\bigl(\abs{\csm_{j}(t_{k_0+1}; \vec{Y})}, \abs{\csm_{j}(t_k; \vec{Y})} \bigr)\ge \thd\Bigr\},$$ 
for $j =1, \ldots, p$. Introduce 
\begin{equation}\label{e:chx}
x := \frac{1}{16}n \lambda \norm{\vec{\delta}}^2,
\end{equation}
and we have 
\begin{subequations}
\begin{align}
& \Prob{\cmp(t_{k_0+1}, t_k) \le 0}\nonumber \\ 
\le\; & \mathbb{P}\biggl\{\sum_{j \in L \cup S\cup N }B_j\Bigl(\csm_{j}(t_{k_0+1}; \vec{X})^2 -  \csm_{j}(t_{k}; \vec{X})^2\Bigr) \nonumber\\
&\qquad\qquad\qquad\qquad \le \sum_{j \in L}\Bigl(\csm_{j}(t_{k_0+1}; \vec{F})^2 -  \csm_{j}(t_{k}; \vec{F})^2\Bigr) - x\biggr\} \nonumber\\
\le \;&\mathbb{P}\Biggl\{\sum_{j \in L }B_j\Bigl(\csm_{j}(t_{k_0+1}; \vec{X})^2 -  \csm_{j}(t_{k}; \vec{X})^2\Bigr) \nonumber\\
&\qquad\qquad\qquad\qquad\le \sum_{j\in L} \Bigl(\csm_{j}(t_{k_0+1}; \vec{F})^2 -  \csm_{j}(t_{k}; \vec{F})^2\Bigr)-\frac{x}{3}\Biggr\} \label{e:dyp1}\\
&  \qquad\quad + \Prob{\sum_{j \in S}B_j\Bigl(\csm_{j}(t_{k_0+1}; \vec{X})^2 -  \csm_{j}(t_{k}; \vec{X})^2\Bigr)\le -\frac{x}{3}} \label{e:dyp2}\\
&  \qquad\quad +   \Prob{\sum_{j \in N }B_j\Bigl(\csm_{j}(t_{k_0+1}; \vec{X})^2 -  \csm_{j}(t_{k}; \vec{X})^2\Bigr)\le -\frac{x}{3}}. \label{e:dyp3}
\end{align}
\end{subequations}
Next we bound the probabilities in \eqref{e:dyp1}, \eqref{e:dyp2} and \eqref{e:dyp3}, separately. 
\begin{enumerate}[i.]
\item
The probability in \eqref{e:dyp1} can be bounded from above by
\begin{multline}\label{e:dysptL}
\Prob{\sum_{j \in L }B_j\csm_{j}(t_{k_0+1}; \vec{X})^2-\sum_{j\in L}\Bigl(1+\csm_{j}(t_{k_0+1}; \vec{F})^2\Bigr)\le -\frac{x}{6}}\\
+\Prob{\sum_{j \in L }B_j\csm_{j}(t_{k}; \vec{X})^2-\sum_{j\in L}\Bigl(1+\csm_{j}(t_{k}; \vec{F})^2\Bigr)\ge \frac{x}{6}}.
\end{multline}
For the first term in \eqref{e:dysptL}, we consider two cases separately. 
\begin{itemize}
\item Dense case, i.e.\ when $s \ge \sqrt{p \log\log n}$ or $s=p$. It follows that $\thd = 1$ and thus $B_j = 1$ for $j \in L$. By \cref{lm:tail.chi2}, we have 
\begin{align*}
&\Prob{\sum_{j \in L }B_j\csm_{j}(t_{k_0+1}; \vec{X})^2-\sum_{j\in L}\Bigl(1+\csm_{j}(t_{k_0+1}; \vec{F})^2\Bigr)\le -\frac{x}{6}}\\
=\; & \Prob{\sum_{j \in L }\csm_{j}(t_{k_0+1}; \vec{X})^2-\sum_{j\in L}\Bigl(1+\csm_{j}(t_{k_0+1}; \vec{F})^2\Bigr)\le -\frac{x}{6}}\\
\le\;& \exp\left(-\frac{x^2}{144\bigl(s + 2n\lambda\norm{\vec{\delta}}^2\bigr)}\right).
\end{align*}
\item Sparse case, i.e.\ when $s < \sqrt{p \log\log n}$ and $s \neq p$. It implies $\thd \ge 2$. Note that, for every $j \in \set{1,\ldots,p}$, 
$$
\abs{\csm_{j}(t_{k_0+1}; \vec{F})}^2 \ge n\lambda\delta_j^2\frac{(1-\lambda)^2}{1-\lambda/2} \ge \frac{1}{3}  n\lambda\delta_j^2,
$$
and that $n\lambda\norm{\vec{\delta}}^2 \ge C_0 s\thd^2/4 \ge C_0 s$.  For $j \in L$, by the choice of $C_0$ in \eqref{e:chC0}, we have
$$
(\abs{\csm_{j}(t_{k_0+1}; \vec{F})} - \thd)^2 \ge \frac{2}{9} {n\lambda\delta_j^2} \ge\frac{1}{6} \bigl(1 + {n\lambda\delta_j^2}\bigr),
$$
and then, by Mill's ratio,
\begin{align*}
\Prob{B_j = 0} & \le \Prob{\abs{\csm_{j}(t_{k_0+1}; \vec{Y})} \le \thd}\\
& \le \Prob{\bigl|{\csm_{j}(t_{k_0+1}; \tilde{\vec{\Xi}})}\bigr| \ge \abs{\csm_{j}(t_{k_0+1}; \vec{F})} - \thd}\\
& \le 2\bigl(1-\Phi(\abs{\csm_{j}(t_{k_0+1}; \vec{F})} - \thd)\bigr)\\
& \le  \exp\left(-\frac{(\abs{\csm_{j}(t_{k_0+1}; \vec{F})} - \thd)^2}{2}\right)\\
& \le \exp\left(-\frac{1}{12}\bigl(1 + {n\lambda\delta_j^2}\bigr)\right).
\end{align*}
Recall that $\Phi(\cdot)$ denotes the distribution function of standard Gaussian random variable. Thus, by \cref{lm:bwcs2}, we obtain
\begin{multline*}
\Prob{\sum_{j \in L }B_j\csm_{j}(t_{k_0+1}; \vec{X})^2-\sum_{j\in L}\Bigl(1+\csm_{j}(t_{k_0+1}; \vec{F})^2\Bigr)\le -\frac{x}{6}}\\
\le 2^s \exp \left(-\min\set{\frac{x}{288},\;\frac{x^2}{144\bigl(s+2n\lambda\norm{\vec{\delta}}^2\bigr)}}\right).
\end{multline*}
\end{itemize}

For the second term in \eqref{e:dysptL}, because $\Prob{B_j = 1}\ge1/2$ for every $j \in L$, we apply \eqref{e:bBwC} in \cref{lm:bwcs} and obtain
\begin{align*}
& \Prob{\sum_{j \in L }B_j\csm_{j}(t_{k}; \vec{X})^2-\sum_{j\in L}\Bigl(1+\csm_{j}(t_{k}; \vec{F})^2\Bigr)\ge \frac{x}{6}} \\
\le\;& \Prob{\sum_{j \in L }B_j\csm_{j}(t_{k}; \vec{X})^2-\sum_{j\in L}B_j\Bigl(1+\csm_{j}(t_{k}; \vec{F})^2\Bigr)\ge \frac{x}{6}}\\
\le \; & \exp\left(-\min\set{\frac{x}{48},\; \frac{x^2}{576\bigl(s + 2n\lambda\norm{\vec{\delta}}^2\bigr)}}\right). 
\end{align*}

\item
We split the probability in \eqref{e:dyp2} according to all possible values of $(B_j)_{j \in S}$, namely, 
\begin{align*}
& \Prob{\sum_{j \in S}B_j\Bigl(\csm_{j}(t_{k_0+1}; \vec{X})^2 -  \csm_{j}(t_{k}; \vec{X})^2\Bigr)\le -\frac{x}{3}}\\
= \; & \sum_{\emptyset \neq J \subset S}\Biggl( \Prob{\sum_{j \in J}B_j \Bigl(\csm_{j}(t_{k_0+1}; \vec{X})^2 -  \csm_{j}(t_{k}; \vec{X})^2\Bigr)\le -\frac{x}{3} \;\Bigg\vert\; B_j = 1 \text{ iff } j \in J}\\
&\qquad\qquad\qquad\qquad\qquad\qquad  \times\Prob{B_j =1 \text{ iff } j \in J}\Biggr)\\
= \; & \sum_{\emptyset \neq J \subset S} \Prob{\sum_{j \in J} \Bigl(\csm_{j}(t_{k_0+1}; \vec{X})^2 -  \csm_{j}(t_{k}; \vec{X})^2\Bigr)\le -\frac{x}{3}}\Prob{B_j =1 \text{ iff } j \in J}.
\end{align*}
For every $\emptyset \neq J \subset S$, we apply \cref{lm:tail.chi2} (cf.\ \cref{r:ebnd}) and obtain
\begin{align*}
 &\Prob{\sum_{j \in J} \Bigl(\csm_{j}(t_{k_0+1}; \vec{X})^2 -  \csm_{j}(t_{k}; \vec{X})^2\Bigr)\le -\frac{x}{3}}\\ 
\le \;& \mathbb{P}\Biggl\{\sum_{j \in J} \biggl(\Bigl(\csm_{j}(t_{k_0+1}; \vec{X})^2 -  \csm_{j}(t_{k}; \vec{X})^2\Bigr) \\
&\qquad\qquad \qquad\qquad\qquad-\Bigl(\csm_{j}(t_{k_0+1}; \vec{F})^2 -  \csm_{j}(t_{k}; \vec{F})^2\Bigr) \biggr)\le -\frac{x}{3}\Biggr\}\\
\le \; &\Prob{\sum_{j \in J}\Bigl(\csm_{j}(t_{k_0+1}; \vec{X})^2 -\Bigl(\csm_{j}(t_{k_0+1}; \vec{F})^2 \Bigr)\le -\frac{x}{6}}\\
&\qquad\qquad\qquad\qquad\qquad +\Prob{\sum_{j \in J}\Bigl(\csm_{j}(t_{k}; \vec{X})^2- \csm_{j}(t_{k}; \vec{F})^2\Bigr)\ge \frac{x}{6}}\\
\le\; &  2\exp\left(-\min\set{\frac{x}{48},\;\frac{x^2}{288(s + 2 n\lambda\norm{\vec{\delta}}^2/21)}}\right).
\end{align*}
Thus, we have 
\begin{align*}
 &\Prob{\sum_{j \in J} \Bigl(\csm_{j}(t_{k_0+1}; \vec{X})^2 -  \csm_{j}(t_{k}; \vec{X})^2\Bigr)\le -\frac{x}{3}}\\ 
\le\; &  \sum_{\emptyset \neq J \subset S}2\exp\left(-\min\set{\frac{x}{48},\;\frac{x^2}{288(s + 2 n\lambda\norm{\vec{\delta}}^2/21)}}\right)\Prob{B_j =1 \text{ iff } j \in J}\\
=\;& 2\exp\left(-\min\set{\frac{x}{48},\;\frac{x^2}{288\bigl(s + 2 n\lambda\norm{\vec{\delta}}^2/21\bigr)}}\right).
\end{align*}
\item
Consider the probability in \eqref{e:dyp3}, and note that for $j \in N$,
\begin{align*}
\Prob{B_j =1} & \le \Prob{\abs{\csm_{j}(t_{k_0+1}; \tilde{\vec{\Xi}})} \ge \thd} +  \Prob{\abs{\csm_{j}(t_{k}; \tilde{\vec{\Xi}})} \ge \thd} \\
& = 2\bigl(1-\Phi(\thd)\bigr) \le \exp\left(-\frac{\thd^2}{2}\right).
\end{align*}
We apply \cref{lm:bwcs}, more precisely, \eqref{e:gbBwC}, and obtain 
\begin{align*}
&\Prob{\sum_{j \in N }B_j\Bigl(\csm_{j}(t_{k_0+1}; \vec{X})^2 -  \csm_{j}(t_{k}; \vec{X})^2\Bigr)\le -\frac{x}{3}} \\
= \; & \Prob{\sum_{j \in N }B_j\Bigl(\csm_{j}(t_{k_0+1}; \vec{\Xi})^2 -  \csm_{j}(t_{k}; \vec{\Xi})^2\Bigr)\le -\frac{x}{3}} \\
\le\; & \Prob{\sum_{j \in N }B_j\Bigl(\csm_{j}(t_{k_0+1}; \vec{\Xi})^2 -  1\Bigr)\le -\frac{x}{3}} \\
& \qquad\qquad\qquad\qquad+ \Prob{\sum_{j \in N }B_j\Bigl(1 -  \csm_{j}(t_{k}; \vec{\Xi})^2\Bigr)\le -\frac{x}{3}} \\
\le\; & 2\exp\left(-\min\set{\frac{x}{48},\;\frac{x^2}{432\sum_{j \in N} \Prob{B_j = 1} }}\right)\\
\le\;& 2\exp\left(-\min\set{\frac{x}{48},\;\frac{x^2}{432\, p\exp(-\thd^2/2)}}\right).
\end{align*}
\end{enumerate}
Therefore, combining all the bounds above and \eqref{e:chx}, we obtain 
\begin{align*}
&\Prob{\cmp(t_{k_0+1}, t_{k}) \le 0} \\
\le\;& 2 \exp\left(-\frac{x^2}{432\, p \exp(-\thd^2/2)}\right)\\
&\qquad+ \bigl(6 + 2^s\ind\set{\alpha >0}\bigr)\exp\left(-\min\set{\frac{x}{288},\; \frac{x^2}{576\bigl(s + 2n\lambda\norm{\vec{\delta}}^2\bigr)}}\right)\\
=\;& \exp\left(-\frac{n^2\lambda^2\norm{\vec{\delta}}^4}{2^{12}3^3 p \exp(-\thd^2/2)}\right)\\
&\qquad+ \bigl(6 + 2^s\ind\set{\alpha >0}\bigr)\exp\left(-\min\set{\frac{n\lambda\norm{\vec{\delta}}^2}{2^93^2},\; \frac{n^2\lambda^2\norm{\vec{\delta}}^4}{2^{14}3^2\bigl(s + 2n\lambda\norm{\vec{\delta}}^2\bigr)}}\right),
\end{align*}
for any fixed $k \in \set{k_0 + 2, \ldots, \log_2 n}$. This together with the union bound (i.e.~Boole's inequality) implies
\begin{align*}
&\Prob{\cmp(t_{k_0+1}, t_{k}) \le 0 \; \text{ for some } \; k = k_0+2, \ldots, \log_2 n} \\
\le\;& \sum_{k=k_0+2}^{\log_2 n}\Prob{\cmp(t_{k_0+1}, t_{k}) \le 0} \le \frac{12}{\log n} \to 0,
\end{align*}
where we use 
\begin{equation}\label{e:pthd}
 p\exp\left(-\frac{\thd^2}{2}\right) =
 \begin{cases}
 p & \text{ if } \thd = 0,\\
 \frac{s^2}{e^2 \log\log n} &  \text{ if } \thd > 0,
 \end{cases}
\end{equation}
and the assumption \eqref{e:db1cp} with $C_0$ in \eqref{e:chC0}.

Next, we consider \eqref{e:rdgs}. For $k \in \set{1,\ldots, k_0 -1}$, i.e.\ on the right side of the change point, and for $j \in \{1,\ldots, p\}$, it holds that
\begin{multline*}
\csm_{j}(t_{k_0}; \vec{F})^2 -  \csm_{j}(t_{k}; \vec{F})^2  \ge \csm_{j}(t_{k_0}; \vec{F})^2 -  \csm_{j}(t_{k_0-1}; \vec{F})^2 \\
 = \frac{1}{2t_{k_0}}n^2\lambda^2\delta_j^2
 \ge \frac{1}{4}n \lambda {\delta}_j^2 \ge \frac{2}{21}n \lambda {\delta}_j^2,
\end{multline*}
and also that 
$$
\abs{\csm_{j}(t_{k_0}; \vec{F})}^2 = n \lambda^2\delta_j^2\frac{n-t_{k_0}}{t_{k_0}} \ge \frac{1}{4}  n\lambda\delta_j^2.
$$
Thus, for $j \in L$, we have
$$
(\abs{\csm_{j}(t_{k_0+1}; \vec{F})} - \thd)^2 \ge \frac{2}{9} {n\lambda\delta_j^2} \ge \frac{1}{6} \bigl(1 + {n\lambda\delta_j^2}\bigr).
$$
Therefore, the same calculation as for \eqref{e:ldgs} remain valid, and thus 
$$
\Prob{\cmp(t_{k_0}, t_{k}) \le 0 \; \text{ for some } \; k = 1,\ldots, k_0-1} \le \frac{12}{\log n} \to 0.
$$

\emph{Part 2. Advanced optimistic search.} We will show that the assertion of \cref{th:mh1cp} holds with $C_0$ given in \eqref{e:chC0} and 
\begin{equation}\label{e:chC1}
C_1 = 2^{29}\, 3^{4}.
\end{equation}
Based on Part 1, we can start from the search interval $(t_*/2, \, 2t_*]$, which contains $n \lambda$ with probability towards one as $n \to \infty$.  

We only need to consider the pairs of probe points $r,t$ at the same side of $n\lambda$, since otherwise no matter which side is dropped off, no mistake will occur. Thus, we introduce
\begin{multline*}
\PP = \Bigl\{(t,r) \;:\; \text{$t,r$ are the pair of probe points in the same step of}\\
\text{optimistic search such that $t,r$ lie on the same side of}\\ \text{the change point, and $t$ is closer to the change point}\Bigr\}.
\end{multline*}
That is, for $(t,r)\in \PP$, it holds $(t -n\lambda)(r-n\lambda) \ge 0$ and  $\abs{t-n\lambda} \le \abs{r - n\lambda}$. 

Consider arbitrarily $(t,r) \in \PP$. Recall that the step size of optimistic search is $\nu = 1/2$, and that the relative distance $d(\cdot, \cdot)$ is defined in \eqref{e:reld}. Since $2t_* \le n$, it always holds 
$$
(r \vee t) \le \frac{2}{3}\left(2t_*- \frac{1}{2}t_*\right) + \frac{1}{2} t_* \le \frac{3}{2}t_* \le \frac{3}{4}n.
$$
Then, using $n\lambda/2 \le t_* \le 2 n\lambda$, we have 
\begin{multline}\label{e:dbnd}
\frac{\abs{t-r}}{4n\lambda} \le \frac{\abs{t-r}}{2t_*} \le \frac{\abs{t-r}}{r \vee t}  \le d(t, r) = \frac{\abs{t-r}}{(t \vee r) \wedge \bigl(n -(r\wedge t)\bigr)}\\
\le \frac{\abs{t-r}}{(t_*/2) \wedge (n - 3n/4)} \le \frac{4\abs{t-r}}{n\lambda},
\end{multline}
that is, $d(t,r) \asymp \abs{t-r}/(n\lambda)$.

Given $(t,r)\in\PP$, there are only two possible cases: 
\begin{itemize}
\item
Case: $n\lambda \le t \le r \le 2t_*$. For $j \in \set{1,\ldots, p}$, we have 
\begin{align*}
\abs{t-r}\delta_j^2\ge \csm_j(t; \mat{F})^2 - \csm_j(r; \mat{F})^2 = n^2\lambda^2\delta_j^2\frac{r-t}{tr} \ge \frac{1}{16}\abs{t-r}\delta_j^2,
\end{align*}
where the second inequality is due to $t \le r \le 2t_* \le 4 n\lambda$.
\item
Case: $t_*/2 < r \le t \le n\lambda$. For $j \in \set{1,\ldots, p}$,  we have 
\begin{align*}
\abs{t-r}\delta_j^2 \ge \csm_j(t; \mat{F})^2  - \csm_j(r; \mat{F})^2& = n^2(1-\lambda)^2\delta_j^2\frac{t-r}{(n-t)(n-r)}\\
& \ge \frac{(1-\lambda)^2}{(1-\lambda/4)^2}\abs{t-r}\delta_j^2\\
& \ge \frac{16}{49}\abs{t-r}\delta_j^2 \ge \frac{1}{16}\abs{t-r}\delta_j^2,
\end{align*}
where the second inequality is due to $n\lambda/4 \le t_*/2 < r \le t$, and the third inequality is due to $\lambda \le 1/2$. 
\end{itemize}
Thus, it follows 
\begin{align*}
& \sum_{j\in L} \Bigl(\csm_{j}(t; \vec{F})^2 -  \csm_{j}(r; \vec{F})^2\Bigr)\\
 =\;& \sum_{j\in L \cup S} \Bigl(\csm_{j}(t; \vec{F})^2 -  \csm_{j}(r; \vec{F})^2\Bigr) \\
 &\qquad\qquad\qquad- \sum_{j\in S} \Bigl(\csm_{j}(t; \vec{F})^2 -  \csm_{j}(r; \vec{F})^2\Bigr)\\
\ge\; &  \frac{1}{16}\abs{t-r}\sum_{j \in L \cup S} \delta_j^2 -  {\abs{t-r}}\sum_{j\in S} \delta_j^2 \\
\ge\; &  \frac{1}{16}\abs{t-r} \norm{\vec{\delta}}^2 - \abs{S}\abs{t-r} \frac{\norm{\vec{\delta}}^2}{32 s} \;\ge\;  \frac{1}{32}\abs{t-r} \norm{\vec{\delta}}^2.
\end{align*}
Fix a pair of probe points $(t,r)\in \PP$, and introduce  
\begin{equation}\label{e:chx2}
x := \frac{1}{32}\abs{t-r} \norm{\vec{\delta}}^2,
\end{equation}
and, for $j\in \set{1,\ldots, p}$, 
$$
B_j := \ind\Bigl\{\max\bigl(\abs{\csm_{j}(t; \vec{Y})}, \abs{\csm_{j}(r; \vec{Y})} \bigr)\ge \thd\Bigr\}. 
$$ 
Then we can bound the probability of making a mistake as below:
\begin{subequations}
\begin{align}
& \Prob{\cmp(t, r) \le 0}\nonumber \\ 
\le\; & \mathbb{P}\biggl\{\sum_{j \in L \cup S\cup N }B_j\Bigl(\csm_{j}(t; \vec{X})^2 -  \csm_{j}(r; \vec{X})^2\Bigr) \nonumber\\
&\qquad\qquad\qquad\qquad \le \sum_{j \in L}\Bigl(\csm_{j}(t; \vec{F})^2 -  \csm_{j}(r; \vec{F})^2\Bigr) - x\biggr\} \nonumber\\
\le \;&\mathbb{P}\Biggl\{\sum_{j \in L }B_j\Bigl(\csm_{j}(t; \vec{X})^2 -  \csm_{j}(r; \vec{X})^2\Bigr) \nonumber\\
&\qquad\qquad\qquad\qquad\le \sum_{j\in L} \Bigl(\csm_{j}(t; \vec{F})^2 -  \csm_{j}(r; \vec{F})^2\Bigr)-\frac{x}{3}\Biggr\} \label{e:osp1}\\
&  \qquad\quad + \Prob{\sum_{j \in S}B_j\Bigl(\csm_{j}(t; \vec{X})^2 -  \csm_{j}(r; \vec{X})^2\Bigr)\le -\frac{x}{3}} \label{e:osp2}\\
&  \qquad\quad +   \Prob{\sum_{j \in N }B_j\Bigl(\csm_{j}(t; \vec{X})^2 -  \csm_{j}(r; \vec{X})^2\Bigr)\le -\frac{x}{3}}. \label{e:osp3}
\end{align}
\end{subequations}

The three terms above can be bounded in a similar way as \eqref{e:dyp1}, \eqref{e:dyp2} and \eqref{e:dyp3} in Part~1, respectively, but we need to decouple the correlation between $\csm_{j}(t; \vec{X})$ and $\csm_{j}(r; \vec{X})$ by means of \cref{lm:diff,lm:difb}. The details are given below. 

\begin{enumerate}[i.]
\item
For \eqref{e:osp1}, we consider two cases separately.
\begin{itemize}
\item 
Dense case, i.e.\ when $s \ge \sqrt{p \log\log n}$ or $s=p$. Note that $\thd = 0$ and thus $B_j = 1$ for $j \in L$. Then \cref{lm:diff} implies
\begin{align*}
&\mathbb{P}\Biggl\{\sum_{j \in L }B_j\Bigl(\csm_{j}(t; \vec{X})^2 -  \csm_{j}(r; \vec{X})^2\Bigr) \nonumber\\
&\qquad\qquad\qquad\qquad\le \sum_{j\in L} \Bigl(\csm_{j}(t; \vec{F})^2 -  \csm_{j}(r; \vec{F})^2\Bigr)-\frac{x}{3}\Biggr\}\\
= \; &\mathbb{P}\Biggl\{\sum_{j \in L }\Bigl(\csm_{j}(t; \vec{X})^2 -  \csm_{j}(r; \vec{X})^2\Bigr) \nonumber\\
&\qquad\qquad\qquad\qquad\le \sum_{j\in L} \Bigl(\csm_{j}(t; \vec{F})^2 -  \csm_{j}(r; \vec{F})^2\Bigr)-\frac{x}{3}\Biggr\}\\
\le\;& 2\exp\left(-\min\set{\frac{x^2}{576\,d(t,r)\bigl(s + 2n\lambda\norm{\vec{\delta}}^2\bigr)},\; \frac{x}{48\sqrt{2d(t,r)}}}\right).
\end{align*}

\item Sparse case, i.e.\ when $s < \sqrt{p \log\log n}$ and $s \neq p$. In this case, it holds that $\thd \ge 2$, and thus that $n\lambda\norm{\vec{\delta}}^2 \ge C_0 s\thd^2/4 \ge C_0 s$. Because $t$ is closer the change point and $t \in [t_*, 3t_*/2]$ (which is due to $\nu =1/2$), we obtain, for $j \in \set{1,\ldots, p}$, 
$$
\abs{\csm_{j}(t; \vec{F})}^2 \ge n \lambda^2\delta_j^2\frac{n-t}{t} \ge \frac{1}{4}  n\lambda\delta_j^2.
$$
Then, the choice of $C_0$ in \eqref{e:chC0} implies, for $j \in L$,
$$
(\abs{\csm_{j}(t; \vec{F})} - \thd)^2 \ge \frac{2}{9} {n\lambda\delta_j^2} \ge\frac{1}{6} \bigl(1 + {n\lambda\delta_j^2}\bigr),
$$
and further, by Mill's ratio,
\begin{align*}
\Prob{B_j = 0} & \le \Prob{\abs{\csm_{j}(t; \vec{Y})} \le \thd}\\
& \le \Prob{\bigl|{\csm_{j}(t; \tilde{\vec{\Xi}})}\bigr| \ge \abs{\csm_{j}(t; \vec{F})} - \thd}\\
& \le 2\bigl(1-\Phi(\abs{\csm_{j}(t; \vec{F})} - \thd)\bigr)\\
& \le  \exp\left(-\frac{(\abs{\csm_{j}(t; \vec{F})} - \thd)^2}{2}\right)\\
& \le \exp\left(-\frac{1}{12}\bigl(1 + {n\lambda\delta_j^2}\bigr)\right).
\end{align*}
Thus, by \cref{lm:difb}\ref{i:difb3}, we obtain
\begin{align*}
&\mathbb{P}\Biggl\{\sum_{j \in L }B_j\Bigl(\csm_{j}(t; \vec{X})^2 -  \csm_{j}(r; \vec{X})^2\Bigr) \nonumber\\
&\qquad\qquad\qquad\qquad\le \sum_{j\in L} \Bigl(\csm_{j}(t; \vec{F})^2 -  \csm_{j}(r; \vec{F})^2\Bigr)-\frac{x}{3}\Biggr\}\\
\le\;& 2\exp\left(-\min\set{\frac{x^2}{576 \, d(t,r)  \left(s + 2n\lambda\snorm{\vec{\delta}}^2\right)},\;\frac{x}{288\sqrt{2d(t,r)}}}\right)\\
& \qquad\qquad\qquad\qquad \times \prod_{j \in L}\left( 1+\exp\Bigl(-\frac{1}{24}\bigl(1+n\lambda\delta_j^2\bigr)\Bigr)\right)\\
\le\; & 2\exp\biggl(-\min\Bigl\{\frac{x^2}{576 \, d(t,r)  \left(s + 2n\lambda\snorm{\vec{\delta}}^2\right)},\;\frac{x}{288\sqrt{2d(t,r)}}\Bigr\} \\
& \qquad\qquad\qquad\qquad + s\exp\Bigl(-\frac{1}{768}\frac{n\lambda\norm{\vec{\delta}}^2}{s}\Bigr)\biggr),
\end{align*}
where the last inequality is due to the definition of $L$ and the basic inequality $1+x \le \exp(x)$.
\end{itemize}
\item
We decompose the probability in \eqref{e:osp2} into events conditioned on all possible values of $(B_j)_{j \in S}$, namely, 
\begin{align*}
& \Prob{\sum_{j \in S}B_j\Bigl(\csm_{j}(t; \vec{X})^2 -  \csm_{j}(r; \vec{X})^2\Bigr)\le -\frac{x}{3}}\\
= \; & \sum_{\emptyset \neq J \subset S}\Biggl( \Prob{\sum_{j \in J}B_j \Bigl(\csm_{j}(t; \vec{X})^2 -  \csm_{j}(r; \vec{X})^2\Bigr)\le -\frac{x}{3} \;\Bigg\vert\; B_j = 1 \text{ iff } j \in J}\\
&\qquad\qquad\qquad\qquad\qquad\qquad  \times\Prob{B_j =1 \text{ iff } j \in J}\Biggr)\\
= \; & \sum_{\emptyset \neq J \subset S} \Prob{\sum_{j \in J} \Bigl(\csm_{j}(t; \vec{X})^2 -  \csm_{j}(r; \vec{X})^2\Bigr)\le -\frac{x}{3}}\Prob{B_j =1 \text{ iff } j \in J}.
\end{align*}
For every $\emptyset \neq J \subset S$, we obtain, by \cref{lm:diff}, 
\begin{align*}
 &\Prob{\sum_{j \in J} \Bigl(\csm_{j}(t; \vec{X})^2 -  \csm_{j}(r; \vec{X})^2\Bigr)\le -\frac{x}{3}}\\ 
\le \;& \mathbb{P}\Biggl\{\sum_{j \in J} \biggl(\Bigl(\csm_{j}(t; \vec{X})^2 -  \csm_{j}(r \vec{X})^2\Bigr) \\
&\qquad\qquad \qquad\qquad\qquad-\Bigl(\csm_{j}(t; \vec{F})^2 -  \csm_{j}(r; \vec{F})^2\Bigr) \biggr)\le -\frac{x}{3}\Biggr\}\\
\le\; &  2\exp\left(-\min\set{\frac{x}{48\sqrt{2d(t,r)}},\;\frac{x^2}{576\, d(t,r)(s + n\lambda\norm{\vec{\delta}}^2/16)}}\right).
\end{align*}
Thus, we have 
\begin{align*}
 &\Prob{\sum_{j \in J} \Bigl(\csm_{j}(t; \vec{X})^2 -  \csm_{j}(r; \vec{X})^2\Bigr)\le -\frac{x}{3}}\\ 
\le\; &  \sum_{\emptyset \neq J \subset S}\Biggl(2\exp\biggl(-\min\set{\frac{x}{48\sqrt{2d(t,r)}},\;\frac{x^2}{576\, d(t,r)(s + n\lambda\norm{\vec{\delta}}^2)}}\biggr)\\
&\qquad\qquad\qquad\qquad\qquad\qquad \times \Prob{B_j =1 \text{ iff } j \in J}\Biggr)\\
=\;& 2\exp\left(-\min\set{\frac{x}{48\sqrt{2d(t,r)}},\;\frac{x^2}{576\,d(t,r)\bigl(s + n\lambda\norm{\vec{\delta}}^2\bigr)}}\right).
\end{align*}

\item 
For \eqref{e:osp3}, we first note that for $j \in N$,
\begin{align*}
\Prob{B_j =1} & \le \Prob{\abs{\csm_{j}(t; \tilde{\vec{\Xi}})} \ge \thd} +  \Prob{\abs{\csm_{j}(r; \tilde{\vec{\Xi}})} \ge \thd} \\
& = 2\bigl(1-\Phi(\thd)\bigr) \le \exp\left(-\frac{\thd^2}{2}\right).
\end{align*}
Then, \cref{lm:difb}\ref{i:difb1} implies
\begin{align*}
&\Prob{\sum_{j \in N }B_j\Bigl(\csm_{j}(t; \vec{X})^2 -  \csm_{j}(r; \vec{X})^2\Bigr)\le -\frac{x}{3}} \\
= \; & \Prob{\sum_{j \in N }B_j\Bigl(\csm_{j}(t; \vec{\Xi})^2 -  \csm_{j}(r; \vec{\Xi})^2\Bigr)\le -\frac{x}{3}} \\
\le\; & 2\exp\left(-\min\set{\frac{x}{48\sqrt{2d(t,r)}},\;\frac{x^2}{864\,d(t,r)\sum_{j \in N} \Prob{B_j = 1} }}\right)\\
\le\;& 2\exp\left(-\min\set{\frac{x}{48\sqrt{2d(t,r)}},\;\frac{x^2}{864\, d(t,r)p\exp(-\thd^2/2)}}\right).
\end{align*}
\end{enumerate}

Thus, we combine all bounds above and obtain, for $(t,r) \in \PP$,
\begin{align*}
&\Prob{\cmp(t, r) \le 0} \\
\le\;& 2\exp\left(-\frac{x^2}{864\, d(t,r)p\exp(-\thd^2/2)}\right) + \Bigl(4 + 2\exp\bigl(se^{-n\lambda\norm{\vec{\delta}}^2/(768\, s)}\ind\set{\alpha >0}\bigr)\Bigr)\\
&\qquad \times \exp\left(-\min\set{\frac{x}{288\sqrt{2d(t,r)}},\; \frac{x^2}{576\,d(t,r)\bigl(s + 2n\lambda\norm{\vec{\delta}}^2\bigr)}}\right)\\
\le\;& 2\exp\left(-\frac{n\lambda\abs{t-r}\norm{\vec{\delta}}^4}{2^{17}3^3 p\exp(-\thd^2/2)}\right) + \Bigl(4 + 2\exp\bigl(se^{-n\lambda\norm{\vec{\delta}}^2/(768\, s)}\ind\set{\alpha >0}\bigr)\Bigr)\\
&\qquad \times \exp\left(-\min\set{\frac{\sqrt{n\lambda\abs{t-r}}\norm{\vec{\delta}}^2}{2^{11}3^2\sqrt{2}},\; \frac{n\lambda\abs{t-r}\norm{\vec{\delta}}^4}{2^{18}3^2\bigl(s + 2n\lambda\norm{\vec{\delta}}^2\bigr)}}\right),
\end{align*}
where the last inequality is due to \eqref{e:dbnd} and \eqref{e:chx2}. 

We claim that there are at most $O\bigl((\log n)^2\bigr)$ possible choices for such pairs of probe points $r,t$. More precisely, we will show that there exists another set $\PP_*$ such that $\PP\subseteq \PP_*$ with $\abs{\PP_*} = O\bigl((\log n)^2\bigr)$, and $\PP_*$ is completely determined by the signal $\mat{F}$, in particular, independent of the noise $\mat{\Xi}$.  We will prove the claim in two steps.

First, we note that in each step of (naive) optimistic search, at least $1/4$ of the interval is dropped off.  Thus, the optimistic search stops with at most 
$$
\frac{\log n}{\log(4/3)} \le 4 \log n \quad \text{steps.}
$$
However, in some steps (e.g.\ two probe points lying on both side of $n\lambda$ and with the same or similar distance to $n\lambda$), which part to drop off may depend not only on the signal $\mat{F}$ but also on the noise $\mat{\Xi}$. In this case, we will include both pairs of probe points that are followed by deleting the leftmost part and the rightmost part, respectively, as possible choices of probe points. This results in a bigger set of pairs of probe points than $\PP$, which is denoted by $\PP_*$.

\begin{figure}[ht]
\centering
\includegraphics[width=0.8\textwidth]{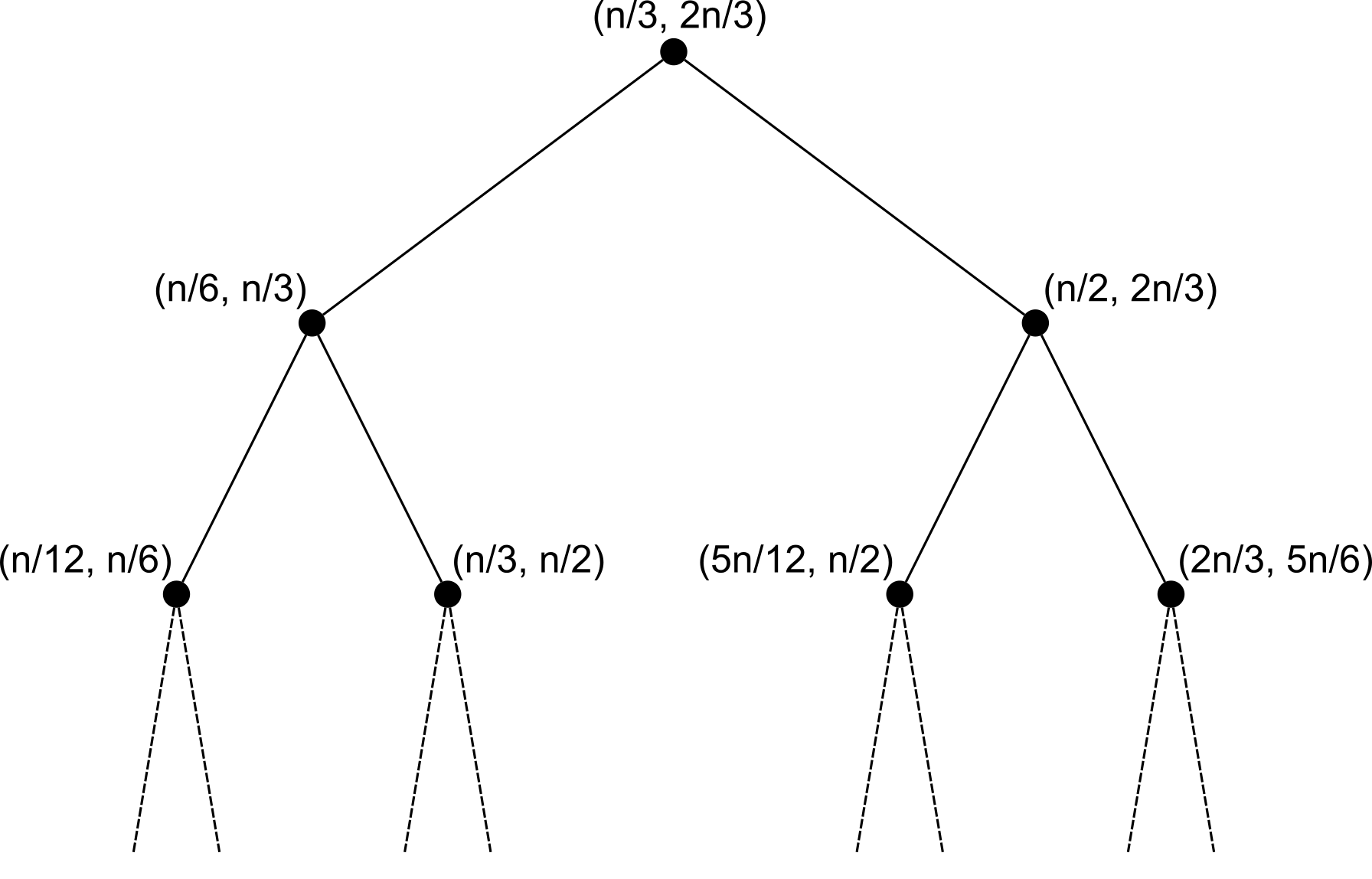} 
\caption{Binary tree structure of (naive) optimistic searches. The starting search interval is $(0,n]$. Each node represents a pair of probe points $(t,r)$. The level of depth corresponds to the number of steps in optimistic searches. }
\label{f:osbt}
\end{figure}

Second, we record all possible choices of pairs of probe points by a binary tree, see \cref{f:osbt} for an illustration.  Let us look into the binary tree: at each level (which consists of all nodes that have the same depth to the root), there is at most one node for which the corresponding pair of probe points have similar distance to $n\lambda$. Namely, given the signal $\mat{F}$, at each level, there is at most one node that yields two child nodes, while for the rest it is clear that which part needs dropped off. In this way, the possible choices of pairs of probe points can be restricted to a subtree that has at most $4 \log n$ width, and thus has at most 
$$
(4 \log n)^2 \quad \text{nodes.}
$$
That is, $\abs{\PP_*} \le (4 \log n)^2$, which shows the above claim. 

Thus, using the union bound, we can bound the probability of wrongly excluding any segment for all pairs of probe points $r,t$ with $\abs{t-r}\ge\eps$, for some $\eps \equiv \eps({n,p,s}) > 0$, as follows:   
\begin{align*}
&\Prob{\cmp(t, r) \le 0\text{ for some $(t,r) \in \PP$ such that $\abs{t-r}\ge\eps$}}\\
\le\; &\Prob{\cmp(t, r) \le 0\text{ for some $(t,r) \in \PP_*$ such that $\abs{t-r}\ge\eps$}}\\
\le\;& \sum_{(t,r) \in \PP_*,\; \abs{t-r}\ge\eps}\Prob{\cmp(t, r) \le 0}\\
\le\;& \sum_{(t,r) \in \PP_*,\; \abs{t-r}\ge\eps} \Biggl( 2\exp\biggl(-\frac{n\lambda\abs{t-r}\norm{\vec{\delta}}^4}{2^{17}3^3  p\exp(-\thd^2/2)}\biggr) \\
& \qquad\qquad\qquad+ \Bigl(4 + 2\exp\bigl(se^{-n\lambda\norm{\vec{\delta}}^2/(768\, s)}\ind\set{\alpha >0}\bigr)\Bigr)\\
&\qquad\times\exp\biggl(-\min\Bigl\{\frac{\sqrt{n\lambda\abs{t-r}}\norm{\vec{\delta}}^2}{2^{11}3^2\sqrt{2}},\; \frac{n\lambda\abs{t-r}\norm{\vec{\delta}}^4}{2^{18}3^2\bigl(s + 2n\lambda\norm{\vec{\delta}}^2\bigr)}\Bigr\}\biggr)\Biggr)\\
\le\; & (4\log n)^2 \Biggl( 2\exp\biggl(-\frac{n\lambda\eps \norm{\vec{\delta}}^4}{2^{17}3^3  p\exp(-\thd^2/2)}\biggr) \\
& \qquad\qquad\qquad+ \Bigl(4 + 2\exp\bigl(se^{-n\lambda\norm{\vec{\delta}}^2/(768\, s)}\ind\set{\alpha >0}\bigr)\Bigr)\\
&\qquad\qquad\qquad\times\exp\biggl(-\min\Bigl\{\frac{\sqrt{n\lambda\eps}\norm{\vec{\delta}}^2}{2^{11}3^2\sqrt{2}},\; \frac{n\lambda\eps\norm{\vec{\delta}}^4}{2^{18}3^2\bigl(s + 2n\lambda\norm{\vec{\delta}}^2\bigr)}\Bigr\}\biggr)\Biggr).
\end{align*}
Thus, under assumption \eqref{e:db1cp} with $C_0$ in \eqref{e:chC0}, if 
\begin{equation}\label{e:epstar}
\eps \ge \eps_* := \frac{C_1}{2} \max\Bigl\{\frac{ \log\log n}{\norm{\vec{\delta}}^2},\, \frac{\min\{s^2, p \log\log n\}}{n\lambda \norm{\vec{\delta}}^4}\Bigr\}, 
\end{equation}
with $C_1$ in \eqref{e:chC1}, then by $\max_{x\ge 0}xe^{-x} =  e^{-1}$ and \eqref{e:pthd} we obtain
$$
\Prob{\cmp(t, r) \le 0\text{ for some $(t,r) \in \PP$ such that $\abs{t-r}\ge\eps$}} \le \frac{128}{\log\log n} \to 0.
$$
That is, when the probe points $t,r$ satisfy $\abs{t-r} \ge \eps_*$, no wrong dropping-off will occur, except on an event with asymptotically vanishing probability. As a consequence, it holds with probability tending to one that $\abs{\hat{\cpt} - \lambda} \le 2 \eps_*/n$, which concludes the proof. 
\end{proof}

\begin{remark}[Rates of probability towards one]
In fact, we have shown in the proof above that
$$
1-\Prob{\abs{\hat\cpt - \cpt} \le C_1 \max\Bigl\{\frac{ \log\log n}{n\norm{\vec{\delta}}^2 },\, \frac{\min\{s^2, p \log\log n\}}{n^2\lambda \norm{\vec{\delta}}^4}\Bigr\}} \le \frac{140}{\log n}.
$$
Such a rate can be improved to any polynomials of $1/\log n$ if constant $C_0$ is chosen larger than the one in \eqref{e:chC0}.
\end{remark}

\begin{remark}[Slightly weaker assumption]
If the dyadic grid search in advanced optimistic search is modified to start from $k = 1$ (or equivalently, $t_1 = n/2$) on, and to stop at the first $k_0$ such that $\cmp(t_{k_0}, t_{k_0+1}) > 0$ and $\cmp(t_{k_0}, t_{k_0-1}) > 0$, then we only need to control random perturbations at pairs of dyadic points of maximal number $O\bigl(\log(2/\lambda) \bigr)$. This, together with slight modification of the proof, will allow a weaker condition
$$
n\lambda \norm{\vec{\delta}}^2 \gtrsim \begin{cases}
\sqrt{p \log\log (2/\lambda)} & \text{if } s \ge \sqrt{p \log\log (2/\lambda)},\\
\max\left(s\log\frac{e \sqrt{p \log\log (2/\lambda)}}{s},\log\log n\right)& \text{if } s \le \sqrt{p \log\log (2/\lambda)}.
\end{cases}
$$
A similar condition appeared also in \citet{PiCV20}.  
\end{remark}

\begin{remark}[Slightly higher accuracy]
Note that in the proof we only need to control the random perturbations at the pairs of probe points until the length of search interval is  $\eps$, which leads to an upper bound $O\bigl(\log(n / \eps)^2\bigr)$ on the number of pairs of probe points. This will lead to a slightly better accuracy of order
$$
\min\set{\gamma > 0 \; : \; \gamma \ge\frac{ \log\log (n/\gamma)}{\norm{\vec{\delta}}^2 } \text{ and } \gamma \ge \frac{\min\{s^2, p \log\log (n/\gamma)\}}{n\lambda \norm{\vec{\delta}}^4}}.
$$
Further, the iterated log factor might be removable if one employs the sub-martingale property of the gain function and partitions the sampling locations into geometrically equally spaced segments, as  detailed in \citet{HRMS21}. This might lead to the essential multiplicity of $\mathrm{polylog} (\eps)$ instead of order $(\log n)^2$, and then it would allow the removability of the iterated log factor. The careful examination will be left as part of future research. 
\end{remark}

\begin{remark}[Combined optimistic search]\label{r:cos}
We note that \cref{th:mh1cp} also holds for the combined optimistic search, introduced in \cref{Appendix_cOS}. This can be proven as follows: Let $\hat{t}_a$ and $\hat{t}_c$ be the outputs of advanced and combined optimistic search, respectively. Similar to the proof above, one can apply \cref{lm:diff,lm:difb} and show that
\begin{multline*}
\mathbb{P}\Bigl\{\cmp(\hat{t}_a, r) \le 0 \text{ for some probe point $r$ in naive optimistic search}\\ \text{such that } \abs{r - n\lambda} \ge 6\eps_*\Bigr\} \to 0,
\end{multline*}
with $\eps_*$ defined in \eqref{e:epstar}. Since $\cmp(\hat{t}_a, \hat{t}_c) \le 0$, it holds that 
$$
\Prob{\abs{\hat{t}_c - n\lambda} \le 6\eps_*} \to 1. 
$$
\end{remark}

\subsection{Multiple change points}\label{ss:mcpt}
Note that \cref{th:osbs}\ref{i:osbs1} is a special case of \cref{th:mhmcp} for $p=1$, so the proof of \cref{th:osbs}\ref{i:osbs1} is omitted. 

The intuition why the OSeedBS in combination with the NOT selection performs optimally is that the selected seeded interval often contains only a single change point close to the center, due to the multiscale nature of seeded intervals. On such a seeded interval, naive and advanced (as well as combined) optimistic searches perform in a minimax optimal way. 

\begin{proof}[Proof of \cref{th:mhmcp}] 
We set the constant $C_0$ in assumption~\eqref{e:dbmcp} as
$$
C_0 = 2^{34} 3^{4},
$$
and the selection threshold $\gamma = C_1 \rho(n,p,s)$ with 
$$
C_1 =2^{33} 3^{4}. 
$$
We will show that the assertion of the theorem holds with 
$$
C_2 =2^{31} 3^{4}.
$$
For notation simplicity, we assume there is no rounding, the decay $a = 1/2$ for seeded intervals, and the step size $\nu = 1/2$ for optimistic searches (otherwise, only the multiplying constants may be different). We split the rest of the proof into five steps. 

\emph{Step 1.} 
Consider intervals $\bigl((\cpt_i - \lambda/2) n,\; (\cpt_i + \lambda/2)n\bigr]$ as potential backgrounds for change points $\cpt_i n$ for $i= 1,\ldots,\ncp$\,. By the construction of seeded intervals, we can find seeded intervals $(c_i -r_i, c_i + r_i]$ such that 
\begin{align*}
(c_i - r_i,\, c_i + r_i] &\subseteq \Bigl((\cpt_i -\frac{1}{2}\lambda)n,\, (\cpt_i + \frac{1}{2} \lambda)n\Bigr], \\
r_i &\ge \frac{1}{6} \lambda n, \\
\text{and}\qquad \abs{c_i - \cpt_i n} &\le \frac{1}{2}r_i\,.
\end{align*}
Note that $\abs{(c_i - r_i,\, c_i + r_i]} \ge \lambda n/3$ and that $(c_i - r_i,\, c_i + r_i]$ contains only a single change point $\cpt_i n$ for every $i = 1,\ldots, \ncp$\,.

\emph{Step 2.} 
By $\hat \cpt_i^0$ we denote the estimated change point (scaled by $1/n$) by naive or advanced (or combined) optimistic search on $(c_i - r_i, \; c_i + r_i]$. Note that there is a single change point in $(c_i - r_i,\,c_i + r_i]$, which is closer to the center of the interval than boundaries, for every $i=1,\ldots,\ncp$. Following similar lines as in the proof of \cref{th:mh1cp}, under the assumption in \eqref{e:dbmcp}, we can show that
$$
\Prob{\abs{\hat \cpt_i^0 - \cpt_i} \le \varepsilon_{i}^0 \equiv C_2^0\left(\frac{\log n}{n\delta_i^2}\,\bigvee\,\frac{s^2 \,\wedge\, p\log n}{n^2\lambda\delta_i^4} \right),\; i = 1,\ldots, \ncp} \ge 1- \frac{9}{n},
$$
with $C_2^0 := 2^{22}3^3$. One important difference in calculation to \cref{ss:scp} is the overall number of events that we need to control is now in polynomial of $n$ rather than in polynomial of $\log n$.
If $\abs{\hat \cpt_i^0 - \cpt_i} \le \varepsilon_{i}^0$, then it holds, for every $j \in \set{1,\ldots, p}$,  
\begin{align*}
\csm_{(c_i - r_i,\,c_i + r_i], j}(\hat \cpt_i^0 n; \mat{F})^2 \ge \frac{1}{16}n\lambda\delta_{i,j}^2 - \frac{3}{3 - 4\eps_i^0}n\eps_i^0\delta_{i,j}^2\ge \frac{1}{18}n\lambda\delta_{i,j}^2,
\end{align*}
where the last inequality is due to $\eps_i^0 \le \lambda /432$. Thus, we obtain: 
\begin{enumerate}[i.]
\item Dense case, i.e.\ when $s \ge \sqrt{p\log n}$ or $s = p$. The union bound and \cref{lm:tail.chi2} imply 
$$
\Prob{G_{(c_i-r_i,\, c_i+r_i]}(\hat \tau_i^0 n)\ge \gamma, \, i=1,\ldots, \kappa} \ge 1 - \frac{10}{n}.
$$ 
\item Sparse case, i.e.\ when $s < \sqrt{p\log n}$ and $s \neq p$. Similar to Part~1 in the proof of \cref{th:mh1cp}, we split the set of coordinates into the set of large jumps, the set of small jumps and the set of no jumps, apply \cref{lm:bwcs,lm:bwcs2}, and then obtain 
$$
\Prob{G_{(c_i-r_i,\, c_i+r_i]}(\hat \tau_i^0 n)\ge \gamma, \, i=1,\ldots, \kappa} \ge 1 - \frac{12}{n}.
$$
\end{enumerate}
That is, for every $i \in\set{1\ldots, \ncp}$, $G_{(c_i-r_i,\, c_i+r_i]}(\hat \tau_i^0 n)$ is above the selection threshold $\gamma$, with probability tending to one.

\emph{Step 3.}
Recall that the NOT selection rule selects the shortest seeded interval $I_*$ which has a value of gain function  above the threshold $\gamma$. Then it follows from Step~2, with probability tending to one, that the length of selected seeded interval $\abs{I_*} \le \lambda$, and thus that $I_*$ contains at most one change point. By \cref{lm:tail.chi2,lm:bwcs} and the union bound, we consider dense and sparse cases separately, and then obtain 
$$
\Prob{\max_{I\subseteq (0,n],\, I \text{ contains no change point}}\max_{t\in I} G_{(l,r]}(t) \le 24\,\rho(n,p,s)} \ge 1- \frac{1}{6 n}.
$$
Since $\gamma \ge 24\,\rho(n,p,s)$, the selected seeded interval $I_*$ contains exactly one change point, with probability tending to one. 

\emph{Step 4.}
We rewrite the selected seeded interval $I_*$ as $\bigl((\cpt_i - u) n, \,(\cpt_i + v) n \bigr]$ for some $i \in \{1,\ldots, \ncp\}$. 
We claim that, with probability tending to one, 
\begin{equation}\label{eq:snr}
\frac{uv}{u+v}n\delta_i^2 \ge \frac{\gamma}{2}.
\end{equation}
Otherwise, we have $\max_{t\in I_*}\sum_{j=1}^p \csm_{I_*, j}(t:\mat{F})^2 \le \gamma/2$, and then, by the union bound, \cref{lm:tail.chi2} and \cref{r:utb}, obtain (noting that $C_1 \ge 36$) 
\begin{multline*}
\Prob{\max_{t\in I_*} G_{I_*}(t) \le \gamma} \ge 1 -\sum_{t\in I_*} \Prob{G_{I_*}(t) \ge \gamma} \\
\ge  1 -\sum_{t\in I_*} \Prob{G_{I_*}(t) - \sum_{j=1}^p \csm_{I_*, j}(t:\mat{F})^2 \ge \frac{\gamma}{2}}\ge 1 - \frac{1}{n^2},
\end{multline*}
which contradicts with the fact that $I_*$ is selected by NOT. 

Let $\hat \cpt_i^1$ be the estimated change point (scaled by $1/n$) by naive or advanced (or combined) optimistic search applied to $I_*$\,. Then, following similar calculations as in the proof of \cref{th:mh1cp}, we can see that \eqref{eq:snr} implies 
\begin{equation*}
\Prob{\abs{\hat \cpt_i^1 - \cpt_i} \le \varepsilon_{i}^1 \equiv C_2^1\left(\frac{\log n}{n\delta_i^2}\,\bigvee\,\frac{s^2 \,\wedge\, p\log n}{n\gamma\delta_i^2} \right),\; i = 1,\ldots, \ncp} \ge 1 - \frac{11}{n},
\end{equation*}
with $C_2^1 := 2^{32} 3^4$. In particular, it holds that $\varepsilon_{i}^1 \le \lambda/4$.

\emph{Step 5.}
Steps 1--4 imply that all change points $\{\cpt_1, \ldots, \cpt_{\ncp}\}$ can be estimated by $\{\hat \cpt_1^1, \ldots, \hat \cpt_{\ncp}^1\}$ with errors $\varepsilon_{i}^1$ and in particular $\hat \ncp \ge \ncp$, with probability tending to one. It remains to show that $\hat \ncp = \ncp$\,. To this end, we note that after the detection of $\ncp$ change points, the remaining seeded intervals either contain no change point or contain one or two change points that are very close to the boundary. In the latter case, if a seeded interval contains say $\cpt_i$ then we must have the distance of $\cpt_i$ to one of the boundaries is no more than $\varepsilon_{i}^1$. This leads to 
$$
\max_{t\in I} \sum_{j =1}^p\csm(t; \mat{F})^2 \le \max_{1 \le i \le \ncp} n\eps_i^1\delta_i^2 \le \frac{1}{2}\gamma
$$
for any remaining interval $I$. Thus, as in Step~4, we can show that the maximum value of gain function $\max_{t\in I}{G_{I}(t)}$ on every remaining seeded interval $I$ is upper bounded by $\gamma$, more precisely,
$$
\Prob{\max_{I \in \mathcal{I}_r}\max_{t\in I} G_{(l,r]}(t) \le \gamma} \ge 1- \frac{12}{ n}.
$$
where $\mathcal{I}_r$ denotes the collection of remaining seeded intervals. This further implies that  OSeedBS will stop after $\ncp$ steps, i.e.\ $\hat\ncp = \ncp$.

\emph{Step 6.} Recall from Step~4 that $\eps_i^1 \le \lambda/4$. Then, with probability tending to one, it holds that each interval $I_i:= \bigl(n(\hat\cpt_{i-1}^1 + \hat\cpt_{i}^1)/2,\ n(\hat\cpt_{i}^1 + \hat\cpt_{i+1}^1)/2\bigr]$ contains only one change point $n \cpt_i$ such that $n\cpt_i$ is at least $n\lambda/4$ apart from the boundaries of $I_i$. Thus, as in the proof of \cref{th:mh1cp}, under the assumption in \eqref{e:dbmcp}, we can establish for naive or advanced (or combined) optimistic search the following 
$$
\Prob{\hat \ncp = \ncp,\,\, \abs{\hat \cpt_i - \cpt_i} \le C_2\left(\frac{ \log n}{\delta_i^2 n}\,\bigvee\, \frac{s^2 \wedge p \log n}{n^2\lambda \delta_i^4}\right), \, i =1, \ldots, \ncp} \ge 1-\frac{24}{n},
$$
which concludes the proof.
\end{proof}

\section{Computation complexity analysis}\label{App:Computation}
In this section, we derive computation complexity analysis for optimistic searches, as stated in \cref{Lem:cOS_speed}, \cref{th:osbs}\ref{i:osbs2} and \cref{p:cmp}. We omit the proof of \cref{th:osbs}\ref{i:osbs2} as it follows from \cref{p:cmp} with $p=1$.

\begin{proof}[Proof of \cref{Lem:cOS_speed}]
We start with the naive optimistic search. Let $I_k$ denote the search window at the $k$-th step of the naive optimistic search. The choice of probe points ensures that a segment of length at least $\nu \abs{I_k}/2\, \wedge\,  \nu  \abs{I_{k-1}}/2$ will be dropped out, with $\nu$ the step size, at the $k$-th step. Since $I_{k}\subseteq I_{k-1}$, we have a proportion of at least $\nu/2$ will be removed at each step, and thus the procedure will stop in 
$$
\left\lceil \frac{1}{\log \bigl(2/(2-\nu)\bigr)}\log(R-L)\right\rceil
\quad \text{steps.}
$$

Consider next the search over dyadic locations (in the advanced optimistic search), which requires $O\bigl(\log (R-L)\bigr)$ evaluations of gain functions, as there are in total $2\lceil\log_2(R-L)\rceil$ dyadic locations. 

Thus, overall, in the worst case one has $O(\log (R-L)) \le O(\log n)$ evaluations for the naive and the advanced optimistic search.
\end{proof}

\begin{proof}[Proof of \cref{p:cmp}]
\emph{Part \ref{i:cmp1}.}
It follows from \cref{Lem:cOS_speed} and the fact that given the cumulative sums each evaluation of the gain or the comparison function takes $O(p)$ computations.

\emph{Part \ref{i:cmp2}.}
We consider first the optimistic searches on seeded intervals. By $\mathcal{I}_{\omega}$ we denote the set of all seeded intervals of length no less than $m \asymp n^{1-\omega}$.  Then the total number of evaluations in naive and advanced (and also combined) optimistic search on all seeded intervals in $\mathcal{I}_{\omega}$ is of order
\begin{align*}
p \sum_{I \in \mathcal{I}_w} \log \abs{I} \; &\lesssim\;  p \sum_{k = 1}^{\ceil{\log_{{1}/{a}}(n^\omega)}}a^{-(k-1)} \log(a^{k-1}n)\\
&\lesssim\;  p \int_0^{\log_{{1}/{a}}(n^\omega)}a^{-x} \log(a^x n) dx\\
&\lesssim \; p \int_{0}^{n^\omega} \log\left(\frac{n}{x}\right) dx = n^\omega \bigl(1 + (1-\omega)\log n \bigr) \\
&\lesssim\;  p \min\{n^{\omega}\log n,\, n\}\,.
\end{align*}

We examine next the NOT selection, and employ the deterministic and nested structure of the seeded intervals $\mathcal{I_{\omega}}$. Each seeded interval $I$ is equipped with a boolean variable $b_I$, denoting whether it is needed to be considered or not (its default value being true), and with an integer variable $t_I$, recording the estimated change point location (its default value being zero).    

Note that the seeded intervals at  $k$-th layer are of the same length $l_k = n a^{k-1}$, and that the NOT selection favours shorter intervals. Thus, we iterate from the highest layer to the lowest layer. More precisely, at $k$-th layer, we check each seeded interval $I$: 
\begin{itemize}
\item If its boolean variable $b_I$ is false, we mark the seeded intervals that lie at $(k-1)$-th layer and contain $t_I$, with the boolean variable being false and the integer variable being $t_I$. Due to the construction of seeded intervals, this can be done in $O(1)$ computations. 
\item If its boolean variable $b_I$ is true, we check whether the candidate change point $\hat{t}_I$ found in this seeded interval has a value of gain function above the selection threshold. If no, we do nothing. If yes, we identify $\hat{t}_I$ as an estimated change point, and mark the seeded intervals that lie at $(k-1)$-th layer and contain $\hat{t}_I$, with the boolean variable being false and the integer variable being $\hat{t}_I$. Similarly,  this can be done in $O(1)$ computations. 
\end{itemize}
It is clear to see that this procedure executes precisely the NOT selection. As each seeded interval is accessed at most $O(1)$ times and each access costs $O(1)$ computations, the NOT selection requires $O(\abs{\mathcal{I_{\omega}}}) = O(n^\omega)$ computations. Note that this computation complexity analysis improves the one in Theorem~1 of \citet{SeedBS} by a log factor. 

The post-processing step involves optimistic searches on $\hat\ncp$ non-overlapping intervals, and thus requires another 
$$
O\bigl(p \min\set{\hat\ncp\log n,\; n}\bigr) = O\bigl(p \min\{n^{\omega}\log n,\, n\}\bigr)
$$
computations.  

Thus, if cumulative sums are pre-computed, the final computation complexity is $O\bigl(p \min\{n^{\omega}\log n,\, n\}\bigr)$; otherwise, the computation of cumulative sums dominates, and leads to the final computation complexity of $O(pn)$. Besides, the overall memory complexity is $O(pn)$. 
\end{proof}

\begin{remark}[Implementation]
It is logically clean and easy to consider the optimistic search part and the selection part, separately, as in \cref{Alg:OSeedBS}. But in regard to implementation, we should merge both parts, and run optimistic searches only for the seeded intervals with boolean variables being true, though this only improves the multiplying constant in the computation complexity analysis, instead of the order. 
\end{remark}

\begin{remark}[Speed-up]
We can speed up OSeedBS by considering first the middle points of seeded intervals as follows:
\begin{enumerate}[i.]
\item  For a given seeded interval, we run optimistic searches only when the value of gain function at the middle of this seeded interval is above the selection threshold. 
\item Further, we use the middle point of seeded intervals to decide which coordinates to look at for optimistic searches, more precisely, all the coordinates that has a cumulative statistic at the middle point with absolute value greater than $\thd$. 
\end{enumerate}
Note that for each change point there is a seeded interval containing it roughly at the center, which will eventually be picked by the NOT selection, as shown in the proof of \cref{th:mhmcp}. Then the gain function evaluated at the middle points of such seeded intervals has similar values as if it were evaluated at true change points. Thus, the same statistical guarantee remains still valid for OSeedBS with the modification above. 

Let $\hat{s}$ denote the maximal number of selected coordinates. Then as $n\to\infty$ it holds that $\Prob{\hat{s} = O(s)} \to 1$. It is clear to see that the overall computation complexity can be improved to $O(\hat{s} \min\set{n^{\omega}\log n, \, n} + p n^{\omega})$, which is $O({s} \min\set{n^{\omega}\log n, \, n} + p n^{\omega})$ with probability tending to one. 

In particular, in case of a single change point, one only needs to consider seeded intervals that starting at $0$ or ending at $n$, the number of which is $O(\log n)$. Then the naive optimistic search with the NOT selection can achieve the same statistical optimality as in \cref{th:mh1cp}, and together with the speed-up idea above it results in $O(p \log n)$ computations, provided that cumulative sums of data are available. Thus, the naive optimistic search with the NOT selection, and the advanced optimistic search, are equivalent in terms of statistical and computational efficiency. 
\end{remark}
 

%% file: Arxiv_main.bbl
\begin{thebibliography}{}

\bibitem[Avanesov and Buzun, 2018]{Avanesov_theory}
Avanesov, V. and Buzun, N. (2018).
\newblock Change-point detection in high-dimensional covariance structure.
\newblock {\em Electron. J. Stat.}, 12(2):3254--3294.

\bibitem[Avriel and Wilde, 1966]{Avriel1966}
Avriel, M. and Wilde, D.~J. (1966).
\newblock Optimality proof for the symmetric {F}ibonacci search technique.
\newblock {\em Fibonacci Quart.}, 4:265--269.

\bibitem[Avriel and Wilde, 1968]{Avriel1968}
Avriel, M. and Wilde, D.~J. (1968).
\newblock Golden block search for the maximum of unimodal functions.
\newblock {\em Management Science}, 14(5):307--319.

\bibitem[Bai and Perron, 1998]{BaPe98}
Bai, J. and Perron, P. (1998).
\newblock Estimating and testing linear models with multiple structural
  changes.
\newblock {\em Econometrica}, 66(1):47--78.

\bibitem[Baranowski et~al., 2019]{Baranowski}
Baranowski, R., Chen, Y., and Fryzlewicz, P. (2019).
\newblock Narrowest-over-threshold detection of multiple change points and
  change-point-like features.
\newblock {\em J. R. Stat. Soc. Ser. B. Stat. Methodol.}, 81(3):649--672.

\bibitem[Bhattacharjee et~al., 2017]{BhBM17}
Bhattacharjee, M., Banerjee, M., and Michailidis, G. (2017).
\newblock Common change point estimation in panel data from the least squares
  and maximum likelihood viewpoints.
\newblock {\em arXiv preprint arXiv:1708.05836}.

\bibitem[Boysen et~al., 2009]{BKLMW09}
Boysen, L., Kempe, A., Liebscher, V., Munk, A., and Wittich, O. (2009).
\newblock Consistencies and rates of convergence of jump-penalized least
  squares estimators.
\newblock {\em Ann. Statist.}, 37(1):157--183.

\bibitem[Bybee and Atchad\'{e}, 2018]{Bybee_Atchade_JMLR}
Bybee, L. and Atchad\'{e}, Y. (2018).
\newblock Change-point computation for large graphical models: a scalable
  algorithm for {G}aussian graphical models with change-points.
\newblock {\em J. Mach. Learn. Res.}, 19:Paper No. 11, 38.

\bibitem[Cho and Kirch, 2022]{ChKi19}
Cho, H. and Kirch, C. (2022).
\newblock Two-stage data segmentation permitting multiscale change points,
  heavy tails and dependence.
\newblock {\em Ann. Inst. Statist. Math.}, 74(4):653--684.

\bibitem[Davies et~al., 2012]{DaHo12}
Davies, L., H\"{o}henrieder, C., and Kr\"{a}mer, W. (2012).
\newblock Recursive computation of piecewise constant volatilities.
\newblock {\em Comput. Statist. Data Anal.}, 56(11):3623--3631.

\bibitem[Dette et~al., 2022]{DePY22}
Dette, H., Pan, G., and Yang, Q. (2022).
\newblock Estimating a {C}hange {P}oint in a {S}equence of {V}ery
  {H}igh-{D}imensional {C}ovariance {M}atrices.
\newblock {\em J. Amer. Statist. Assoc.}, 117(537):444--454.

\bibitem[Donoho and Johnstone, 1994]{DoJo94}
Donoho, D.~L. and Johnstone, I.~M. (1994).
\newblock Ideal spatial adaptation by wavelet shrinkage.
\newblock {\em Biometrika}, 81(3):425--455.

\bibitem[Fan et~al., 2009]{fan2009_chain_network}
Fan, J., Feng, Y., and Wu, Y. (2009).
\newblock Network exploration via the adaptive lasso and {SCAD} penalties.
\newblock {\em Ann. Appl. Stat.}, 3(2):521--541.

\bibitem[Frick et~al., 2014]{Frick_etal}
Frick, K., Munk, A., and Sieling, H. (2014).
\newblock Multiscale change point inference.
\newblock {\em J. R. Stat. Soc. Ser. B. Stat. Methodol.}, 76(3):495--580.

\bibitem[Friedman et~al., 2008]{glasso}
Friedman, J., Hastie, T., and Tibshirani, R. (2008).
\newblock Sparse inverse covariance estimation with the graphical {Lasso}.
\newblock {\em Biostatistics}, 9(3):432--441.

\bibitem[Friedrich et~al., 2008]{FKLW08}
Friedrich, F., Kempe, A., Liebscher, V., and Winkler, G. (2008).
\newblock Complexity penalized {$M$}-estimation: fast computation.
\newblock {\em J. Comput. Graph. Statist.}, 17(1):201--224.

\bibitem[Fryzlewicz, 2014]{Fryzlewicz_WBS}
Fryzlewicz, P. (2014).
\newblock Wild binary segmentation for multiple change-point detection.
\newblock {\em Ann. Statist.}, 42(6):2243--2281.

\bibitem[Fryzlewicz, 2020]{Fryzlewicz_WBS2}
Fryzlewicz, P. (2020).
\newblock Detecting possibly frequent change-points: {W}ild {B}inary
  {S}egmentation 2 and steepest-drop model selection.
\newblock {\em J. Korean Statist. Soc.}, 49(4):1027--1070.

\bibitem[Gibberd and Nelson, 2017]{GibberdNelson}
Gibberd, A.~J. and Nelson, J. D.~B. (2017).
\newblock Regularized estimation of piecewise constant {G}aussian graphical
  models: the group-fused graphical {L}asso.
\newblock {\em J. Comput. Graph. Statist.}, 26(3):623--634.

\bibitem[{Gibberd} and {Roy}, 2017]{Gibberd_Roy}
{Gibberd}, A.~J. and {Roy}, S. (2017).
\newblock {Multiple changepoint estimation in high-dimensional {Gaussian}
  graphical models}.
\newblock {\em arXiv:1712.05786}.

\bibitem[Hallac et~al., 2017]{Hallac_network}
Hallac, D., Park, Y., Boyd, S., and Leskovec, J. (2017).
\newblock Network inference via the time-varying graphical {Lasso}.
\newblock In {\em Proceedings of the 23rd ACM SIGKDD International Conference
  on Knowledge Discovery and Data Mining}, pages 205--213.

\bibitem[{Hotz} et~al., 2013]{Hotz_ionchannel}
{Hotz}, T., {Sch{\"u}tte}, O.~M., {Sieling}, H., {Polupanow}, T.,
  {Diederichsen}, U., {Steinem}, C., and {Munk}, A. (2013).
\newblock Idealizing ion channel recordings by a jump segmentation
  multiresolution filter.
\newblock {\em IEEE Trans. Nanobioscience}, 12(4):376--386.

\bibitem[Howard et~al., 2021]{HRMS21}
Howard, S.~R., Ramdas, A., McAuliffe, J., and Sekhon, J. (2021).
\newblock Time-uniform, nonparametric, nonasymptotic confidence sequences.
\newblock {\em Ann. Statist.}, 49(2):1055--1080.

\bibitem[Hu et~al., 2021]{Chan_Chen2017}
Hu, S., Huang, J., Chen, H., and Chan, H.~P. (2021).
\newblock Likelihood scores for sparse signal and change-point detection.
\newblock {\em arXiv e-prints}, pages arXiv--2105.

\bibitem[Jackson et~al., 2005]{jackson2005algorithm}
Jackson, B., Scargle, J.~D., Barnes, D., Arabhi, S., Alt, A., Gioumousis, P.,
  Gwin, E., Sangtrakulcharoen, P., Tan, L., and Tsai, T.~T. (2005).
\newblock An algorithm for optimal partitioning of data on an interval.
\newblock {\em IEEE Signal Process. Lett.}, 12(2):105--108.

\bibitem[Kaul et~al., 2021]{KFJS21}
Kaul, A., Fotopoulos, S.~B., Jandhyala, V.~K., and Safikhani, A. (2021).
\newblock Inference on the change point under a high dimensional sparse mean
  shift.
\newblock {\em Electron. J. Stat.}, 15(1):71--134.

\bibitem[{Kaul} et~al., 2019]{Kaul_multiple}
{Kaul}, A., {Jandhyala}, V.~K., and {Fotopoulos}, S.~B. (2019).
\newblock {Detection and estimation of parameters in high dimensional multiple
  change point regression models via $\ell_1/\ell_0$ regularization and
  discrete optimization}.
\newblock {\em arXiv:1906.04396}.

\bibitem[Kaul et~al., 2019]{Kaul_2step}
Kaul, A., Jandhyala, V.~K., and Fotopoulos, S.~B. (2019).
\newblock An efficient two step algorithm for high dimensional change point
  regression models without grid search.
\newblock {\em J. Mach. Learn. Res.}, 20:Paper No. 111, 40.

\bibitem[Kiefer, 1953]{Kiefer1953}
Kiefer, J. (1953).
\newblock Sequential minimax search for a maximum.
\newblock {\em Proc. Amer. Math. Soc.}, 4:502--506.

\bibitem[Killick et~al., 2012]{Killick_etal}
Killick, R., Fearnhead, P., and Eckley, I.~A. (2012).
\newblock Optimal detection of changepoints with a linear computational cost.
\newblock {\em J. Amer. Statist. Assoc.}, 107(500):1590--1598.

\bibitem[Kim et~al., 2005]{Kim_finance}
Kim, C.-J., Morley, J.~C., and Nelson, C.~R. (2005).
\newblock The structural break in the equity premium.
\newblock {\em J. Bus. Econom. Statist.}, 23(2):181--191.

\bibitem[{Kov{\'a}cs} et~al., 2022]{SeedBS}
{Kov{\'a}cs}, S., {Li}, H., {B{\"u}hlmann}, P., and {Munk}, A. (2022).
\newblock {Seeded binary segmentation: a general methodology for fast and
  optimal change point detection}.
\newblock {\em Biometrika, to appear}.

\bibitem[Laurent et~al., 2012]{LLM12}
Laurent, B., Loubes, J.-M., and Marteau, C. (2012).
\newblock Non asymptotic minimax rates of testing in signal detection with
  heterogeneous variances.
\newblock {\em Electron. J. Stat.}, 6:91--122.

\bibitem[{Leonardi} and {B\"{u}hlmann}, 2016]{LeonBuhl}
{Leonardi}, F. and {B\"{u}hlmann}, P. (2016).
\newblock {Computationally efficient change point detection for
  high-dimensional regression}.
\newblock {\em arXiv:1601.03704}.

\bibitem[Li et~al., 2016]{LMS16}
Li, H., Munk, A., and Sieling, H. (2016).
\newblock F{DR}-control in multiscale change-point segmentation.
\newblock {\em Electron. J. Stat.}, 10(1):918--959.

\bibitem[{Liu} et~al., 2021]{Liu_Gao_Samworth_dyadic}
{Liu}, H., {Gao}, C., and {Samworth}, R.~J. (2021).
\newblock {Minimax rates in sparse, high-dimensional changepoint detection}.
\newblock {\em Ann. Statist., to appear}.

\bibitem[Londschien et~al., 2021]{Londschien}
Londschien, M., Kov\'{a}cs, S., and B\"{u}hlmann, P. (2021).
\newblock Change-point detection for graphical models in the presence of
  missing values.
\newblock {\em J. Comput. Graph. Statist.}, 30(3):768--779.

\bibitem[{Lu} et~al., 2017]{intelligent_sampling}
{Lu}, Z., {Banerjee}, M., and {Michailidis}, G. (2017).
\newblock {Intelligent sampling for multiple change-points in exceedingly long
  time series with rate guarantees}.
\newblock {\em arXiv:1710.07420}.

\bibitem[Madrid~Padilla et~al.,
  2022]{Padilla_Yu_Wang_Rinaldo_multivariate_nonparametric}
Madrid~Padilla, O.~H., Yu, Y., Wang, D., and Rinaldo, A. (2022).
\newblock Optimal nonparametric multivariate change point detection and
  localization.
\newblock {\em IEEE Trans. Inform. Theory}, 68(3):1922--1944.

\bibitem[Maidstone et~al., 2017]{Maidstone}
Maidstone, R., Hocking, T., Rigaill, G., and Fearnhead, P. (2017).
\newblock On optimal multiple changepoint algorithms for large data.
\newblock {\em Stat. Comput.}, 27(2):519--533.

\bibitem[Mazumder and Hastie, 2012a]{Mazumder2012_connected_components}
Mazumder, R. and Hastie, T. (2012a).
\newblock Exact covariance thresholding into connected components for
  large-scale graphical lasso.
\newblock {\em J. Mach. Learn. Res.}, 13:781--794.

\bibitem[Mazumder and Hastie, 2012b]{Mazumder2012_dpglasso}
Mazumder, R. and Hastie, T. (2012b).
\newblock The graphical lasso: new insights and alternatives.
\newblock {\em Electron. J. Stat.}, 6:2125--2149.

\bibitem[Niu et~al., 2016]{review_Niu}
Niu, Y.~S., Hao, N., and Zhang, H. (2016).
\newblock Multiple change-point detection: a selective overview.
\newblock {\em Statist. Sci.}, 31(4):611--623.

\bibitem[Olshen et~al., 2004]{CBS}
Olshen, A.~B., Venkatraman, E.~S., Lucito, R., and Wigler, M. (2004).
\newblock {Circular binary segmentation for the analysis of array‐based {DNA}
  copy number data}.
\newblock {\em Biostatistics}, 5(4):557--572.

\bibitem[Page, 1954]{Page54}
Page, E.~S. (1954).
\newblock Continuous inspection schemes.
\newblock {\em Biometrika}, 41:100--115.

\bibitem[Petersen and Pedersen, 2012]{matrix_cookbook}
Petersen, K.~B. and Pedersen, M.~S. (2012).
\newblock The matrix cookbook.

\bibitem[Pilliat et~al., 2020]{PiCV20}
Pilliat, E., Carpentier, A., and Verzelen, N. (2020).
\newblock Optimal multiple change-point detection for high-dimensional data.
\newblock {\em arXiv preprint arXiv:2011.07818}.

\bibitem[{Reeves} et~al., 2007]{Reeves_climatology}
{Reeves}, J., {Chen}, J., {Wang}, X.~L., {Lund}, R., and {Lu}, Q.~Q. (2007).
\newblock {A review and comparison of changepoint detection techniques for
  climate data}.
\newblock {\em J. Appl. Meteorol. Climatol.}, 46:900--915.

\bibitem[Roy et~al., 2017]{RoyMichailidis}
Roy, S., Atchad\'{e}, Y., and Michailidis, G. (2017).
\newblock Change point estimation in high dimensional {M}arkov random-field
  models.
\newblock {\em J. R. Stat. Soc. Ser. B. Stat. Methodol.}, 79(4):1187--1206.

\bibitem[Rufibach and Walther, 2010]{Rufibach_Walther}
Rufibach, K. and Walther, G. (2010).
\newblock The block criterion for multiscale inference about a density, with
  applications to other multiscale problems.
\newblock {\em Journal of Computational and Graphical Statistics},
  19(1):175--190.

\bibitem[Tibshirani, 1996]{lasso_original}
Tibshirani, R. (1996).
\newblock Regression shrinkage and selection via the lasso.
\newblock {\em J. Roy. Statist. Soc. Ser. B}, 58(1):267--288.

\bibitem[Truong et~al., 2020]{review_Truong}
Truong, C., Oudre, L., and Vayatis, N. (2020).
\newblock Selective review of offline change point detection methods.
\newblock {\em Signal Process.}, 167:107299.

\bibitem[Vanegas et~al., 2021]{vanegas2020multiscale}
Vanegas, L.~J., Behr, M., and Munk, A. (2021).
\newblock Multiscale quantile segmentation.
\newblock {\em J. Amer. Statist. Assoc., published online}.

\bibitem[Verzelen et~al., 2020]{VFLR20}
Verzelen, N., Fromont, M., Lerasle, M., and Reynaud-Bouret, P. (2020).
\newblock Optimal change-point detection and localization.
\newblock {\em arXiv preprint arXiv:2010.11470}.

\bibitem[Vostrikova, 1981]{Vostrikova}
Vostrikova, L.~Y. (1981).
\newblock Detecting 'disorder' in multidimensional random processes.
\newblock {\em Soviet Mathematics Doklady}, 24:55--59.

\bibitem[Wang et~al., 2021a]{WaYR21}
Wang, D., Yu, Y., and Rinaldo, A. (2021a).
\newblock Optimal change point detection and localization in sparse dynamic
  networks.
\newblock {\em Ann. Statist.}, 49(1):203--232.

\bibitem[Wang et~al., 2021b]{Wang_Yu_Rinaldo_theory}
Wang, D., Yu, Y., and Rinaldo, A. (2021b).
\newblock {Optimal covariance change point localization in high dimensions}.
\newblock {\em Bernoulli}, 27(1):554--575.

\bibitem[Wang et~al., 2021c]{Wang_Willet_regression}
Wang, D., Zhao, Z., Lin, K.~Z., and Willett, R. (2021c).
\newblock Statistically and computationally efficient change point localization
  in regression settings.
\newblock {\em J. Mach. Learn. Res.}, 22:Paper No. [248], 46.

\bibitem[Wang and Samworth, 2018]{Wang_highdim_mean_change}
Wang, T. and Samworth, R.~J. (2018).
\newblock High dimensional change point estimation via sparse projection.
\newblock {\em J. R. Stat. Soc. Ser. B. Stat. Methodol.}, 80(1):57--83.

\bibitem[Witten et~al., 2011]{Witten2011_connected_components}
Witten, D.~M., Friedman, J.~H., and Simon, N. (2011).
\newblock New insights and faster computations for the graphical lasso.
\newblock {\em J. Comput. Graph. Statist.}, 20(4):892--900.

\bibitem[Yu and Chen, 2021]{YuCh21}
Yu, M. and Chen, X. (2021).
\newblock Finite sample change point inference and identification for
  high-dimensional mean vectors.
\newblock {\em J. R. Stat. Soc. Ser. B. Stat. Methodol.}, 83(2):247--270.

\bibitem[Zhang and Siegmund, 2007]{ZhSi07}
Zhang, N.~R. and Siegmund, D.~O. (2007).
\newblock A modified {B}ayes information criterion with applications to the
  analysis of comparative genomic hybridization data.
\newblock {\em Biometrics}, 63(1):22--32.

\end{thebibliography}
